\documentclass[letterpaper]{article} 

\usepackage[]{aaai25}  

\nocopyright 

\usepackage{times}  
\usepackage{helvet}  
\usepackage{courier}  
\usepackage[hyphens]{url}  
\usepackage{graphicx} 
\usepackage{natbib}  
\usepackage{caption} 
\usepackage{microtype}
\usepackage{makecell}
\usepackage{algorithm}
\usepackage{algorithmic}
\usepackage{newfloat}
\usepackage{listings}
\DeclareCaptionStyle{ruled}{labelfont=normalfont,labelsep=colon,strut=off} 
\floatstyle{ruled}
\newfloat{listing}{tb}{lst}{}
\floatname{listing}{Listing}
\title{Weak Strategyproofness in Randomized Social Choice}
\author {
    Felix Brandt\textsuperscript{\rm 1},
    Patrick Lederer\textsuperscript{\rm 2}
}
\affiliations {
    \textsuperscript{\rm 1}Technical University of Munich,
    \textsuperscript{\rm 2}UNSW Sydney \\
    brandtf@in.tum.de, p.lederer@unsw.edu.au 
}

\newcounter{remarkcount}
\newcommand{\remark}{\stepcounter{remarkcount}\paragraph{Remark \theremarkcount.}}

\usepackage{fix-cm}
	\usepackage{xspace}
	\usepackage{amsmath}
	\usepackage{amsthm}
    \usepackage{amsfonts}
    \usepackage{amssymb}

	\usepackage{nicefrac}
	\usepackage{tikz}
 \usetikzlibrary{positioning,arrows,shapes,decorations,calc,fit,arrows.meta} 
	\usepackage{colortbl}
	\usepackage{booktabs}
	\usepackage{cleveref}
	\usepackage{array}
	\usepackage{graphicx}
	\usepackage{xcolor}
	\usepackage{enumitem}
	\usepackage{tabularx}
	\usepackage{thm-restate}
	\usepackage{ragged2e}
    \usepackage{microtype}
    \usepackage{cleveref}
    
    \def\multiset#1#2{\ensuremath{\left(\kern-.3em\left(\genfrac{}{}{0pt}{}{#1}{#2}\right)\kern-.3em\right)}}

	\theoremstyle{definition}

	\theoremstyle{plain}
	\newtheorem{theorem}{Theorem}
	\newtheorem{lemma}{Lemma}
	
	\newtheorem{proposition}{Proposition}

\usepackage[backgroundcolor=orange!15!white,bordercolor=orange!80!black,textsize=small]{todonotes} 
\presetkeys{todonotes}{inline}{}
\renewcommand{\todo}[1]{} 

\newcolumntype{L}[1]{>{\raggedright\let\newline\\\arraybackslash\hspace{0pt}}m{#1}}
\newcolumntype{K}[1]{>{\raggedright\let\newline\\\arraybackslash\hspace{0pt}}p{#1}}

\newcolumntype{C}[1]{>{\centering\let\newline\\\arraybackslash\hspace{0pt}}m{#1}}
\newcolumntype{R}[1]{>{\raggedleft\let\newline\\\arraybackslash\hspace{0pt}}m{#1}}
\renewcommand{\arraystretch}{1.2}

\newenvironment{profile}{\medmuskip=0mu\relax
	\thickmuskip=1mu\relax
	\tabular}{\endtabular\smallskip}

\newcommand{\profilewidth}{\columnwidth}
    
\begin{document}
\maketitle

\begin{abstract}
    An important---but very demanding---property in collective decision-making is strategyproofness, which requires that voters cannot benefit from submitting insincere preferences. \citet{Gibb77a} has shown that only rather unattractive rules are strategyproof, even when allowing for randomization. However, Gibbard's theorem is based on a rather strong interpretation of strategyproofness, which deems a manipulation successful if it increases the voter's expected utility for \emph{at least one} utility function consistent with his ordinal preferences. In this paper, we study weak strategyproofness, which deems a manipulation successful if it increases the voter's expected utility for \emph{all} utility functions consistent with his ordinal preferences. We show how to systematically design attractive, weakly strategyproof social decision schemes (SDSs) and explore their limitations for both strict and weak preferences. In particular, for strict preferences, we show that there are weakly strategyproof SDSs that are either \emph{ex post} efficient or Condorcet-consistent, while neither even-chance SDSs nor pairwise SDSs satisfy both properties and weak strategyproofness at the same time. By contrast, for the case of weak preferences, we discuss two sweeping impossibility results that preclude the existence of appealing weakly strategyproof SDSs. 
\end{abstract}

\section{Introduction}

Any mechanism that relies on the private information of agents should incentivize agents to report their private information truthfully. However, designing mechanisms that satisfy this property---known as \emph{strategyproofness}---is a challenging task in many domains of economic interest. This is particularly true for the field of social choice, which studies voting rules that aggregate the voters' preferences into a collective decision: a seminal result by \citet{Gibb73a} and \citet{Satt75a} shows that voters can benefit by lying about their true preferences in \emph{any} reasonable deterministic voting rule. 
Early hopes that more positive results can be achieved for randomized voting rules were shattered by \citet{Gibb77a}. \citeauthor{Gibb77a} considered social decision schemes (SDSs), which return a probability distribution for each profile of individual preferences, and the final winner will be chosen by chance according to this distribution.

In particular, \citet{Gibb77a} has shown that the only strategyproof and \emph{ex post} efficient SDSs are \emph{random dictatorships}, which choose each voter with a fixed probability and implement this voter's favorite alternative as the winner of the election. While this result allows for more rules than the Gibbard-Satterthwaite theorem, random dictatorships suffer from a large variance and fail to identify good compromise alternatives.
The latter observation is related to the fact that random dictatorships violate \emph{Condorcet-consistency}, i.e., they may fail to select an alternative that beats all other alternatives in a pairwise majority {comparison.}

Like all results on strategyproof SDSs, Gibbard's random dictatorship theorem crucially hinges on the exact definition of strategyproofness, which, in turn, depends on the assumptions of the voters' preferences over lotteries. \citet{Gibb77a} postulates that a voter prefers one lottery to another if the former yields at least as much expected utility as the latter for every utility function that is consistent with his true preferences. Then, an SDS is called strategyproof if voters always prefer the outcome when voting truthfully to every outcome they could obtain by lying about their true preferences. This strategyproofness notion, which we will call \emph{strong strategyproofness}, is predominant in the literature \citep[e.g.,][]{Barb79a,Proc10a,BLR23a} because it guarantees that voters cannot manipulate regardless of their exact utility functions. However, as demonstrated by the random dictatorship theorem, strong strategyproofness mainly leads to negative results.

In this paper, we will thus study a weaker notion of strategyproofness, first considered by \citet{PoSc86a} and later popularized by \citet{BoMo01a} in the context of random assignment. To this end, we observe that the preferences over lotteries defined by \citet{Gibb77a} are incomplete because there are lotteries such that a voter prefers neither of them to the other. This view results in two ways to define strategyproofness, depending on how we interpret incomparable lotteries. Strong strategyproofness, as defined by \citet{Gibb77a}, views a deviation to an incomparable lottery as a successful manipulation. By contrast, we only consider deviations to comparable lotteries as successful and deem all others as unsuccessful. 
The resulting strategyproofness notion is called \emph{weak strategyproofness} and requires that voters cannot obtain a strictly preferred lottery by lying about their true preferences.

\paragraph{Our contribution.} We improve our understanding of weak strategyproofness by contributing various positive and negative results. A summary of these results is given in \Cref{tab:resultssummary}.

In our first theorem, we introduce the large class of \emph{score-based} SDSs and show that all SDSs within this class are weakly strategyproof. This result allows, e.g., to construct appealing weakly strategyproof SDSs that satisfy Condorcet-consistency or \emph{ex post} efficiency, or that approximate deterministic voting rules arbitrarily close. Both of these objectives are impossible to achieve for strong strategyproofness \citep{Proc10a,BLR23a}. Secondly, we present characterizations of weakly strategyproof tops-only SDSs, which only have access to the voters' favorite alternatives. In this context, we also cover the design of weakly strategyproof \emph{even-chance} SDSs, which always return a uniform lottery over some subset of alternatives. It has often been argued that such SDSs are more acceptable because uniform lotteries are easier to grasp cognitively and to implement in practice \citep[e.g.,][]{Fish72a,Gard79a,BSS19a}. However, no attractive even-chance SDS satisfies strong strategyproofness.

We also analyze the limitations of weakly strategyproof SDSs. In particular, for strict preferences, we show that no weakly strategyproof, even-chance SDS simultaneously satisfies \emph{ex post} efficiency and Condorcet-consistency, and that no weakly strategyproof, pairwise, and neutral SDS (which can only access the pairwise majority comparisons between alternatives) satisfies \emph{ex post} efficiency. These results indicate that Condorcet-consistency and \emph{ex post} efficiency may be incompatible for \emph{all} weakly strategyproof SDSs.
However, such a result seems very difficult to obtain as, for every $\varepsilon>0$, there are weakly strategyproof and Condorcet-consistent SDSs that assign at most $\varepsilon$ probability to Pareto-dominated alternatives.

Finally, we also consider weakly strategyproof SDSs for weak preferences. In this setting, we provide the first easily verifiable proof of an important impossibility theorem by \citet{BBEG16a}, showing that no weakly strategyproof SDS simultaneously satisfies anonymity, neutrality, and \emph{ex ante} efficiency. When restricting attention to even-chance SDSs, we strengthen this result by proving that all weakly strategyproof and \emph{ex post} efficient SDSs always randomize over the favorite alternatives of at most two fixed voters. These results show that, if we allow weak preferences, even mild forms of strategyproofness preclude attractive SDSs.

\paragraph{Related work.}

Studying weaker forms of strategyproofness is an active area in social choice theory \citep[e.g.,][]{BoMo01a,Balb16a,ABBB15a,BBEG16a,Lede21c,MeSe21a}. 
Unfortunately, in the realm of voting, this approach mainly led to strengthened impossibility results: for instance, \citet{Lede21c} studies a strategyproofness notion that lies logically between weak and strong strategyproofness and shows that it is still incompatible with Condorcet-consistency. 
Despite these negative results, weak strategyproofness has received little attention in social choice theory. In particular, only few SDSs, such as the Condorcet rule \citep{PoSc86a} or the egalitarian simultaneous reservation rule by \citet{AzSt14a}, are known to be weakly strategyproof. In more recent works, weak strategyproofness was used to prove impossibility theorems for the case of weak preferences \citep{ABBB15a,BBS15b,BBEG16a}. This approach culminated in a sweeping impossibility theorem for weak preferences: no weakly strategyproof SDS satisfies anonymity, neutrality, and \emph{ex ante} efficiency \citep{BBEG16a}.

Some of our results can also be compared to results for set-valued voting rules (which return sets of winning alternatives instead of lotteries). In particular, even-chance SDSs can be interpreted as set-valued voting rules and weak strategyproofness then translates to a strategyproofness notion called even-chance strategyproofness \citep[e.g.,][]{Gard79a,BSS19a}. This strategyproofness notion is slightly stronger than commonly considered set-valued strategyproofness notions such as Fishburn-strategyproofness \citep{Fish72a} or Kelly-strategyproofness \citep{Kell77a}, and our paper is thus related to recent work on set-valued voting rules \citep[e.g.,][]{BoEn21a,BSS19a,BrLe21a}. 

\section{Preliminaries}

Let $N=\{1,\dots, n\}$ denote a set of $n$ voters and let $A=\{a_1,\dots, a_m\}$ denote a set of $m\geq 2$ alternatives. Every voter $i\in N$ reports a (weak) \emph{preference relation} $\succsim_i$, which is a complete and transitive binary relation on $A$. The strict part of $\succsim_i$ is denoted by $\succ_i$ (i.e., $x\succ_i y$ iff $x\succsim_i y$ and not $y\succsim_i x$) and the indifference part by $\sim_i$ (i.e., $x\sim_i y$ iff $x\succsim_i y$ and $y\succsim_i x$). A preference relation $\succsim_i$ is called \emph{strict} if its irreflexive part coincides with ${\succ_i}$. We represent preference relations by comma-separated lists, where brackets indicate that a voter is indifferent between some alternatives. For instance, $a, \{b,c\}$ denotes that the considered voter prefers $a$ to both $b$ and $c$ and is indifferent between the latter two alternatives. We denote the set of all strict preference relations by $\mathcal{L}$ and the set of all weak preference relations by~$\mathcal{R}$.

A (weak) \emph{preference profile $R=(\succsim_1,\dots, \succsim_n)$} is a vector that specifies the preference relations $\succsim_i$ of all voters $i\in N$. A preference profile is \emph{strict} if the preference relations of all voters are strict. 
The set of all strict preference profiles is given by $\mathcal{L}^N$, and the set of all weak preference profiles is $\mathcal{R}^N$. 
We represent preference profiles as collections of preference relations, where the set of voters that report a preference relation is stated directly before the preference relation. For instance, $\{1,2,3\}: a,b,c$ means that voters $1$, $2$, and $3$ prefer $a$ to $b$ to $c$. To improve readability, we omit curly brackets for singleton sets. 

The main objects of study in this paper are social decision schemes, which intuitively are voting rules that may use chance to determine the winner of an election. To formalize this, we define \emph{lotteries} as probability distributions over the alternatives, i.e., a lottery $p$ is a function of the type $A\rightarrow [0,1]$ such that $\sum_{x\in A} p(x)=1$. Moreover, by $\Delta(A)$ we denote the set of all lotteries over $A$. Then, a \emph{social decision scheme (SDS)} on a domain $\mathcal{D}\in\{\mathcal{L}^N, \mathcal{R}^N\}$ is a function that maps every preference profile $R\in\mathcal{D}$ to a lottery $p\in \Delta(A)$. We denote by $f(R,x)$ the probability that the SDS $f$ assigns to alternative $x$ in the profile $R$ and extend this notion to sets of alternatives $X$ by defining $f(R,X)=\sum_{x\in X} f(R,x)$.

We will sometimes restrict our attention to even-chance SDSs. These SDSs pick a set of alternatives and randomize uniformly over these alternatives. Formally, an SDS $f$ is \emph{even-chance} if it chooses for every profile $R$ a non-empty set of alternatives $X$ such that $f(R,x)=\frac{1}{|X|}$ if $x\in X$ and $f(R,x)=0$ otherwise. Even-chance lotteries are appealing because of their simplicity and because non-uniform randomization may be difficult to implement in the real world. On top of that, even-chance SDSs can naturally be interpreted as set-valued voting rules,
which have been studied in detail in the social choice community.

\subsection{Strategyproofness}

The central axiom for our analysis is strategyproofness which demands that voters cannot benefit by lying about their true preferences. To define this axiom for SDSs, we need to specify how voters compare lotteries over the alternatives. We assume for this that voters have utility functions $u_i:A\rightarrow\mathbb{R}$ and aim to maximize their expected utility. However, the exact utility functions are not known as voters only reveal their ordinal preferences over alternatives.
We will thus quantify over all utility functions $u_i$ that are consistent with the voter's preference relation $\succsim_i$, i.e., that satisfy that $u_i(x)\geq u_i(y)$ if and only if $x\succsim_i y$ for all $x,y\in A$. A voter then prefers lottery $p$ to lottery $q$, denoted by $p\succsim_i q$, if $p$ guarantees him at least as much expected utility as $q$ for \emph{every} utility function $u_i$ that is consistent with his preference relation, i.e., if $\mathbb{E}[u_i(p)]\geq \mathbb{E}[u_i(q)]$ for all consistent utility functions $u_i$. 
Alternatively, this lottery extension can also be defined based on stochastic dominance. To state this definition, we let the \emph{upper contour set}  $U({\succsim_i}, x)=\{y\in A\colon y\succsim_i x\}$ of $x$ denote the set of alternatives that voter $i$ weakly prefers to $x$. It then holds that $p\succsim_i q$ if and only if $p(U({\succsim_i}, x))\geq q(U({\succsim_i}, x))$ for all $x\in A$ 
\citep[see, e.g.,][]{Sen11a,BBEG16a}.

Importantly, the voters' preferences over lotteries, as defined above, are incomplete, i.e., there are lotteries $p$ and $q$ and a preference relation $\succsim_i$ such that neither $p\succsim_i q$ nor $q\succsim_i p$. For example, for the preference relation $a,b,c$, the lotteries $p$ and $q$ defined by $p(a)=p(b)=p(c)=\nicefrac{1}{3}$ and $q(b)=1$ are incomparable as neither of them stochastically dominates the other. Consequently, there are two ways to define strategyproofness depending on how we handle incomparable lotteries. The approach suggested by \citet{Gibb77a} counts a deviation to an incomparable lottery as a successful manipulation and strategyproofness hence prohibits such deviations. This results in \emph{strong strategyproofness}, which requires of an SDS $f$ that $f(R)\succsim_i f(R')$ for all profiles $R,R'$ and voters $i\in N$ such that ${\succsim_j}={\succsim_j'}$ for all $j\in N\setminus \{i\}$. By contrast, we will not count the deviation to an incomparable lottery as a successful manipulation. This leads to a weaker form of strategyproofness: an SDS $f$ is \emph{weakly strategyproof} if $f(R')\not\succ_i f(R)$ for all profiles $R$, $R'$ and voters $i\in N$ such that ${\succsim_j}={\succsim_j}'$ for all $j\in N\setminus \{i\}$. 

\subsection{Further Axioms}

We conclude this section by stating four standard axioms. 

\paragraph{Anonymity.} Anonymity is a basic fairness property that states that the identities of the voters should not matter. Formally, an SDS $f$ is \emph{anonymous} if $f(R)=f(\pi(R))$ for all permutations $\pi:N\rightarrow N$ and profiles $R$, where the profile $R'=\pi(R)$ is  given by ${\succsim'_{\pi(i)}}={\succsim_i}$ for all voters $i\in N$.

\paragraph{Neutrality.} Similar to anonymity, neutrality is a fairness property that requires that alternatives are treated equally. In more detail, an SDS $f$ is \emph{neutral} if $f(\tau(R), \tau(x))=f(R,x)$ for all permutations $\tau:A\rightarrow A$ and profiles $R$. This time, the profile $R'=\tau(R)$ is defined by $\tau(x)\succsim_i' \tau(y)$ if and only if $x\succsim_i y$ for all $x,y\in A$ and $i\in N$.

\paragraph{\emph{Ex post} efficiency.} \emph{Ex post} efficiency postulates that alternatives should have no chance of being selected if there is another alternative that makes at least one voter better off without making any other voter worse off. To this end, we say an alternative $x$ \emph{Pareto-dominates} another alternative $y$ in a profile $R$ if $x\succsim_i y$ for all voters $i\in N$ and $x\succ_i y$ for some $i\in N$. 
Conversely, an alternative is \emph{Pareto-optimal} in a profile $R$ if it is not Pareto-dominated by any other alternative. Finally, an SDS $f$ is \emph{ex post efficient} if $f(R,x)=0$ for all profiles $R$ and Pareto-dominated alternative $x$.

\paragraph{Condorcet-consistency.} Condorcet-consistency demands that a Condorcet winner, an alternative that beats every other alternative in a pairwise majority comparison, should be selected with probability $1$ whenever it exists. To formalize this, we define the \emph{majority relation $\succsim_M$} of a profile $R$ by $x\succsim_M y$ if and only if $|\{i\in N\colon x\succ_i y\}|\geq |\{i\in N\colon y\succ_i x\}|$ for all $x,y\in A$. Moreover, $\succ_M$ denotes the strict part of $\succsim_M$ and $\sim_M$ its indifferent part. Then, an alternative $x$ is a \emph{Condorcet winner} in a profile $R$ if $x\succ_M y$ for all $y\in A\setminus \{x\}$, and an SDS $f$ is \emph{Condorcet-consistent} if $f(R,x)=1$ whenever $x$ is the Condorcet winner in $R$.

\section{Results}

We are now ready to state our results. We first present theorems that allow the design attractive weakly strategyproof SDSs (\Cref{subsec:strictprefpos}), and then discuss the limitations of weakly strategyproof SDSs (\Cref{subsec:strictprefneg,subsec:weakpref}). Due to space constraints, we defer most of our proofs to the appendix and present proof sketches instead.

\subsection{Possibility Theorems for Strict Preferences}\label{subsec:strictprefpos}

In this section, we will show how to design weakly strategyproof SDSs when voters have strict preferences. In more detail, we will first present a large class of weakly strategyproof SDSs (cf. \Cref{thm:profscoring}) and then give two characterizations of weakly strategyproof SDSs that only depend on the voters' favorite alternatives (cf. \Cref{thm:topsonly}). 

We first introduce the class of score-based SDSs and show that all these SDSs are weakly strategyproof. For this, let $R^{i:yx}$ be the profile derived from another profile $R$ by only changing voter $i$'s preferences from $x\succ_iy$ to~${y\succ_i^{i:yx} x}$. In particular, this requires that there is no $z\in A$ with $x\succ_i z\succ_i y$. Next, a function $s:\mathcal{L}^N\times A\rightarrow \mathbb{R}_{\geq 0}\cup \{\infty\}$ is a \emph{score function} if it satisfies for all profiles $R\in\mathcal{L}^N$, distinct alternatives $x,y,z\in A$, and voters $i\in N$ that 

\begin{itemize}[leftmargin=*,topsep=3pt,itemsep=0pt]
\item $s(R,x)=\infty$ implies $s(R,y)\neq\infty$ (at most one infinity),
\item $s(R,z)=s(R^{i:yx},z)$ (localizedness),
\item $s(R,y)\leq s(R^{i:yx},y)$ (monotonicity), and
\item $s(R,y)=s(R^{i:yx},y)$ implies $s(R,x)=s(R^{i:yx},x)$ unless $s(R,y)=\infty$, or $s(R,x)=\infty$ and $s(R,y)>0$ (balancedness).
\end{itemize}

We note that a score function can assign a score of infinity to at most one alternative. We thus assume the usual arithmetic rules for infinity: for all $x\in\mathbb{R}$, it holds that $\infty>x$, $\infty+x=\infty$, $\frac{x}{\infty}=0$, and $\frac{\infty}{\infty}=1$. 
Finally, an SDS $f$ on $\mathcal{L}^N$ is \emph{score-based} if there is a score function $s$ such that $\sum_{y\in A} s(R,y)>0$ and $f(R,x)=\frac{s(R,x)}{\sum_{y\in A} s(R,y)}$ for all alternatives $x\in A$ and profiles $R\in\mathcal{L}^N$.

We will next discuss several examples of score-based SDSs to illustrate this class and its versatility. To this end, we first consider two classical score functions, namely the Copeland score function $s_{\mathit{C}}(R,x)={|\{y\in A\setminus \{x\}\colon x\succ_M y\}|}+\frac{1}{2}|\{y\in A\setminus \{x\}\colon x\sim_M y\}|$ and the plurality score function $s_{\mathit{P}}(R,x)=|\{i\in N\colon \forall y\in A\setminus \{x\} \colon x\succ_i y\}|$. Both of these functions are indeed score functions according to our definition and the corresponding SDSs are thus score-based. Moreover, for every strictly monotonically increasing function $g:\mathbb{R}_{\geq 0}\rightarrow\mathbb{R}_{\geq 0}$, it holds that $s_{\mathit{C}}^g(R,x)=g(s_{\mathit{C}}(R,x))$ and $s_{\mathit{P}}^g(R,x)=g(s_{\mathit{P}}(R,x))$ are also score functions. For example, this means that the SDSs defined by the functions $s^k_{\mathit{C}}$ and $s^k_{\mathit{P}}$, which take the $k$-th power of $s_{\mathit{C}}(R,x)$ and $s_{\mathit{P}}(R,x)$, are score-based. Even SDSs that seem rather unrelated to score functions belong to our class. For instance, the Condorcet rule (which chooses the Condorcet winner with probability $1$ whenever it exists and randomizes uniformly over all alternatives otherwise) is the score-based SDS defined by the score function $s$ with $s(R,x)=\infty$ if $x$ is the Condorcet winner in $R$ and $s(R,x)=1$ otherwise. Similarly, the function $s_{\mathit{C}}^{k,\mathit{CW}}$, which assigns a score of infinity to the Condorcet winner and otherwise coincides with $s_{\mathit{C}}^{k}$, satisfies all our conditions and is thus a score function.

We will now prove that all score-based SDSs are weakly strategyproof.

\begin{restatable}{theorem}{profscoring}
\label{thm:profscoring}
    Every score-based SDS on $\mathcal{L}^N$ satisfies weak strategyproofness.
\end{restatable}
\begin{proof}[Proof sketch]
    Let $f$ be a score-based SDS and let $s$ be its score function. Moreover, we consider two profiles $R,R'\in\mathcal{L}^N$ and a voter $i\in N$ such that ${\succsim_j}={\succsim_j'}$ for all $j \in N\setminus \{i\}$ and $f(R)\neq f(R')$. To simplify this proof sketch, we additionally assume that $0<s(R,x)<\infty$ and $s(R',x)<\infty$ for all $x\in A$, and that ${\succsim_i}=x_1,x_2,\dots,x_m$. Next, we define $s_\mathit{total}(\hat R)=\sum_{x\in A} s(\hat R,x)$ and consider three cases. First, if $s_\mathit{total}(R)<s_\mathit{total}(R')$, we use the monotonicity and localizedness of $s$ to show that $s(R,x_1)\geq s(R',x_1)$ by transforming $R$ to $R'$ with a swap sequence that never reinforces~$x_1$. Since $s_\mathit{total}(R)<s_\mathit{total}(R')$, it follows that $f(R,x_1)=\frac{s(R,x_1)}{s_\mathit{total}(R)}>\frac{s(R',x_1)}{s_\mathit{total}(R')}=f(R',x_1)$ and thus $f(R')\not\succ_i f(R)$. If $s_\mathit{total}(R)>s_\mathit{total}(R')$, we can use a similar argument by showing that the score of voter $i$'s least preferred alternative $x_m$ weakly increases when going from $R$ to $R'$. Finally, if $s_\mathit{total}(R)=s_\mathit{total}(R')$, we let $x_h$ denote the alternative such that  $s(R,x_\ell)=s(R',x_\ell)$ for all $\ell<h$ and $s(R,x_h)\neq s(R',x_h)$. Then, we prove that $s(R,x_h)> s(R',x_h)$, which shows that $f(R, U(\succsim_i, x_h))>f(R', U(\succsim_i, x_h))$ and thus $f(R')\not\succ_i f(R)$. Finally, slightly more involved arguments extend this analysis to the case that $s_\mathit{total}(R)=\infty$ or $s_\mathit{total}(R')=\infty$.
\end{proof}

\Cref{thm:profscoring} has a number of important consequences. Firstly, this result implies that the score-based SDSs defined by $s_{\mathit{P}}^k$ and $s_{\mathit{C}}^k$ are weakly strategyproof. Since all these rules fail strong strategyproofness when $k\neq 1$, this demonstrates that the space of weakly strategyproof SDSs is significantly richer than the one of strongly strategyproof SDSs. Secondly, we note that the score-based SDSs defined by $s_{\mathit{P}}^k$ and $s_{\mathit{C}}^k$ approximate the Plurality rule and Copeland rule (which simply choose the alternatives with maximal Plurality and Copeland score, respectively) arbitrarily closely by increasing the exponent $k$. This stands in sharp contrast to a result by \citet{Proc10a} who has shown that strongly strategyproof SDSs are poor approximations of common deterministic voting rules. Thirdly, \Cref{thm:profscoring} implies that there are interesting weakly strategyproof SDSs that are \emph{ex post} efficient or Condorcet-consistent. For instance, all score-based SDSs defined by a score-function $s_{\mathit{P}}^g(R,x)=g(s_{\mathit{P}}(R,x))$ are \emph{ex post} efficient when $g$ satisfies that $g(0)=0$ and $g(x)>g(y)$ for all $x,y\in\mathbb{N}_0$ with $x>y$, and the Condorcet rule as well as the score-based SDS defined by $s_{\mathit{C}}^{k,\mathit{CW}}$ are Condorcet-consistent. This stands again in contrast to results for strong strategyproofness because \citet{BLR23a} have shown that all strongly strategyproof SDSs can put at most probability $2/m$ on Condorcet winners, and that all strongly strategyproof SDSs that assign at most a probability of less than $1/m$ to Pareto-dominated alternatives have a random dictatorship component.

A natural follow-up question for \Cref{thm:profscoring} is whether the class of score-based SDSs is equivalent to the set of weakly strategyproof SDSs. This is not the case since the omninomination rule $f^O$, which randomizes uniformly over the set of top-ranked alternatives $\textit{OMNI}(R)\!=\!\{{x\!\in\! A}\colon s_{\mathit{P}}(R,x)>0\}$, is weakly strategyproof but not score-based. In particular, every score function that induces this SDS fails balancedness or localizedness, so it cannot be score-based.

To give some characterizations for weakly strategyproof SDSs, we will next focus on the class of tops-only SDSs, which only depend on the voters' favorite alternatives. More formally, let $T_i(R)=\{x\in A\colon x\succsim_i y \text{ for all }y\in A\}$ denote the set of voter $i$'s favorite alternatives and note that $|T_i(R)|=1$ if $R$ is strict. Then, an SDS $f$ is \emph{tops-only} if $f(R)=f(R')$ for all preference profiles $R$ and $R'$ such that $T_i(R)=T_i(R')$ for all $i\in N$. We will now provide two characterizations of weak strategyproofness for tops-only SDSs on $\mathcal{L}^N$: firstly, we will show that, for tops-only SDSs, weak strategyproofness is equivalent to a monotonicity property. Furthermore, we will characterize the class of weakly strategyproof SDSs that are tops-only, even-chance, anonymous, and neutral as parameterized omninomination rules. These SDSs are defined by two parameters $\theta_1>\frac{n}{2}$ and $\theta_2$ and coincide with $f^O$ except for two special cases: \emph{(i)} if a single alternative is top-ranked by at least $\theta_1$ voters, then this alternative is assigned probability $1$, and \emph{(ii)} if no such alternative exists and more that $\theta_2$ alternatives are top-ranked in total, then we randomize uniformly over all alternatives. More formally, an SDS $f$ is a \emph{parameterized omninomination rule} if there are two parameters $\theta_1\in \{\lceil{\frac{n+1}{2}}\rceil,\dots, n+1\}$ and $\theta_2\in \{0,\dots, m-1\}$ such that

\begin{itemize}
    \item $f(R,x)=1$ for all profiles $R$ and alternatives $x\in A$ such that $s_{\mathit{P}}(R,x)\geq \theta_1$,
    \item $f(R)=f^O(R)$ for all profiles $R$ such that $\max_{x\in A} s_{\mathit{P}}(R,x)< \theta_1$ and $|\textit{OMNI}(R)|\leq \theta_2$,
    \item $f(R,x)=\frac{1}{m}$ for all profiles $R$ and alternatives $x\in A$ such that $\max_{x\in A} s_{\mathit{P}}(R,x)< \theta_1$ and $|\textit{OMNI}(R)|> \theta_2$.
\end{itemize}

\begin{restatable}{theorem}{topsonly}\label{thm:topsonly}
     Let $f$ denote a tops-only SDS on $\mathcal{L}^N$.
    \begin{enumerate}[leftmargin=*,label=\arabic*)]
        \item $f$ is weakly strategyproof if and only if $f(R)=f(R')$ or $f(R,T_i(R))>f(R',T_i(R))$ for all profiles $R,R'\in\mathcal{L}^N$ and voters $i\in N$ such that ${\succsim_j}={\succsim_j'}\,$  for all $j\in N\setminus \{i\}$.
        \item $f$ is weakly strategyproof, even-chance, anonymous, and neutral if and only if it is a parameterized omninomination rule.
    \end{enumerate}
\end{restatable}
\begin{proof}
    We will only prove the first claim here and defer the proof of the second part of the theorem to the appendix.
    Thus, let $f$ denote a tops-only SDS. We first show the direction from right to left and hence suppose that $f$ satisfies the given condition. Now, let $R$ and $R'$ denote two preference profiles and $i$ a voter such that ${\succsim_j}={\succsim_j'}$ for all $j\in N\setminus \{i\}$. Moreover, let $x$ denote voter $i$'s favorite alternative in $R$ and $y$ his favorite alternative in $R'$. If $f(R)=f(R')$, voter $i$ clearly cannot manipulate by deviating from $R$ to $R'$. Hence, we suppose that $f(R)\neq f(R')$, which requires that $x\neq y$ due to tops-onlyness. In turn, the condition of our theorem implies that $f(R',x)<f(R,x)$ if $f(R)\neq f(R')$ and $x\neq y$. This implies that $f(R')\not\succ_i f(R)$, so $f$ is weakly strategyproof.

    Next, we will show that $f$ fails weak strategyproofness if it fails the condition in the theorem. To this end, assume there are profiles $R$ and $R'$, a voter $i$, and an alternative $x$ such that ${\succsim_j}={\succsim_j'}$ for all $j\in N\setminus \{i\}$, $T_i(R)=\{x\}$, $f(R)\neq f(R')$, and $f(R,x)\leq f(R',x)$. We define $Z^+=\{z\in A\setminus \{x\}\colon f(R',z)\geq f(R,z)\}$ and $Z^-=\{z\in A\setminus \{x\}\colon f(R',z)< f(R,z)\}$ and observe that $Z^-\neq \emptyset$ since $f(R)\neq f(R')$. Moreover, we consider the profile $R^*$ derived from $R$ by assigning voter $i$ a strict preference relation $\succsim_i^*$ with $x\succ_i^* z^+\succ_i^* z^-$ for all $z^+\in Z^+$ and $z^-\in Z^-$. By tops-onlyness, $f(R)=f(R^*)$. On the other hand, it holds by construction that $f(R')\neq f(R^*)$ and $f(R', U(\succsim^*_i, y))\geq f(R^*, {U(\succsim_i^*,y)})$ for all $y\in A$, so $f(R')\succ_i^* f(R^*)$ and $f$ fails weak strategyproofness. 
\end{proof}

\remark 
The second claim of \Cref{thm:topsonly} can be used to characterize the SDS that assigns probability $1$ to an alternative if it is top-ranked by a strict majority of voters and randomizes uniformly over $\mathit{OMNI}(R)$ if no such alternative exists: among all SDSs that are tops-only, even-chance, weakly strategyproof, anonymous, and neutral, it is the one that uses the least amount of randomization.

\remark
Every strongly strategyproof SDS is score-based. This follows from a result by \citet{Gibb77a}, which states that an SDS is strongly strategyproof if and only if the SDS itself is localized and monotonic. In our terminology, this means that an SDS is strongly strategyproof if and only if it is defined by score function $s:\mathcal{L}^N\times A\rightarrow\mathbb{R}_{\geq 0}$ that satisfies monotonicity, localizedness, and that $s(R,x)-s(R^{i:yx},x)=s(R^{i:yx},y)-s(R,y)$ for all $x,y\in A$ and $R\in\mathcal{L}^N$. By replacing the last constraint by balancedness (and even allowing at most one alternative with $s(R,x)=\infty$), we derive significantly more flexibility in the design of weakly strategyproof SDSs.

\subsection{Impossibility Theorems for Strict Preferences}\label{subsec:strictprefneg}

We will now turn to the limitations of weakly strategyproof SDSs for the case that voters report strict preferences. To this end, we observe that, while our results in \Cref{subsec:strictprefpos} allow to construct weakly strategyproof SDSs that are arbitrarily close to simultaneously satisfying \emph{ex post} efficiency and Condorcet-consistency (e.g., the score-based rule defined by $s_{\mathit{C}}^k$ for very large $k$), we were not able to construct a weakly strategyproof SDS that satisfies both axioms at the same time. As it turns out, constructing such an SDS may be impossible. In more detail, we subsequently prove that no even-chance SDS (cf. \Cref{thm:EvenchanceCondImp}) and no neutral and pairwise SDSs (cf. \Cref{thm:impSDS}) simultaneously satisfies weak strategyproofness, Condorcet-consistency, and \emph{ex post} efficiency. This shows that the most common approaches for designing Condorcet-consistent SDSs do not allow to simultaneously satisfy weak strategyproofness, \emph{ex post} efficiency, and Condorcet-consistency.

Let us first consider even-chance SDSs.

\begin{restatable}{theorem}{evenchance}
	\label{thm:EvenchanceCondImp}
    Assume that $m\geq 5$ and $n\geq 5$ is odd. No even-chance SDS on $\mathcal{L}^N$ satisfies weak strategyproofness, Condorcet-consistency, and \emph{ex post} efficiency.
\end{restatable}
\begin{proof}[Proof Sketch]
    For the proof of this result, we first show two auxiliary claims: assuming that the number of voters $n$ is odd, we prove \emph{(i)} that every weakly strategyproof and Condorcet-consistent even-chance SDS assigns probability $1$ to an alternative $x$ if and only if $x$ is the Condorcet winner, and \emph{(ii)} that such SDSs can never randomize over exactly two alternatives. We then consider the following two preference profiles; additional alternatives are bottom-ranked by all voters.\smallskip
    
    \setlength\tabcolsep{3 pt}
    {\medmuskip=0mu\relax
	\thickmuskip=1mu\relax
    \noindent\begin{tabular}{L{0.07\profilewidth} L{0.3\profilewidth} L{0.3\profilewidth} L{0.3\profilewidth}}
		$R^1$: & $1$: $b,e,d,c,a$ & $2:$ $a,b,c,e,d$ & $3:$ $e,d,c,a,b$
  \end{tabular}
   \begin{tabular}{L{0.07\profilewidth} L{0.43\profilewidth} L{0.43\profilewidth} L{0.45\profilewidth}}
		& $\{4, \dots, \frac{n+3}{2}\}:$ $b,c,a,e,d$ & $\{\frac{n+5}{2}, \dots,n \}$: $e,d,a,b,c$
	\end{tabular}\smallskip
 
    \noindent\begin{tabular}{L{0.07\profilewidth} L{0.3\profilewidth} L{0.3\profilewidth} L{0.3\profilewidth}}
	$\hat R^1$: & $1$: $b,e,d,c,a$ & $2:$ $a,b,c,e,d$ & $3:$ $d,a,e,b,c$
    \end{tabular}
   \begin{tabular}{L{0.07\profilewidth} L{0.43\profilewidth} L{0.43\profilewidth} L{0.45\profilewidth}}
	& $\{4, \dots, \frac{n+3}{2}\}$: $b,c,a,e,d$ & $\{\frac{n+5}{2}, \dots,n\}$: $e,d,a,b,c$
    \end{tabular}}\smallskip

    For these profiles, we show based on our auxiliary claims that every even-chance SDS that is weakly strategyproof, Condorcet-consistent, and \emph{ex post} efficient uniformly randomizes over $\{a,b,c,e\}$ for $R^1$ and over $\{a,b,d,e\}$ for $\hat R^1$. This means that voter $3$ can manipulate by deviating from $R^1$ to $\hat R^1$, contradicting weak strategyproofness. 
\end{proof}

Let us now turn to pairwise and neutral SDSs. An SDS is \emph{pairwise} if $f(R)=f(R')$ for all profiles $R$, $R'$ such that $|\{i\in N: x\succ_i y\}|-|\{i\in N: y\succ_i x\}|=|\{i\in N\colon x\succ_i' y\}|-|\{i\in N: y\succ_i' x\}|$ for all $x,y\in A$. In other words, an SDS is pairwise if it only depends on the weighted majority relation. There are many important pairwise SDSs \citep[see, e.g.,][Chapter 3 and~4]{BCE+14a} and most Condorcet-consistent voting rules are pairwise. Hence, showing that all pairwise, neutral, and weakly strategyproof SDSs fail \emph{ex post} efficiency can be seen as evidence that no SDS simultaneously satisfies weak strategyproofness, Condorcet-consistency, and \emph{ex post} efficiency. We need to slightly extend the domain of SDSs for the following result: $\mathcal{L}^*=\bigcup_{N\subseteq\mathbb{N}\text{ is finite and non-empty}} \mathcal{L}^N$ contains all strict preference profiles for every finite and non-empty electorate. This is required because pairwiseness establishes relationships between profiles with different numbers of voters.

\begin{restatable}{theorem}{impSDS}
    \label{thm:impSDS}
    Assume that $m\geq 5$. No pairwise, neutral, and weakly strategyproof SDS on $\mathcal{L}^*$ satisfies \emph{ex post} efficiency.
\end{restatable}
\begin{proof}[Proof Sketch]
    For this proof sketch, we focus on the case that there are $m=5$ alternatives. Moreover, we will only consider profiles with $n=3$ voters; this suffices as we work in $\mathcal{L}^*$. Nevertheless, the impossibility theorem can be extended to all odd values $n>3$ by adding voters with inverse preferences as pairwiseness requires that the outcome does not change when adding such voters. Now, assume for contradiction that there is an SDS $f$ that satisfies all axioms of our theorem. First, we show that $f$ is invariant under weakening alternatives that obtain probability $0$, i.e., if $R'$ arises from a profile $R$ by only weakening an alternative $x$ with $f(R,x)=0$, then $f(R)=f(R')$. Next, we focus on the profiles $R$ and $R'$. 
\smallskip

    \noindent\begin{profile}{C{0.07\profilewidth} C{0.25\profilewidth} C{0.25\profilewidth} C{0.25\profilewidth}}
        $R$: & $1$: $a,b,e,c,d$ & $2$: $b,c,e,a,d$ & $3$: $e,c,a,b,d$\\
        $R'$: & $1$: $a,b,e,c,d$ & $2$: $b,c,a,e,d$ & $3$: $e,c,a,b,d$
    \end{profile}

    In particular, we show that $f(R,a)=f(R,b)=f(R,e)=\frac{1}{3}$ and $f(R',a)=f(R',b)=f(R',c)=\frac{1}{3}$ by analyzing several auxiliary profiles. However, this means that voter $2$ can manipulate by deviating from $R$ to $R'$ as $f(R')\succ_2 f(R)$, so $f$ fails weak strategyproofness.
\end{proof}

\remark
	All axioms except of the even-chance condition are required for \Cref{thm:EvenchanceCondImp}: the omninomination rule $f^O$ satisfies \emph{ex post} efficiency and weak strategyproofness, the Condorcet rule satisfies Condorcet-consistency and weak strategyproofness, and uniformly randomizing over known majoritarian choice sets such as the uncovered set satisfies \emph{ex post} efficiency and Condorcet-consistency. Moreover, the bounds on $n$ and $m$ are tight: if $n\leq 4$, the SDS that chooses the Condorcet winner with probability $1$ if there is one and otherwise randomizes uniformly over the top-ranked alternatives satisfies all given axioms; if $m\leq 4$, the SDS that chooses the Condorcet winner with probability $1$ if there is one, randomizes uniformly over the top-ranked alternative if there are two that are first-ranked by exactly half of the voters, and otherwise randomizes uniformly over the set of \emph{ex post} efficient alternatives meets all conditions of \Cref{thm:EvenchanceCondImp}. Finally, we conjecture that the even-chance assumption is not required for the impossibility. 

\remark
    Just like numerous other results on Condorcet-consistency \citep[e.g.,][]{BoEn21a,BLS22c}, we cannot prove \Cref{thm:EvenchanceCondImp} for even $n$. However, we show in the appendix that there is no even-chance SDS that satisfies weak strategyproofness, \emph{ex post} efficiency, and strong Condorcet-consistency. The last condition requires that an alternative is chosen with probability $1$ if and only if it is the Condorcet winner, and it is typically satisfied by all strategyproof, Condorcet-consistent SDSs. 
    
\remark
	The main result of \citet{BrLe21a} implies that, under mild additional assumptions, no pairwise set-valued voting rule satisfies Condorcet-consistency, \emph{ex post} efficiency, and a strategyproofness notion called Fishburn-strategyproofness. When interpreting even-chance SDSs as set-valued voting rules, \Cref{thm:EvenchanceCondImp} extends this observation to \emph{all} set-valued voting rules at the expense of using a slightly stronger notion of strategyproofness.

\begin{table*}[tb]
\setlength{\parskip}{0pt}
\renewcommand{\arraystretch}{1.5}
\centering
    \begin{tabular}[t]{K{0.2\textwidth}K{0.34\textwidth}K{0.38\textwidth}}
        & Strict preferences & Weak preferences\\\hline
        \emph{ex post} efficiency  & 
        \vspace{-12pt}
        \begin{itemize}[topsep=-12pt, rightmargin=0pt, leftmargin=0pt]
            \item[$\oplus$] Various tops-only SDSs (Thm \ref{thm:profscoring},\ref{thm:topsonly})
            \item[$\ominus$] No pairwise SDS$^*$ (Thm \ref{thm:impSDS})\vspace{-10pt}
        \end{itemize} 
        & 
        \vspace{-12pt}
        \begin{itemize}[topsep=-12pt, rightmargin=0pt, leftmargin=0pt]
        \item[$\oplus$] Random serial dictatorship \citep{ABBB15a}
        \item[$\ominus$] No \emph{ex ante} efficient SDSs$^*$ \citep{BBEG16a}
        \item[$\ominus$] Only (bi)dictatorial even-chance SDSs (Thm \ref{thm:bidicatorship}) 
        \vspace{-10pt}
        \end{itemize}
        \\\hline
        Condorcet-consistency
        & 
         \vspace{-12pt}
        \begin{itemize}[topsep=0pt, rightmargin=0pt, leftmargin=0pt]
        \item[$\oplus$] Variants of Copeland's rule (Thm \ref{thm:profscoring})\vspace{-10pt}
        \end{itemize}
        & 
         \vspace{-12pt}
        \begin{itemize}[topsep=0pt, rightmargin=0pt, leftmargin=0pt]
        \item[$\ominus$] No SDS \citep{Bran11c}\vspace{-10pt}
        \end{itemize}
        \\\hline
        \emph{ex post} efficiency and Condorcet-consistency 
        &  
         \vspace{-12pt}
        \begin{itemize}[topsep=0pt, rightmargin=0pt, leftmargin=0pt]
        \item[$\oplus$] Approximately \emph{ex post} efficient and Condorcet-consistent SDSs (Thm \ref{thm:profscoring})
        \item[$\ominus$] No even-chance SDS (Thm \ref{thm:EvenchanceCondImp})
        \end{itemize}
        & 
         \vspace{-12pt}
        \begin{itemize}[topsep=0pt, rightmargin=0pt, leftmargin=0pt]
            \item[$\ominus$] No SDS \citep{Bran11c}\vspace{-14pt}
        \end{itemize}
    \end{tabular}
    \caption{Summary of our results. Each cell states which weakly strategyproof SDSs satisfy which axioms for strict preferences (left column) and weak preferences (right column).
    Results marked by a $\oplus$ symbol are possibility theorems, whereas the $\ominus$ symbol indicates impossibility theorems. Results with an asterisk ($^*$) additionally need anonymity and/or neutrality.}
    \label{tab:resultssummary}
\end{table*}

\subsection{Impossibility Theorems for Weak Preferences}\label{subsec:weakpref}

In this section, we prove two impossibility theorems for weakly strategyproof SDSs on $\mathcal{R}^N$: firstly, we present a simplified proof of the main result of \citet{BBEG16a} (cf. \Cref{thm:SDimpossibity}) in the appendix; secondly, we prove an even more severe impossibility for even-chance SDSs (cf. \Cref{thm:bidicatorship}). 
These results demonstrate that there are no attractive weakly strategyproof SDSs when voters have weak preferences. 

We start by revisiting the impossibility theorem by \citet{BBEG16a}. It shows that no anonymous, neutral, and weakly strategyproof SDS is \emph{ex ante} efficient. The last condition, also known as SD-efficiency, is a strengthening of \emph{ex post} efficiency focusing on lotteries rather than individual alternatives. In more detail, a lottery $p$ \emph{ex ante dominates} a lottery $q$ in a profile $R$ if $p\succsim_i q$ for all $i\in N$ and $p\succ_i q$ for some $i\in N$. Conversely, a lottery is \emph{ex ante efficient} if it is not \emph{ex ante} dominated by any lottery, and an SDS $f$ is \emph{ex ante efficient} if $f(R)$ is \emph{ex ante} efficient for every profile $R$. \emph{Ex ante} efficiency ensures that there is no lottery that weakly increases the expected utility of all voters and strictly for at least one voter. Hence, the impossibility theorem by \citet{BBEG16a} shows that no weakly strategyproof SDS on $\mathcal{R}^N$ satisfies mild efficiency constraints. 

\begin{theorem}[\citet{BBEG16a}]\label{thm:SDimpossibity}
    Assume $n\geq 4$ and $m\geq 4$. No anonymous and neutral SDS on $\mathcal{R}^N$ satisfies \emph{ex ante} efficiency and weak strategyproofness. 
\end{theorem}

\citet{BBEG16a} have shown this result by a computer-generated proof, which reasons over 47 (canonical) preference profiles. As it is very difficult for humans to verify the correctness of this 14-page proof, \citeauthor{BBEG16a} had the proof checked by the interactive theorem prover Isabelle/HOL. 
By contrast, we give a rather simple proof of this result in the appendix, which argues over only $13$ profiles ($10$ canonical profiles) and takes less than two pages. Its correctness can easily be verified by the avid reader.

\Cref{thm:SDimpossibity} crucially relies on \emph{ex ante} efficiency. Indeed, if \emph{ex ante} efficiency is replaced with \emph{ex post} efficiency, random serial dictatorship satisfies all the axioms \citep{ABBB15a}. This SDS randomly chooses an order over the voters and each voter in the sequence then acts as dictator, breaking the ties left by the previous dictators. 
This leads to the question whether there are also reasonable \emph{even-chance} SDSs on $\mathcal{R}^N$ that satisfy weak strategyproofness and \emph{ex post} efficiency. Unfortunately, it turns out that this is not the case: every such SDS can only randomize over the top-ranked alternatives of at most two voters. To make this formal, we say an SDS $f$ is \emph{dictatorial} if there is a voter $i\in N$ such that $f(R, T_i(R))=1$ for all profiles $R$, and \emph{bidictatorial} if there are two distinct voters $i,j\in N$ such that $f(R, T_i(R)\cup T_j(R))=1$ for all profiles $R$. Clearly, dictatorial and bidictatorial SDSs are undesirable as at most two voters can influence which alternatives are returned with positive probability.

\begin{restatable}{theorem}{bidictatorship}\label{thm:bidicatorship}
    Assume $m\geq 3$ and $n\geq 3$. Every \emph{ex post} efficient and weakly strategyproof even-chance SDS on $\mathcal{R}^N$ is dictatorial or bidictatorial. 
\end{restatable}
\begin{proof}[Proof Sketch]
    Let $f$ denote an even-chance SDS that is \emph{ex post} efficient and weakly strategyproof. The proof of this theorem is focusing on the decisive groups and weak dictators of $f$. To this end, we say that a voter $i$ is a weak dictator for $f$ if $f(R, T_i(R))>0$ for all profiles $R$ and a group of voters $G\subseteq N$ is decisive for $f$ if $f(R, T_i(R))=1$ for all voters $i\in G$ and profiles $R$ such that all voters in $G$ report the same preference relation in $R$.
    First, a result of \citet{BBL21b} implies that if a voter $i$ is not a weak dictator for $f$, then $N\setminus \{i\}$ is decisive for $f$. We next show that there are at least one and at most two weak dictators $i,j$ for $f$, so we get for every voter $h\not\in \{i,j\}$ that $N\setminus \{h\}$ is decisive. The last insight for our proof is a contraction lemma stating that if there are two decisive groups $G$ and $G'$ for $f$ such that $|G|=|G'|$ and $|G\cap G'|=|G|-1$, then $G\cap G'$ is also decisive for $f$. By applying this to our decisive groups, we infer that the set of weak dictators is decisive. Based on this observation, we finally show that $f$ is (bi)dictatorial. 
\end{proof}

\remark
All axioms are required for \Cref{thm:bidicatorship}. Every constant even-chance SDS only fails \emph{ex post} efficiency, the SDS that randomizes uniformly over the Pareto-optimal alternatives is \emph{ex post} efficient and not (bi)dictatorial. Finally, random serial dictatorship satisfies all axioms but it is not even-chance nor (bi)dictatorial.

\remark
\Cref{thm:bidicatorship} has interesting connections to known results. Firstly, based on much stronger strategyproofness notions, \citet{Feld80a} and \citet{BDS01a} show that all strategyproof and \emph{ex post} efficient even-chance SDSs are dictatorial or bidictatorial when the voters' preferences are strict. For instance, \citet{Feld80a} uses strong strategyproofness and his result is thus a corollary of the theorem by \citet{Gibb77a}. 
\Cref{thm:bidicatorship} demonstrates that a much weaker strategyproofness notion still allows to deduce the same result when allowing for weak preferences. Secondly, \Cref{thm:bidicatorship} is related to a result by \citet{BSS19a} who show that no set-valued voting rule on $\mathcal{R}^N$ satisfies anonymity, \emph{ex post} efficiency, and a strategyproofness notion due to \citet{Fish72a}. When interpreting even-chance SDSs as set-valued voting rules, weak strategyproofness is only slightly stronger than Fishburn-strategyproofness, but it yields the much more restrictive conclusion of (bi)dictatorial rules. Moreover, \Cref{thm:bidicatorship} is stronger than Corollary 2 of \citeauthor{BSS19a}.

\section{Conclusion}

In this paper, we study randomized voting rules, so-called social decision schemes (SDSs), with respect to weak strategyproofness. This strategyproofness notion only deems a manipulation successful if it increases the voter’s expected utility for all utility functions consistent with his ordinal preferences. We show that weak strategyproofness allows for some positive results. For example, in contrast to results on strong strategyproofness \citep{BLR23a}, there are Condorcet-consistent weakly strategyproof SDSs that are approximately \emph{ex post} efficient. 
We also explore the limitations of weak strategyproofness and show, e.g., that no even-chance SDS simultaneously satisfies weak strategyproofness, Condorcet-consistency, and \emph{ex post} efficiency when preferences are strict. Moreover, we prove much more severe impossibility theorems for weak preferences, highlighting a sharp contrast between strict and weak preferences. We refer to \Cref{tab:resultssummary} for a complete overview of our results.

Our work points to a number of interesting and challenging open questions: firstly, based on our results in \Cref{subsec:strictprefneg}, we conjecture that no weakly strategyproof SDS satisfies both Condorcet-consistency and \emph{ex post} efficiency. If this conjecture was true, this result would effectively unify several results analyzing the existence of strategyproof and Condorcet-consistent SDSs \citep[e.g.,][]{Lede21c,BLS22c,BrLe21a}. Secondly, our negative results for the case of weak preferences lead to the question of whether all weakly strategyproof and \emph{ex post} efficient SDSs on the domain of weak preferences only randomize over the top-ranked alternatives of the voters. 

\section*{Acknowledgements}

We thank Alexander Thole and Ren\'e Romen for helpful discussions. This work was supported by the Deutsche Forschungsgemeinschaft under grants BR 2312/11-2 and BR 2312/12-1, and by the NSF-CSIRO grant on "Fair Sequential Collective Decision-Making" (RG230833).

\clearpage
\appendix

\section{Omitted Proofs}

In this appendix, we present all proofs that have been omitted from the main body. For a better readability, we put the proof of every theorem into an own subsection. 

\subsection{Proof of \Cref{thm:profscoring}}

We start by discussing the full proof of \Cref{thm:profscoring}.

\profscoring*
\begin{proof}
    Let $f$ denote a score-based SDS and let $s$ denote its corresponding score function. Moreover, consider two strict preference profiles $R$ and $R'$ that differ only in the preference relation of a single voter $i$. For simplicity, we will assume that voter $i$'s preference relation is given by ${\succsim_i}=x_1,x_2,\dots, x_m$. Our goal is to show that $f(R')\not\succ_i f(R)$, which implies that $f$ is weakly strategyproof. To this end, we define the profiles $R^{1}, \dots, R^k$ by \emph{(i)} $R^1=R$, \emph{(ii)} $R^k=R'$, $R^{k}=(R^{k-1})^{i:yx}$ for some alternatives $x,y\in A$, and \emph{(iii)} each pair of alternatives is swapped at most once in the sequence. In particular, condition \emph{(iii)} ensures that if $x\succ_i y$ and $x\succ_i' y$, then $x\succ_i^k y$ for every profile $R^k$ in our sequence. Moreover, we define $s_\mathit{total}(R)=\sum_{x\in A} s(R,x)$ and $s_\mathit{total}(R')=\sum_{x\in A} s(R',x)$ as the total scores assigned to the alternatives in $R$ and $R'$, and consider a case distinction with respect to these parameters. \medskip

    \textbf{Case 1: ${s_\mathit{total}(R)<s_\mathit{total}(R')}$.} First, we assume that $s_\mathit{total}(R)<s_\mathit{total}(R')$. In particular, this means that $s(R,x)\neq\infty$ for all $x\in A$. In this case, we let $x^*$ denote the most preferred alternative of voter $i$ with $s(R,x)>0$; such an alternative exists because $\sum_{y\in A} s(R,y)>0$ by definition. Next, let $h$ denote the index such that $x^*=x_{h}$ and note that $s(R,x_{\ell})=0$ for all $\ell<h$. To prove the theorem in this case, we will first show that $s(R^j, x_\ell)=0$ for all $j\in \{1,\dots, k\}$ and $\ell<h$. We start by considering the alternative $x_1$ and two profiles $R^{j}$, $R^{j+1}$ in our sequence, and assume that $s(R^{j},x)=0$. If the swap from $R^j$ to $R^{j+1}$ does not involve $x_1$, we have that $s(R^{j+1}, x)=s(R^{j},x)=0$ due to the localizedness of $s$. On the other hand, if the swap from $R^j$ to $R^{j+1}$ involves $x_1$, it must have been weakened as the sequence $R^1,\dots, R^k$ swaps each pair of alternatives at most once. The monotonicity of $s$ thus shows that $s(R^{j},x_1)\geq s(R^{j+1},x_1)$, which implies that $s(R^{j+1},x_1)=0$. Since $s(R^1,x_1)=0$, we can repeatedly apply these arguments to infer that $s(R^j,x_1)=0$ for all profiles $R^j$ in our sequence.
    
    Next, we assume inductively that there is $\ell\in \{1,\dots, h-2\}$ such that $s(R^j,x_{\ell'})=0$ for all $j\in \{1,\dots, k\}$ and $\ell'\in \{1,\dots, \ell\}$. We will show again that that $s(R^j,x_{\ell+1})=0$ implies that $s(R^{j+1}, x_{\ell+1})=0$ for all $j\in \{1,\dots,k-1\}$. To this end, consider two arbitrary profiles $R^j$ and $R^{j+1}$ on our sequence and assume that $s(R^j,x_{\ell+1})=0$. If $x_{\ell+1}$ is not involved in the swap from $R^j$ to $R^{j+1}$, then $s(R^{j+1}, x_{\ell+1})=s(R^j,x_{\ell+1})=0$ by localizedness. Next, if $x_{\ell+1}$ is weakened from $R^j$ to $R^{j+1}$, then $s(R^{j+1}, x_{\ell+1})=s(R^j,x_{\ell+1})=0$ due to monotonicity. Finally, if $x_{\ell+1}$ is reinforced against an alternative $y$ when going from $R^j$ to $R^{j+1}$, $y$ must be in the set $\{x_1,\dots, x_\ell\}$. Since $s(R^{j},y)=s(R^{j+1},y)=0$ by the induction hypothesis and $s(R^j,x_{\ell+1})=0$ by assumption, we infer that $s(R^{j+1}, x_{\ell+1})=0$ by balancedness. Finally, by repeatedly applying this argument, we infer that $s(R^j, x_{\ell+1})=0$ for all $j\in \{1,\dots, k\}$ as $s(R^1, x_{\ell+1})=0$. This proves the induction step and hence also our auxiliary claim.
    
    Finally, we will show that $s(R,x^*)\geq s(R',x^*)$, which implies that $f(R,x^*)=\frac{s(R,x^*)}{s_\mathit{total}(R)}>\frac{s(R',x^*}{s_\mathit{total}(R')}=f(R',x^*)$. Since $f(R, x)=0=f(R',x)$ for all $x\in A$ with $x\succ_i x^*$, this shows that $f(R, U(\succsim_i,x^*))>f(R', U(\succsim_i,x^*))$, so voter $i$ cannot manipulate by deviating from $R$ to $R'$. To show that $s(R,x^*)\geq s(R',x^*)$, we consider again two profiles $R^{j}$ and $R^{j+1}$ in our sequence and assume that $s(R^j,x^*)\neq\infty$. If the swap from $R^j$ to $R^{j+1}$ does not involve $x^*$ or weakens $x^*$, we can derive that $s(R^j,x^*)\geq s(R^{j+1},x^*)$ analogously to the previous analysis. Finally, if $x^*$ is reinforced from $R^j$ to $R^{j+1}$, it is reinforced against an alternative $y$ with $y\succ_i x^*$. Since $s(R^j,y)=s(R^{j+1},y)=0$ and $s(R^j,y)\neq\infty$, balancedness implies that $s(R^j,x^*)=s(R^{j+1},x^*)$. By chaining these arguments and observing that $s(R,x^*)\neq\infty$ as $s_\mathit{total}(R)<\infty$, it follows that $s(R,x^*)\geq s(R',x^*)$.\medskip

    \textbf{Case 2: ${s_\mathit{total}(R)>s_\mathit{total}(R')}$.} For the second case, suppose that $s_\mathit{total}(R)>s_\mathit{total}(R')$, which implies that $s_\mathit{total}(R')<\infty$. In this case, let $x^*$ and $x'$ denote voter $i$'s least favorite alternatives in $R$ such that $s(R,x^*)>0$ and $s(R',x')>0$, respectively. First, if $x^*\succ_i x'$, voter $i$ cannot manipulate by deviating from $R$ to $R'$ as $f(R', x')=\frac{s(R',x')}{s_\mathit{total}(R')}>0$ and thus $f(R, {U(\succsim_i, x^*)})=1>f(R', {U(\succsim_i, x^*)})$. Next, let $h$ denote the index such that $x^*=x_h$, which means that $s(R,x_\ell)=s(R',x_\ell)=0$ for all $\ell\in \{h+1, \dots, m\}$. Our aim is to show that $s(R',x^*)\geq s(R,x^*)$ as this implies that voter $i$ cannot manipulate by deviating from $R$ to $R'$. In more detail, this means that $s(R,x^*)\neq\infty$ because $s_\mathit{total}(R)>s_\mathit{total}(R')$ implies that $s(R',x^*)\neq\infty$. In turn, we infer from $s(R',x^*)\geq s(R,x^*)$ and $s_\mathit{total}(R)>s_\mathit{total}(R')$ that $f(R,x^*)=\frac{s(R,x^*)}{s_\mathit{total}(R)}<\frac{s(R',x^*)}{s_\mathit{total}(R')}=f(R',x^*)$. Because $f(R,x)=f(R',x)=0$ for all $x\in A$ with $x^*\succ_ix$, this means that $f(R, {U(\succsim_i,x^*)\setminus \{x^*\}})=1- f(R, x^*)>1-f(R',x^*)=f(R', {U(\succsim_i,x^*)\setminus \{x^*\}})$, which demonstrates that voter $i$ cannot manipulate by deviating from $R$ to~$R'$. 

    To prove that $s(R',x^*)\geq s(R,x^*)$, we will first show that $s(R^j,x_\ell)=0$ for all profiles $R^j$ in our sequence and alternatives $x_\ell\in \{x_{h+1}, \dots, x_m\}$. To this end, we first consider alternative $x_m$ and two profiles $R^j$ and $R^{j+1}$ on our sequence. Since each pair of alternatives is swapped at most once along our sequence, $x_m$ is never weakened in our analysis, so it follows from localizedness and monotonicity that $s(R^{j+1},x_m)\geq s(R^j, x_m)$ for all $j\in \{1,\dots, k-1\}$. Finally, since $s(R^k, x_m)=0$, this implies that $s(R^j,x_m)=0$ for all profiles in our sequence. 
    
    Next, assume inductively that there is $\ell\in \{h+2,\dots, m\}$ such $s(R^j,x_{\ell'})=0$ for all $j\in \{1,\dots, k\}$ and alternatives $x_{\ell'}$ with $\ell'\in \{\ell,\dots, m\}$. We will show that the same holds for the alternative $x_{\ell-1}$. To this end consider two profiles $R^j$, $R^{j+1}$ in the sequence and suppose that $s(R^{j+1},x_{\ell-1})=0$. Now, if $x_{\ell-1}$ is not moved at all or reinforced when going from $R^{j}$ to $R^{j+1}$, then it follows from localizedness and monotonicity that $s(R^j,x_{\ell-1})=0$, too. On the other hand, if $x_{\ell-1}$ is weakened against another alternative $x$ when going from $R^{j}$ to $R^{j+1}$, then $x\in \{x_\ell,\dots,x_m\}$. Put differently, this means that we reinforce $x$ against $x_{\ell-1}$ in this step. 
    Since the induction hypothesis implies that $s(R^{j+1},x)=s(R^j,x)=0$, it follows from balancedness that $s(R^j,x_{\ell-1})=s(R^{j+1}, x_{\ell-1})=0$. Finally, since $s(R^k,x_{\ell-1})=0$, we conclude that $s(R^j,x_\ell)=0$ for all $j\in \{1,\dots, k\}$ and $\ell\in \{1,\dots, h\}$. 

    As the last step, we will show that $s(R',x^*)\geq s(R,x^*)$. To this end, we consider again two profiles $R^j$ and $R^{j+1}$ in our sequence from $R$ to $R'$. If $x^*$ is not swapped at all or reinforced when going from $R^{j}$ to $R^{j+1}$, it follows that $s(R^{j+1},x^*)\geq s(R^j,x^*)$ due to the localizedness and monotonicity of $s$. By contrast, if $x^*$ is weakened against another alternative $x$ when going from $R^{j}$ to $R^{j+1}$, then $x$ is in $\{x_{h+1},\dots, x_m\}$ as our sequence swaps each pair of alternatives at most once. Since $s(R^{j+1},x)=s(R^j,x)=0$ by our previous analysis, balancedness implies that $s(R^{j},x^*)=s(R^{j+1},x^*)$. By repeatedly applying this argument, we finally infer that $s(R',x^*)=s(R^k, x^*)\geq s(R^1,x^*)=s(R,x^*)$. \medskip

    \textbf{Case 3: ${s_\mathit{total}(R)=s_\mathit{total}(R')}<\infty$.}
    Next, we suppose that ${s_\mathit{total}(R)=s_\mathit{total}(R')}<\infty$. For this case, it is necessary to specify the sequence $R^1,\dots, R^k$ that transforms $R$ to $R'$. In particular, consider the following sequence: for deriving $R^{j+1}$ from $R^j$, we identify the most-preferred alternative $x$ according to $\succsim_i'$ that is not ranked at its correct position, and we reinforce $x$ against the alternative directly above it. Less formally, this results in the following sequence: starting at $R$, we first reinforce voter $i$'s favorite alternative $x_1'$ in $R'$ until we arrive at a profile where $x_1'$ is his top-ranked alternative; next, we reinforce voter $i$'s second-favorite alternative $x_2'$ in $R'$ until we arrive at a profile where $x_1'$ is his favorite alternative and $x_2'$ is his second-favorite alternative; next, we will repeat the same process with his third-most preferred alternative $x_3'$ in $R'$ and so on. It should be easy to see that this sequence of profiles $R^1,\dots, R^k$ indeed satisfies our three critera stated in the beginning. 
    
    We will next prove that $s_\mathit{total}(R)=s_\mathit{total}(R')<\infty$ implies that $s_\mathit{total}(R^j)<\infty$ for all profiles in this sequence. To this end, assume for contradiction that this is not the case and let $R^j$ denote the first profile in our sequence such that $s_\mathit{total}(R^j)=\infty$. In particular, we have that $s_\mathit{total}(R^{j-1})<\infty$. Now, let $x^*$ denote the alternative that is reinforced when going from $R^{j-1}$ to $R^j$. By localizedness and monotonicity, it must be that $s_\mathit{total}(R^j,x^*)=\infty$ because the score of every other alternative is weakly decreasing. Furthermore, by the definition of our sequence, we will keep reinforcing $x^*$ until it is at its correct position. Monotonicity implies for the corresponding steps that the score of $x^*$ has to remain $\infty$. Finally, since $x^*$ is now at its correct position, it will not be moved anymore, so localizedness implies that it score will always be $\infty$, i.e., it holds that $s(R^{j'},x^*)=\infty$ for all profiles $R^{j'}$ with $j'\in \{j,\dots, k\}$. However, this means that $s(R',x^*)=\infty$, which contradicts that $s_\mathit{total}(R')<\infty$.

    For the remainder of this case, we will hence assume that $s_\mathit{total}(R^j)<\infty$ for all $j\in \{1,\dots, k\}$. In particular, this means that there are no exceptions to balancedness, i.e., if we reinforce (or weaken) an alternative $x$ against another alternative $y$ when moving from a profile $R^j$ to $R^{j+1}$ and $s(R^j,x)=s(R^{j+1},x)$, then $s(R^j,y)=s(R^{j+1},y)$. Furthermore, if $s(R,x)=s(R',x)$ for all $x\in A$, then $f(R)=f(R')$ and voter $i$ cannot manipulate by deviating from $R$ to $R'$. Hence, suppose that there is an alternative $x\in A$ such that $s(R,x)\neq s(R',x)$ and let $x^*$ denote voter $i$'s favorite alternative in $R$ with $s(R,x^*)\neq s(R',x^*)$. Moreover, we let $h$ denote the index such that $x^*=x_h$, which means that $s(R,x_\ell)=s(R',x_\ell)$ for all $x_\ell\in \{x_1,\dots, x_{h-1}\}$. We will show that $s(R,x^*)>s(R',x^*)$ because then $f(R, U(\succsim_i, x^*))>f(R', U(\succsim_i,x^*))$. 
    
    To this end, we will first show show that $s(R^j, x_\ell)=s(R, x_\ell)$ for all profiles $R^j$ on our sequence and all alternatives $x_\ell$ with $x_\ell\in \{x_1, \dots, x_{h-1}\}$. We consider first the alternative $x_1$. Since this is voter $i$'s favorite alternative in $R$, this alternative is never reinforced in our sequence. Hence, it follows from localizedness and monotonicity that $s(R^j,x_1)\geq s(R^{j+1},x_1)$ for all $j$. Due to the assumption that $s(R^1,x_1)=s(R^k,x_1)$, we then derive that $s(R^j,x_1)=s(R,x_1)$ for all $j\in \{1,\dots, k\}$. Next, assume inductively that there is an index $\ell\in \{1,\dots, h-2\}$ such that $s(R^j,x)=s(R,x)$ for all $x\in \{x_1, \dots, x_\ell\}$ and all profiles $R^j$ in our sequence. We will show that the same holds for $x_{\ell+1}$. For this, we consider two profiles $R^{j}$ and $R^{j+1}$ in our sequence. If $x_{\ell+1}$ is not moved at all or weakened in this step, then $s(R^j, x_{\ell+1})\geq s(R^{j+1}, x_{\ell+1})$ due to localizedness and monotonicity. On the other hand, if $x_{\ell+1}$ is reinforced when going from $R^j$ to $R^{j+1}$, it must be reinforced against an alternative $x\in \{x_1,\dots, x_\ell\}$. Since $s(R^j,x)=s(R^{j+1},x)$ by our induction hypothesis, $s_\mathit{total}(R^j)<\infty$, and $s_\mathit{total}(R^{j+1})<\infty$, we infer from the balancedness of $s$ that $s(R^j,x_{\ell+1})=s(R^{j+1},x_{\ell+1})$. In summary, we have $s(R^j,x_{\ell+1})\geq s(R^{j+1},x_{\ell+1})$ for all $j\in \{1,\dots, k-1\}$ and since $s(R^1,x_{\ell+1})=s(R^k,x_{\ell+1})$, it follows again that all these inequalities are tight. This completes the induction step. 

    Finally, we will show that $s(R,x^*)\geq s(R',x^*)$, which implies that $s(R,x^*)>s(R',x^*)$ as $s(R,x^*)\neq s(R',x^*)$. To this end, consider again two profiles $R^j$ and $R^{j+1}$ in our sequence. If $x^*$ is not moved or weakened when going from $R^j$ to $R^{j+1}$, localizedness and monotonicity imply that $s(R^j, x^*)\geq s(R^{j+1},x^*)$. On the other hand, if $x^*$ is reinforced in this step, it is reinforced against an alternative $x\in \{x_1,\dots, x_{h-1}\}$. Since $s(R^j,x)=s(R^{j+1},x)$, $s_\mathit{total}(R^j)<\infty$, and $s_\mathit{total}(R^{j+1})<\infty$, balancedness shows that $s(R^j,x^*)=s(R^{j+1},x^*)$. This proves that $s(R^j,x^*)\geq s(R^{j+1},x^*)$ and inductively applying this argument then shows that $s(R,x^*)\geq s(R',x^*)$.\medskip

    \textbf{Case 4: ${s_\mathit{total}(R)=s_\mathit{total}(R')}=\infty$.}
    As the last case, we assume that ${s_\mathit{total}(R)=s_\mathit{total}(R')}=\infty$. This means that there are two alternatives $x^*$ and $x'$ such that $s(R,x^*)=\infty$ and $s(R',x')=\infty$. Moreover, as there is at most one alternative with a score of infinity of a profile, it follows that $f(R,x^*)=1$ and $f(R',x')=1$. Now, if $x^*=x'$, it is certainly no manipulation to deviate from $R$ to $R'$, so we assume that $x\neq x^*$. In this case, we will prove that $x^*\succ_i x'$, which shows that voter $i$ cannot manipulate by deviating from $R$ to $R'$. 

    For this case, we will again consider a specific sequence of profiles between $R$ and $R'$. To this end, let $U=\{x\in X\colon x\succ_i x^*\}$, $L=\{x\in X\colon x^*\succ_i x\}$, $U'=\{x\in X\colon x\succ_i' x^*\}$, and $L'=\{x\in X\colon x^*\succ_i' x\}$. Now, consider first the profile $R^1$ where voter $i$ reports the preference relation $\succsim_i^1$ defined by $u\succ_i^1 x^*\succ_i^1 \ell$ for all $u\in U$, $\ell\in L$, $u\succ_i^1 u'$ iff $u\succ_i' u'$ for all $u,u'\in U$, and $\ell\succ_i^1 \ell'$ if and only if $\ell\succ_i' \ell'$ for all $\ell,\ell'\in L$. Put less formally, we derive $R^1$ from $R$ by letting voter $i$ sort the alternatives in $U$ and $L$ according to $\succsim_i'$. Now, it should be clear that we can transform $R$ into $R^1$ by pairwise swaps that do not involve $x^*$, so we can infer from a repeated application of localizedness that $s(R^1,x^*)=\infty$. 
    Next, let $R^2$ denote the profile derived from $R^1$ by reinforcing $x^*$ against all alternatives in $U\setminus U'$. We observe that these alternatives are ranked directly ahead of $x^*$ in $\succsim_i^1$ since $U\setminus U'\subseteq L'$ and $u\succ_i' x^*\succ_i' \ell$ for all $u\in U'$, $\ell\in L'$. Hence, we only need to reinforce $x^*$ in this step, so monotonicity implies that $s(R^2,x^*)=\infty$. 
    We note that, in $R^2$, it holds that $x^*\succ_i^2 x$ for all $x\in L\cup L'$. Hence, we now reorder the alternatives $x\in L\cup L'$ in voter $i$'s preference relation according to $\succsim_i'$. Localizedness implies again that $s(R^3,x^*)=\infty$. In particular, this means that $s(R^2,x)<\infty$ for all other alternatives $x\in A\setminus \{x^*\}$ because there can be at most one alternative with a score of infinity. 

    Finally, we can now transform $R^3$ into $R'$ by reinforcing the alternatives in $L\cap U'$ against alternatives in $(U\cap U')\cup \{x^*\}$. In particular, we can transform $R^3$ to $R'$ by pairwise swaps such that no swaps reinforce an alternative in $U$. Hence, monotonicity and localizedness imply that $s(R',x)\leq s(R^3,x)<\infty$ for all $x\in U$. Because we assume that there is an alternative $x'\neq x^*$ with $s(R',x')=\infty$, it follows that that $x'\in L$. Thus, it is again no manipulation for voter $i$ to deviate from $R$ to $R'$. 
\end{proof}

\subsection{Proof of \Cref{thm:topsonly}}\label{subsec:topsonly}

We next will show the second claim of \Cref{thm:topsonly}. To this end, we recall that $s_{\mathit{P}}(R,x)$ denotes the plurality score of alternative $x$ in the profile $R$, i.e., the number of voters who prefer $x$ the most in $R$. Furthermore, for the proof of this claim, we will view even-chance SDSs as functions that return sets of alternatives instead of lotteries. Since even-chance lotteries randomize uniformly over a set of alternatives, this representation is without loss of information. Furthermore, the voters' preferences over lotteries turn into the following preferences over sets of alternatives: a voter $i$ prefers a set $X$ to another set $Y$, denoted by $X\succsim_i Y$, if and only if $\frac{|X\cap U(\succsim_i, x)|}{|X|}\geq \frac{|Y\cap U(\succsim_i, x)|}{|Y|}$ for all $x\in A$. The definition of weak strategyproofness does not change: an even-chance SDS $f$ (interpreted as a set-valued voting rule) is weakly strategyproof if $f(R')\not\succ_i f(R)$ for all profiles $R,R'$ and voters $i\in N$ with ${\succsim_j}={\succsim_j'}$. In this context, weak strategyproofness is equivalent to the well-known notion of even-chance strategyproofness \citep[e.g.,][]{Gard79a,BDS01a,BSS19a}.

\topsonly*
\begin{proof}
    We will only show the second claim here as first statement has been proven in the main body.\medskip
    
    $\impliedby$ We start by showing that every parameterized omninomination rule satisfies our the given axioms. To this end, let $f$ denote a parameterized omninomination rule and let $\theta_1\in \{\lceil\frac{n+1}{2}\rceil,\dots, n+1\}$, $\theta_2\in \{0,\dots, m-1\}$ denote its parameters. It should be clear that $f$ is by definition anonymous, neutral, even-chance, and tops-only. It thus only remains to show that it is weakly strategyproof. For showing this, let $R$ and $R'$ denote two strict preference profiles and $i$ a voter such that ${\succsim_j}={\succsim_j'}$ for all $j\in N\setminus \{i\}$. Moreover, let $x^*$ denote voter $i$'s favorite alternative in $R$ and $y^*$ his favorite alternative in $R'$. We first note that, if $x^*=y^*$, then $f(R)=f(R')$ as $f$ is tops-only. We will hence assume that $x^*\neq y^*$ and consider a case distinction for this. 
     
    First, assume that $s_{\mathit{P}}(R,x)\geq\theta_1$ for some alternative $x$, which means that $f(R,x)=1$. Now, if $x=x^*$, voter $i$ cannot manipulate as his favorite alternative is chosen with probability $1$. On the other hand, if voter $i$ prefers $x$ not the most, then $s_{\mathit{P}}(R',x)\geq\theta_1$ and $f(R',x)=1$ and voter $i$ cannot manipulate by deviating from $R$ to $R'$. 
    
    Next, assume that $s_{\mathit{P}}(R,x)<\theta_1$ for all $x\in A$ and $|\textit{OMNI}(R)|\leq \theta_2$. In this case, $f(R)=f^O(R)$ and hence $f(R,x^*)=\frac{1}{|\textit{OMNI}(R)|}$. First, if there is an alternative $x\in A$ with $s_{\mathit{P}}(R',x)\geq \theta_1$, then $f(R',x)=1$ and $f(R',x^*)=0$ as voter $i$ cannot increase the number of voters that top-rank his favorite alternative and thus $x\neq x^*$. Next, if $f(R',x^*)<\frac{1}{|\textit{OMNI}(R)|}$, it is no successful manipulation to deviate to $R'$. We infer from this that $|\textit{OMNI}(R')|\leq |\textit{OMNI}(R)|$. Furthermore, if $\textit{OMNI}(R')=\textit{OMNI}(R)$, then $f(R)=f(R')$ and voter $i$ cannot manipulate. On the other hand, if $\textit{OMNI}(R')\neq \textit{OMNI}(R)$ and $|\textit{OMNI}(R')|\leq |\textit{OMNI}(R)|$, it must be the case that $x^*\not\in \textit{OMNI}(R')$ as this is the only alternative that voter $i$ can potentially remove from this set. However, it then follows that $f(R',x^*)=0$, so $f(R')\not\succ_i f(R)$.
    
    Lastly, assume that $s_{\mathit{P}}(R,x)<\theta_1$ for all $x\in A$ and $|\textit{OMNI}(R)|> \theta_2$. In this case, we have that $f(R,x)=\frac{1}{m}$ for all $x\in A$. If there is an alternative $x$ with $s_{\mathit{P}}(R',x)\geq \theta_1$, then $f(R',x)=1$. However, $x\neq x^*$ as voter $i$ cannot increase the number of voters that improve his favorite alternative. Thus, $f(R',x^*)=0$ and no manipulation is possible. Similarly, if $\max_{x\in A} s_{\mathit{P}}(R',x)<\theta_1$ and $|\textit{OMNI}(R')|\leq \theta_2$, then $x^*\not\in \textit{OMNI}(R')$ and $f(R',x^*)=0$. Finally, if $\max_{x\in A} s_{\mathit{P}}(R',x)<\theta_1$ and $|\textit{OMNI}(R')|> \theta_2$, then $f(R)=f(R')$ and again no manipulation is possible. Hence, we conclude that $f$ is weakly strategyproof.\medskip

    $\implies$ Let $f$ denote a tops-only even-chance SDS that satisfies weak strategyproofness, anonymity, and neutrality. We will prove in four steps that $f$ is a parameterized omninomination rule. To this end, we define $N(R,x)$ as the set of voters that prefer $x$ the most and note that $s_{\mathit{P}}(R,x)=|N(R,x)|$. Moreover, we treat $f$ in this proof as a set-valued voting rule, as described in the beginning of this section. Now, we will first prove that if $s_{\mathit{P}}(R,x)\geq s_{\mathit{P}}(R,y)$ and $y\in f(R)$, then $x\in f(R)$, too. Next, we assume that there is a profile $R$ and a voter $i$ such that $x\not\in f(R)$ even though $N(R,x)\neq\emptyset$. If such a profile does not exist, we can set $\theta_1=n+1$ and proceed to the last step. We then infer that there is a value $\theta_1$ such that $f(R)=\{x\}$ whenever $\theta_1$ or more voters top-rank $x$. As third step, we then prove that $\textit{OMNI}(R)\subseteq f(R)$ for all profiles $R$ with $\max_{y\in A} s_{\mathit{P}}(R,y)<\theta_1$. Hence, a single alternative is chosen if it is top-ranked by at least $\theta_1$ voters, and $f$ otherwise chooses $\textit{OMNI}(R)$ or $A$. For the latter, we note that if $\textit{OMNI}(R)\subsetneq f(R)$, then there is an alternative $x\in f(R)$ with $s_{\mathit{P}}(R,x)=0$, so the first step implies that all alternatives need to be chosen. Finally, we show that, if $f(R)=\textit{OMNI}(R)$ for some profile $R$ with $\textit{OMNI}(R)\neq A$, then the same holds for all profiles $R'$ with $|\textit{OMNI}(R')|<|\textit{OMNI}(R)|$ and $\max_{x\in A} s_{\mathit{P}}(R',x)<\theta_1$. This implies that there is a parameter $\theta_2\in \{0,\dots,m-1\}$ such that $f(R)=\textit{OMNI}(R)$ if $\max_{y\in A} s_{\mathit{P}}(R,y)<\theta_1$ and $|\textit{OMNI}(R)|\leq \theta_2$ and $f(R,x)=\frac{1}{m}$ for all $x\in A$ if $\max_{y\in A} s_{\mathit{P}}(R,y)<\theta_1$ and $|\textit{OMNI}(R)|> \theta_2$.\medskip

    \textbf{Step 1:} Consider a profile $R$ and two alternatives $x,y\in A$ such that $s_{\mathit{P}}(R,x)\geq s_{\mathit{P}}(R,y)$. We will show that $y\in f(R)$ implies that $x\in f(R)$. If $s_{\mathit{P}}(R,x)= s_{\mathit{P}}(R,y)$, this follows immediately from tops-onlyness, anonymity, and neutrality, because $x$ and $y$ are symmetric if we only consider the voters' favorite alternatives. Hence, assume that $s_{\mathit{P}}(R,x)>s_{\mathit{P}}(R,y)$. In this case, we let the voters $i\in N(R,x)$ one after another change their preference relation to one where they top-rank $y$. This results in a sequence of profiles $R^1,R^2,\dots$ such that $R^1=R$ and $s_{\mathit{P}}(R^i,x)=s_{\mathit{P}}(R^{i+1},x)+1$ and $s_{\mathit{P}}(R^i,y)=s_{\mathit{P}}(R^i,y)-1$ for all profiles $R^i$. If $x\not\in f(R^i)$ for some profile $R^i$, then it follows from Claim 1 of this theorem that $f(R^i)=f(R^{i+1})$ as it is impossible that the probability of $x$ decreases. Because $x\not\in f(R^1)$, this means that $f(R^1)=f(R^i)$ for all profiles $R^i$ in our sequence. Now, if there is a profile $R^i$ such that $s_{\mathit{P}}(R^i,x)=s_{\mathit{P}}(R^i,y)$ (which happens if $s_{\mathit{P}}(R,x)-s_{\mathit{P}}(R,y)$ is a multiple of $2$), it follows again from anonymity, neutrality, and tops-onlyness that $x\in f(R^i)$ if and only if $y\in f(R^i)$. However, this contradicts that $f(R)=f(R^i)$. On the other hand, if no such profile exists in our sequence, then there is an index $i$ such that $s_{\mathit{P}}(R^i,x)=s_{\mathit{P}}(R^i,y)+1$ and $s_{\mathit{P}}(R^{i+1},x)=s_{\mathit{P}}(R^{i+1},y)-1$. This implies that $s_{\mathit{P}}(R^i,x)=s_{\mathit{P}}(R^{i+1},y)$ because $s_{\mathit{P}}(R^i,x)-1=s_{\mathit{P}}(R^{i+1},x)=s_{\mathit{P}}(R^{i+1},y)-1$, and that $s_{\mathit{P}}(R^i,y)=s_{\mathit{P}}(R^{i+1},x)$ because $s_{\mathit{P}}(R^i,y)+1=s_{\mathit{P}}(R^i,x)=s_{\mathit{P}}(R^{i+1},x)+1$. Hence, by anonymity, neutrality, and tops-onlyness, we have that $y\in f(R^i)$ if and only if $x\in f(R^{i+1})$, which conflicts again with the observation that $f(R)=f(R^i)=f(R^{i+1})$. Since we exhausted all cases, we conclude that $x\in f(R)$ if  $s_{\mathit{P}}(R,x)\geq s_{\mathit{P}}(R,y)$ and $y\in f(R)$.\medskip

    \textbf{Step 2:} Next, we will show that there is a parameter $\theta_1\in \{\lceil\frac{n+1}{2}\rceil,\dots, n+1\}$ such that $f(R)=\{x\}$ whenever $x$ is top-ranked by at least $\theta_1$ voters. To this end, we assume that there is a profile $R$ such that $\textit{OMNI}(R)\not\subseteq f(R)$. Otherwise, we can simply set $\theta_1=n+1$. Without loss of generality, we order the alternatives according to the plurality score, i.e., we assume that $s_{\mathit{P}}(R,x_1)\geq s_{\mathit{P}}(R,x_2)\geq \dots\geq s_{\mathit{P}}(R,x_m)$. By Step 1, there is an index $i$ such that $f(R)=\{x_1,\dots, x_i\}$. We note that this also means that $s_{\mathit{P}}(R,x_i)>s_{\mathit{P}}(R,x_{i+1})\geq 1$, where the last inequality follows as there is an alternative $x\in \textit{OMNI}(R)\setminus f(R)$. Now, consider the profile $R'$ that is derived from $R$ by making $x_1$ into the favorite alternative of a voter in $N(R,x_i)$. First, if $x_i\not\in f(R')$, then $f(R')\subsetneq f(R)$ since $s_{\mathit{P}}(R',x_i)\geq s_{\mathit{P}}(R',x_j)$ for all $j>i$. On the other hand, if $x_i\in f(R')$, it must hold that $f(R')=f(R)$. Otherwise, Step 1 implies that $f(R)\subseteq f(R')$, which means that the probability of $x_1$ in $R'$ is less than in $R$. However, this conflicts with Claim 1 of this theorem because the probability of $x_1$ is not allowed to decrease when it is made into a voter's favorite alternative. In summary, it holds that $f(R')\subseteq f(R)$. Finally, we can repeat this process until $x_i$ is no longer chosen because we will eventually arrive at a profile $R''$ such that $s_{\mathit{P}}(R'',x_i)=s_{\mathit{P}}(R'',x_{i+1})$. Assuming that $x_i$ is chosen in all profiles before $R''$, our previous analysis shows that $f(R'')\subseteq f(R)$. However, if $x_i\in f(R'')$, Step 1 implies that $x_{i+1}\in f(R'')$, too. Hence, $x_i\not\in f(R'')$ and our process will indeed find a profile $\bar R$ such that $f(\bar R)\subsetneq f(R)$.

    Next, we can repeat the process in the previous paragraph until we arrive at a profile $\hat R$ with $f(\hat R)=\{x\}\subsetneq \textit{OMNI}(\hat R)$ by repeatedly removing the chosen alternative with minimal plurality score from the choice set. Since such a profile exists, we can define $R^*$ and $x^*$ as the profile-alternative pair that minimizes $s_{\mathit{P}}(R^*,x^*)$ subject to $f(R^*)=\{x^*\}$. We then define $\theta_1=s_{\mathit{P}}(R^*,x^*)$ and we will next show that $f(R)=\{x\}$ for all profiles $R$ and alternatives $x$ such that $s_\mathit{P}(R,x)\geq \theta_1$. To this end, consider first a profile $R$ where $x^*$ is top-ranked by at least $\theta_1$ voters. By anonymity, we can assume that $N(R^*,x^*)\subseteq N(R,x^*)$ as we can permute the voters without changing the outcome. Hence, we can derive the profile $R$ from $R^*$ by only changing the favorite alternatives of the voters $i\in N\setminus N(R^*,x^*)$. Since these alternatives are not chosen, Claim 1 of this theorem implies that the outcome is not allowed to change. Hence, we have that $f(R)=\{x^*\}$ for all profiles $R$ in which alternative $x^*$ is top-ranked by at least $\theta_1$ voters. Finally, neutrality generalizes this insight to all alternatives. This then implies that $\theta_1\geq \lceil \frac{n+1}{2}\rceil$; otherwise there can be multiple alternatives that are top-ranked by $\theta_1$ voters.\medskip

    \textbf{Step 3:} We will next show that, if $\max_{x\in A} s_{\mathit{P}}(R,x)<\theta_1$, then all top-ranked alternatives must be chosen. To this end, we assume for contradiction that there is a profile $R$ such that $\max_{x\in A} s_{\mathit{P}}(R,x)<\theta_1$ but $\textit{OMNI}(R)\not\subseteq f(R)$. By Step 2, we then infer that $\theta_1<n$ and that $|f(R)|>1$ as there is no profile $R'$ and alternative $x$ such that $f(R')=\{x\}$ and $s_{\mathit{P}}(R',x)<\theta_1$. Next, we again order the alternatives according to their plurality score, i.e., we assume that $s_{\mathit{P}}(R,x_1)\geq s_{\mathit{P}}(R,x_2)\geq\dots\geq s_{\mathit{P}}(R,x_m)$. Step 1 shows that there is an index $i$ such that $f(R)=\{x_1,\dots, x_{i}\}$ and $s_{\mathit{P}}(R,x_i)>s_{\mathit{P}}(R,x_{i+1})$. Now, let $N^-(R)=\{i\in N\colon T_i(R)\not\subseteq f(R)\}$ and note that $N^-(R)\neq\emptyset$ by assumption. 
    The central insight for this step is that $s_{\mathit{P}}(R,x_1)+|N^-(R)|<\theta_1$. Otherwise (i.e., if $s_{\mathit{P}}(R,x_1)+|N^-(R)|\geq \theta_1$), the voters in $N^-(R)$ can one after another deviate to a preference relation where $x_1$ is their favorite alternative. Claim 1 of this theorem implies for every step that the outcome is not allowed to change, so it holds for the final profile $R'$ that $f(R)=f(R')$. However, $x_1$ is top-ranked by at least $\theta_1$ voters in $R'$. Hence, Step 2 implies that $f(R')=\{x_1\}$. These two observations contradict each other as $|f(R)|>1$, so we conclude that $s_{\mathit{P}}(R,x_1)+|N^-(R)|< \theta_1$. 

    We will next focus on the case that $s_{\mathit{P}}(R,x_1)+|N^-(R)|< \theta_1$. To this end, let $R'$ denote the profile derived from $R$ by letting making $x_1$ into the favorite alternative of a voter in $N(R,x_2)$. First, we observe that $s_{\mathit{P}}(R',x_1)\leq s_{\mathit{P}}(R,x_1)+|N^-(R)|<\theta_1$, so $|f(R')|>1$. Furthermore, $s_{\mathit{P}}(R',x_2)\geq s_{\mathit{P}}(R',x_{i+1})$ because $s_{\mathit{P}}(R,x_2)> s_{\mathit{P}}(R,x_{i+1})$. Hence, $f(R)\subsetneq f(R')$ if $x_{i+1}\in f(R')$ by Step 1. However, this means that the probability of $x_1$ decreases when making it into the favorite alternative of a voter. This conflicts with Claim 1, so $x_{i+1}\not\in f(R)$ and $f(R')\subseteq f(R)$. In particular, this shows that $N^-(R)\subseteq N^-(R')$. Moreover, it holds by construction that $s_{\mathit{P}}(R',x_1)=s_{\mathit{P}}(R,x_1)+1$ and thus $s_{\mathit{P}}(R',x_1)+|N^-(R')|>s_{\mathit{P}}(R,x_1)+|N^-(R)|$. Now, if $s_{\mathit{P}}(R',x_1)+|N^-(R')|\geq \theta_1$, we derive the same contradiction as in the last paragraph. On the other hand, if $s_{\mathit{P}}(R',x_1)+|N^-(R')|<\theta_1$, we can set $R=R'$, redefine $x_1, x_2, \dots, x_m$, and repeat this step. Since $\theta_1<n$, we will eventually arrive at a profile $R''$ such that $s_{\mathit{P}}(R'',x_1)+|N^-(R'')|\geq \theta_1$ and $|f(R'')|>1$. For this profile, we have the same contradiction as in the last paragraph, thus showing that the assumption that $\textit{OMNI}(R)\not\subseteq f(R)$ is wrong.\medskip

    \textbf{Step 4:} Steps 2 and 3 imply that $f(R)=\{x\}$ whenever $x$ is top-ranked by at least $\theta_1$ voters in $R$ and $\textit{OMNI}(R)\subseteq f(R)$ if no such alternative exists. It thus remains to specify the second parameter $\theta_2$ of $f$ to complete the proof. We observe for this that, for all profiles $R$ with $\max_{x\in A} s_{\mathit{P}}(R,x)<\theta_1$, it holds by Steps 1 and 3 that $f(R)=\textit{OMNI}(R)$ or $f(R)=A$. First, if $\textit{OMNI}(R)=A$, both of these cases coincide. To derive the parameter $\theta_2$, we will show that if $f(R)=\textit{OMNI}(R)$ for some profile $R$ with $|\textit{OMNI}(R)|<m$, then $f(R')=\textit{OMNI}(R')$ for all profiles $R'$ with $|\textit{OMNI}(R')|\leq |\textit{OMNI}(R)|$ and $\max_{x\in A} s_{\mathit{P}}(R',x)<\theta_1$. We can then define $\theta_2$ by $\theta_2=|\textit{OMNI}(R^*)|$, where $R^*$ is a profile that maximizes $|\textit{OMNI}(R^*)|$ subject to $f(R^*)=\textit{OMNI}(R^*)$ and $f(R^*)\neq A$. (If no such profile exists, we can simply define $\theta_2=0$.) By our auxiliary claim, the previous insights, and the definition of $\theta_2$, it holds that $f(R)=\textit{OMNI}(R)$ if $\max_{x\in A} s_{\mathit{P}}(R,x)<\theta_1$ and $|\textit{OMNI}(R)|\leq\theta_2$ and that $f(R)=A$ if $\max_{x\in A} s_{\mathit{P}}(R,x)<\theta_1$ and $|\textit{OMNI}(R)|>\theta_2$, thus showing that $f$ is indeed a parameterized omninomination rule.

    It remains to prove the claim that if $f(R)=\textit{OMNI}(R)$ for some profile $R$ with $\textit{OMNI}(R)\neq A$, then $f(R')=\textit{OMNI}(R')$ for all profiles $R'$ with $|\textit{OMNI}(R')|\leq |\textit{OMNI}(R)|$ and $\max_{x\in A} s_{\mathit{P}}(R',x)<\theta_1$. First, if $|\textit{OMNI}(R)|=1$, this claim follows directly from neutrality and tops-onlyness. The reason for this is that $|\textit{OMNI}(R)|=1$ implies that all voters top-rank the same alternative, and all such profiles are symmetric to each other when restricting the attention to the voters' favorite alternatives. Hence, we assume that $|f(R)|=|\textit{OMNI}(R)|>1$, which implies that $\max_{x\in A} s_{\mathit{P}}(R,x)<\theta_1$. 
    
    We will first show that $f(R')=\textit{OMNI}(R')$ for all profiles $R'$ with  $\max_{x\in A} s_{\mathit{P}}(R,x)<\theta_1$ and $\textit{OMNI}(R)=\textit{OMNI}(R')$. In this case, let $R^1,\dots, R^k$ denote a sequence of profiles such that \emph{(i)} $R^1=R$ and \emph{(ii)} $R^{i+1}$ is derived from $R^i$ by identifying two alternatives $x,y$ with $s_{\mathit{P}}(R^i,x)>s_{\mathit{P}}(R',x)$ and $s_{\mathit{P}}(R^i, y)<s_{\mathit{P}}(R',y)$ and making $y$ into the favorite alternative of a voter in $N(R^i,x)$. We first note that, unless $s_{\mathit{P}}(R^i, z)=s_{\mathit{P}}(R',z)$ for all $z\in A$, such alternatives $x,y$ are guaranteed to exist as $\sum_{z\in A} s_{\mathit{P}}(R^i,z)=\sum_{z\in A} s_{\mathit{P}}(R',z)$. Hence, our sequence terminates in a profile $R^k$ such that $s_{\mathit{P}}(\bar R^k,z)=s_{\mathit{P}}(R',z)$ for all $z\in A$. Anonymity then shows that $f(R')=f(R^k)$, so it suffices to prove that $f(R^k)=f(R)=\textit{OMNI}(R)$. For this, we observer that $\min(s_{\mathit{P}}(R,z), s_{\mathit{P}}(R',z))\leq s_{\mathit{P}}(R^i,z)\leq \max(s_{\mathit{P}}(R,z), s_{\mathit{P}}(R',z))$ for all profiles $R^i$ in our sequence and alternatives $z\in A$. This implies that $\textit{OMNI}(R^i)=\textit{OMNI}(R)$ and that $\max_{z\in A} s_{\mathit{P}}(R^i,z)<\theta_1$. By Step 3, it thus holds that $\textit{OMNI}(R^i)\subseteq f(R^i)$ for all profiles in our sequence. Finally, if $f(R^i)=\textit{OMNI}(R^i)$ for some profile $R^i$, then $f(R^{i+1})=\textit{OMNI}(R^{i+1})$. In more detail, if $f(R^{i+1})\neq\textit{OMNI}(R^{i+1})$, it follows that $f(R^{i+1})=A$. Next, let $y$ denote the alternative such that $s_{\mathit{P}}(R^{i+1},y)=s_{\mathit{P}}(R^i,y)+1$. Since $\textit{OMNI}(R^{i})=\textit{OMNI}(R)\neq A$, the probability of $y$ decreases when going from $R^{i}$ to $R^{i+1}$ even though $N(R^i,y)\subsetneq N(R^{i+1},y)$. this contradicts Claim 1 of this theorem, thus showing that the assumption that $f(R^{i+1})\neq \textit{OMNI}(R^{i+1})$ is wrong. Finally, since $f(R^1)=\textit{OMNI}(R^1)$, we conclude that $f(R^k)=\textit{OMNI}(R^k)=\textit{OMNI}(R)$. This proves that $f(R')=f(R^k)=\textit{OMNI}(R')$. 
    
    By the discussion in the last paragraph, it holds that $f(R')=\textit{OMNI}(R')$ for all profiles $R'$ with $\textit{OMNI}(R')=\textit{OMNI}(R)$ and $\max_{z\in A} s_{\mathit{P}}(R',z)<\theta_1$. Moreover, by neutrality, this argument generalizes to all profiles $R''$ with $|\textit{OMNI}(R'')|=|\textit{OMNI}(R)|$ and $\max_{z\in A} s_{\mathit{P}}(R'',z)<\theta_1$. Next, we assume that $R'$ is a profile such that $|\textit{OMNI}(R')|=|\textit{OMNI}(R)|-1$ and $\max_{z\in A} s_{\mathit{P}}(R',z)<\theta_1$, and we will prove that $f(R')=\textit{OMNI}(R')$. By repeatedly applying this argument (and the insights of the last paragraph), it then follows that $f(\hat R)=\textit{OMNI}(\hat R)$ for all profiles $\hat R$ with $|\textit{OMNI}(\hat R)|\leq |\textit{OMNI}(R)|$ and $\max_{z\in A} s_{\mathit{P}}(\hat R,z)<\theta_1$. To show that $f(R')=\textit{OMNI}(R')$ for the considered profile $R'$, let $x$ and $y$ denote an alternatives with $N(R',x)=\emptyset$ and $|N(R',y)|>1$. Such alternatives exist as $|\textit{OMNI}(R')|<| \textit{OMNI}(R)|$ is otherwise impossible. Since $\max_{z\in A} s_{\mathit{P}}(R',z)<\theta_1$, we have that $\textit{OMNI}(R')\subseteq f(R')$ by Step 3, and Step 1 implies in turn that $f(R')=A$ if $f(R')\neq \textit{OMNI}(R')$. Now, let $R''$ denote the profile derived from $R'$ by letting a voter who top-ranks $y$ make $x$ into his favorite alternative. By construction, we have that $|\textit{OMNI}(R'')|=|\textit{OMNI}(R)|$ and $\max_{z\in A} s_{\mathit{P}}(R'',z)\leq \max_{z\in A} s_{\mathit{P}}(R',z)<\theta_1$. Hence, we conclude that $f(R'')=\textit{OMNI}(R'')$. However, since $|\textit{OMNI}(R)|<m$, this means that the probability of $y$ increases when going from $R'$ to $R''$, which conflicts with Claim 1 of this theorem. Hence, $f(R'')=\textit{OMNI}(R'')$ implies that $f(R')=\textit{OMNI}(R')$, which completes the proof of this step. 
\end{proof}

\subsection{Proof of \Cref{thm:EvenchanceCondImp}}\label{subsec:evenchancCondImp}

In this section, we provide the proof of \Cref{thm:EvenchanceCondImp} and moreover show a variant of this result for the case that $n$ is even. To this end, we will interpret even-chance SDSs as set-valued voting rules (see \Cref{subsec:topsonly} for details). Befoer showing \Cref{thm:EvenchanceCondImp}, we will prove two auxiliary lemmas that analyze when an even-chance SDS that satisfies weak strategyproofness, \emph{ex post} efficiency, and Condorcet-consistency is allowed to return choice sets of size $1$ or $2$. In more detail, we show next that, if the number of voters $n$ is odd, every weakly strategyproof and Condorcet-consistent SDS chooses a single winner if and only if it is the Condorcet winner. We note for the subsequent lemma that the even-chance condition is not required here.

\begin{lemma}\label{lem:sCC}
	Assume that the number of voters $n$ is odd and let $f$ denote an SDS on $\mathcal{L}^N$ that satisfies Condorcet-consistent and weak strategyproofness. It holds for all preference profiles $R$ that $f(R,x)=1$ if and only if $x$ is the Condorcet winner in $R$. 
\end{lemma}
\begin{proof}
	Let $f$ denote a Condorcet-consistent and weakly strategyproof even-chance SDS and assume that the number of voters $n$ is odd. Since $f$ is Condorcet-consistent, it holds by definition that $f(R,x)=1$ if $x$ is the Condorcet winner in $R$. We hence focus on the converse direction and assume for contradiction that there is a profile $R$ and an alternative $x$ such that $f(R,x)=1$ even though $x$ is not the Condorcet winner in $R$. Since the number of voters $n$ is odd and $x$ is not the Condorcet winner in $R$, there is another alternative $y$ such that $y\succ_M x$. Moreover, let $I\subseteq N$ denote a set of voters such that $y\succ_i x$ for all $i\in I$ and $|I|=\frac{n+1}{2}$.
	
	We will now start to modify the profile $R$. In particular, we first let the voters $i\in N\setminus I$ deviate one after another to a preference relation where $x$ is their most preferred alternative and $y$ their second-most preferred one. By weak strategyproofness, it follows that, if $x$ is chosen with probability $1$ before this manipulation, $x$ also needs to be with probability $1$ after the manipulation; otherwise the deviating voter can manipulate by undoing this step. Hence, it holds for the profile $R'$ derived by this process that $f(R',x)=1$. 
	
	Finally, we let a voter $i^*\in I$ deviate to a preference relation where $y$ is top-ranked. In the resulting profile $R''$, alternative $y$ is the Condorcet winner, so $f(R'',y)=1$ by Condorcet-consistency. In more detail, in $R''$, the $\frac{n-1}{2}$ voters in $N\setminus I$ all prefer only $x$ to $y$ and the voter $i^*$ top-ranks $y$. Hence, $y\succ_M'' z$ for every alternative $z\in A\setminus \{x,y\}$. On the other hand, all voters in $I$ still prefer $y$ to $x$, so $y\succ_M'' x$. Finally, the observations that $f(R'',y)=1$ and $f(R',x)=1$ conflict with weak strategyproofness as voter $i^*$ prefers $y$ to $x$ in $R'$. This is the desired contradiction, so the assumption that $f(R,x)=1$ even though $x$ is not the Condorcet winner in $R$ must be wrong. 
\end{proof}

Next, we will show that an even-chance SDS $f$ that satisfies weak strategyproofness, Condorcet-consistency, and \emph{ex post} efficiency can only choose a set of size $2$ if the two chosen alternatives are in a majority tie. In particular, if $n$ is odd, this means that such an SDS can never choose a choice set of size $2$ as majority ties are impossible.

\begin{lemma}\label{lem:no2}
	Let $f$ denote an even-chance SDS on $\mathcal{L}^N$ that satisfies weak strategyproofness, \emph{ex post} efficiency, and Condorcet-consistency. It holds for all profiles $R$ and distinct alternatives $x,y\in A$ that $f(R)=\{x,y\}$ implies $x\sim_M y$. 
\end{lemma}
\begin{proof}
	Let $f$ denote an even-chance SDS that satisfies all given axioms and assume for contradiction that there is a profile $R$ and two distinct alternatives $x,y\in A$ such that $f(R)=\{x,y\}$ and not $x\sim_M y$. Without loss of generality, we suppose that $x\succ_M y$. Now, let $i\in N$ denote an arbitrary voter and let $R^1$ be the profile derived from $R$ by making $x$ and $y$ into the two favorite alternative of voter $i$ without reordering them. In particular, this means that $x\succ_i^1 y$ if and only if $x\succ_i y$. We will show that $f(R^1)=\{x,y\}$, too. To this end, we assume without loss of generality that $x\succ_i y$ and observe that $f(R^1)\neq \{x\}$ as voter $i$ can otherwise manipulate by deviating from $R$ to $R^1$. Moreover, if $x\not\in f(R^1)$ voter $1$ can manipulate by deviating from $R^1$ to $R$ as $x$ is his favorite alternative in $R^1$. As third point, if $y\not\in f(R^1)$ and $f(R^1)\neq \{x\}$, then voter $1$ can manipulate by deviating from $R^1$ to $R$ because his two favorite alternatives are chosen in $R$ but not in $R^1$. Finally, if $f(R)\subsetneq f(R^1)$, voter $1$ can again manipulate by deviating from $R^1$ to $R$. Hence, the only valid choice set for $R^1$ is $\{x,y\}$.
	
	It is now easy to see that we can repeat this argument for one voter after another to arrive at a profile $R'$, where all voters report $x$ and $y$ as their favorite two alternatives and $f(R')=\{x,y\}$. However, $x$ is the Condorcet winner in $R'$ because $x\succ_M y$ by assumption and we never reordered $x$ and $y$ in the preference relations of the voters. Hence, $f(R')=\{x,y\}$ contradicts Condorcet-consistency. This shows that our original assumption is wrong and $f(R)=\{x,y\}$ is indeed only possible if $x\sim_M y$. 
	\end{proof}

We are now ready to prove \Cref{thm:EvenchanceCondImp}.

\evenchance*
\begin{proof}
	In this proof, we focus on the case that $m=5$ because we can simply add further alternatives that are bottom-ranked by all voters to extend our result to $m>5$. These universally bottom-ranked alternatives are Pareto-dominated and thus do not affect our analysis. Now, assume for contradiction that there is an even-chance SDS that satisfies all given axioms for $m=5$ and an odd number of voters $n\geq 5$. We will focus on the following two profiles $R^1$ and $\hat R^1$ in our proof. 
 
    \setlength\tabcolsep{3 pt}
    {\medmuskip=0mu\relax
	\thickmuskip=1mu\relax
    \noindent\begin{tabular}{L{0.07\profilewidth} L{0.3\profilewidth} L{0.3\profilewidth} L{0.3\profilewidth}}
		$R^1$: & $1$: $b,e,d,c,a$ & $2:$ $a,b,c,e,d$ & $3:$ $e,d,c,a,b$
  \end{tabular}
   \begin{tabular}{L{0.07\profilewidth} L{0.43\profilewidth} L{0.43\profilewidth} L{0.45\profilewidth}}
		& $\{4, \dots, \frac{n+3}{2}\}:$ $b,c,a,e,d$ & $\{\frac{n+5}{2}, \dots,n \}$: $e,d,a,b,c$
	\end{tabular}\smallskip

    \noindent\begin{tabular}{L{0.07\profilewidth} L{0.3\profilewidth} L{0.3\profilewidth} L{0.3\profilewidth}}
	$\hat R^1$: & $1$: $b,e,d,c,a$ & $2:$ $a,b,c,e,d$ & $3:$ $d,a,e,b,c$
    \end{tabular}
   \begin{tabular}{L{0.07\profilewidth} L{0.43\profilewidth} L{0.43\profilewidth} L{0.45\profilewidth}}
	& $\{4, \dots, \frac{n+3}{2}\}$: $b,c,a,e,d$ & $\{\frac{n+5}{2}, \dots,n\}$: $e,d,a,b,c$
    \end{tabular}}\smallskip

We observe that profiles there is no Condorcet winner in both profiles. In more detail, for both profiles, it can be checked that $a\succ_M b$, $b\succ_M c$, $b\succ_M d$, $b\succ_M e$, and  $d\succ_M^1 a$. Moreover, $e$ Pareto-dominates $d$ in $R^1$ and $b$ Pareto-dominates $c$ in $\hat R^1$. Consequently, $e\not\in f(R^1)$ and $c\not\in f(\hat R^1)$ by \emph{ex post} efficiency. Subsequently, we will show that $f(R^1)=\{a,b,c,e\}$ and $f(\hat R^1)=\{a,b,d,e\}$. This implies that voter $3$ can manipulate by deviating from $R^1$ to $\hat R^1$ as he prefers $d$ to $c$.\medskip

 \textbf{Claim 1: $f(R^1)=\{a,b,c,e\}$}
 
First, we will prove that $f(R^1)=\{a,b,c,e\}$. To this end, we first note that $|f(R^1)|\geq 3$ due to \Cref{lem:sCC,lem:no2}. Since $d\not\in f(R^1)$ by \emph{ex post} efficiency, we can hence show that $f(R^1)=\{a,b,c,e\}$ by proving that $|f(R^1)|\neq 3$. We do so by considering each possible subset of size $3$ individually. \medskip

\emph{Case 1.1: $f(R^1)\neq \{b,c,e\}$.$\qquad$} Assume for contradiction that $f(R^1)=\{b,c,e\}$ and consider the profile $R^2$ shown below, which is derived from $R^1$ by swapping $a$ and $b$ in the preference relation of voter $2$. It can be checked that $b$ is the Condorcet winner in $R^2$ because the set $\{1,2\} \cup \{4,\dots\frac{n+3}{2}\}$ contains more than half of the voters and all of these voters top-rank $b$. Thus, Condorcet-consistency requires that $f(R^2)=\{b\}$. However this means that voter $2$ can manipulate by deviating from $R^1$ to $R^2$ since he prefers $b$ to both $c$ and $e$. Hence, the assumption that $f(R^1)=\{b,c,e\}$ conflicts with weak strategyproofness.

    {\medmuskip=0mu\relax
	\thickmuskip=1mu\relax
    \noindent\begin{tabular}{L{0.07\profilewidth} L{0.3\profilewidth} L{0.3\profilewidth} L{0.3\profilewidth}}
		$R^2$: & $1$: $b,e,d,c,a$ & $2:$ $b,a,c,e,d$ & $3:$ $e,d,c,a,b$
  \end{tabular}
   \begin{tabular}{L{0.07\profilewidth} L{0.43\profilewidth} L{0.43\profilewidth} L{0.45\profilewidth}}
		& $\{4, \dots, \frac{n+3}{2}\}:$ $b,c,a,e,d$ & $\{\frac{n+5}{2}, \dots,n \}$: $e,d,a,b,c$
	\end{tabular}}\medskip

\emph{Case 1.2: $f(R^1)\neq \{a,c,e\}$.$\qquad$} Assume for contradiction that $f(R^1)=\{a,c,e\}$ and consider the profile $R^3$ shown below, which is derived from $R^1$ by swapping $b$ and $e$ in the preference relation of voter $1$. 
Analogously to the last case, it can be checked that $e$ is top-ranked by more than half of the voters, so Condorcet-consistency requires that $f(R^3)=\{e\}$. However, since $f(R^1)=\{a,c,e\}$ by assumption, we derive that voter $1$ can manipulate because he prefers $\{e\}$ to $\{a,c,e\}$. This contradicts weak strategyproofness, so the assumption that $b\not\in f(R^1)$ must have been wrong.

{\medmuskip=0mu\relax
	\thickmuskip=1mu\relax
    \noindent\begin{tabular}{L{0.07\profilewidth} L{0.3\profilewidth} L{0.3\profilewidth} L{0.3\profilewidth}}
		$R^3$: & $1$: $e,b,d,c,a$ & $2:$ $a,b,c,e,d$ & $3:$ $e,d,c,a,b$
  \end{tabular}
   \begin{tabular}{L{0.07\profilewidth} L{0.43\profilewidth} L{0.43\profilewidth} L{0.45\profilewidth}}
		& $\{4, \dots, \frac{n+3}{2}\}:$ $b,c,a,e,d$ & $\{\frac{n+5}{2}, \dots,n \}$: $e,d,a,b,c$
	\end{tabular}}\medskip

\emph{Case 1.3: $f(R^1)\neq \{a,b,e\}$.$\qquad$}
Assume for contradiction that $f(R^1)=\{a,b,e\}$ and consider the profile $R^4$ shown below, where voter $n$ deviates to $a,e,d,b,c$. We first note that there is no Condorcet winner in $R^4$. Moreover, $d$ is still Pareto-dominated by $e$ in $R^4$, so $f(R^4)\subseteq \{a,b,c,e\}$. Since there is no Condorcet winner in $R^4$, it follows from \Cref{lem:sCC,lem:no2} that $|f(R^4)|\geq 3$. Now, if $c\in f(R^4)$, voter $4$ prefers $f(R^1)=\{a,b,e\}$ to $f(R^4)$ as $c$ is his least preferred alternative in $R^4$. In more detail, if $|f(R^4)|=3$ and $c\in f(R^4)$, then we only substitute $c$ with another alternative $x$, which makes voter $n$ better off since $x\succ_n^4 c$. On the other hand, if $|f(R^4)|=4$, then $f(R^4)=\{a,b,c,e\}$ and voter $n$ prefers $f(R^3)$ as $c$ is his least preferred alternative. Hence, we have that $f(R^4)=\{a,b,e\}$. 
 
By repeatedly applying this argument for one voter after another, we can also derive for the profile $R^5$ shown below that $f(R^5)=\{a,b,e\}$. In particular, we note that $a$ is still not the Condorcet winner in $R^5$ as $c\succ_M^5 a$. 

Finally, consider the profile $R^6$, which is derived from $R^5$ by swapping $a$ and $c$ in the preference relation of voter~$1$. Since $a$ is the Condorcet winner in $R^6$, it holds that $f(R^6)=\{a\}$. However, this means that voter $1$ can manipulate by deviating from $R^6$ to $R^5$ since he prefers $\{a,b,e\}$ to $\{a\}$ according to $\succsim_1^6$. This contradicts weak strategyproofness, so the assumption that $f(R^1)=\{a,b,e\}$ must be wrong.\smallskip

{\medmuskip=0mu\relax
	\thickmuskip=1mu\relax
    \noindent\begin{tabular}{L{0.07\profilewidth} L{0.3\profilewidth} L{0.3\profilewidth} L{0.3\profilewidth}}
		$R^4$: & $1$: $b,e,d,c,a$ & $2:$ $a,b,c,e,d$ & $3:$ $e,d,c,a,b$
  \end{tabular}
   \begin{tabular}{L{0.07\profilewidth} L{0.47\profilewidth} L{0.43\profilewidth}}
		 & $\{4, \dots, \frac{n+3}{2}\}:$ $b,c,a,e,d$ & $n$: $a,e,d,b,c$ \\
   \end{tabular}
   \begin{tabular}{L{0.07\profilewidth} L{0.6\profilewidth}}
  & $\{\frac{n+5}{2}, \dots,n-1\}$: $e,d,a,b,c$
	\end{tabular}}\smallskip

 {\medmuskip=0mu\relax
	\thickmuskip=1mu\relax
    \noindent\begin{tabular}{L{0.07\profilewidth} L{0.3\profilewidth} L{0.3\profilewidth} L{0.3\profilewidth}}
		$R^5$: & $1$: $b,e,d,c,a$ & $2:$ $a,b,c,e,d$ & $3:$ $e,d,c,a,b$
  \end{tabular}
    \begin{tabular}{L{0.07\profilewidth} L{0.43\profilewidth} L{0.43\profilewidth}}
		 & $\{4, \dots, \frac{n+3}{2}\}:$ $b,c,a,e,d$ & $\{\frac{n+5}{2}, \dots,n \}$: $a,e,d,b,c$
	\end{tabular}}\smallskip
 
 {\medmuskip=0mu\relax
	\thickmuskip=1mu\relax
    \noindent\begin{tabular}{L{0.07\profilewidth} L{0.3\profilewidth} L{0.3\profilewidth} L{0.3\profilewidth}}
		$R^6$: & $1$: $b,e,d,a,c$ & $2:$ $a,b,c,e,d$ & $3:$ $e,d,c,a,b$
  \end{tabular}
    \begin{tabular}{L{0.07\profilewidth} L{0.43\profilewidth} L{0.43\profilewidth}}
		&  $\{4, \dots, \frac{n+3}{2}\}:$ $b,c,a,e,d$ & $\{\frac{n+5}{2}, \dots,n \}$: $a,e,d,b,c$
	\end{tabular}}\medskip
\medskip

\emph{Case 1.4: $f(R^1)\neq \{a,b,c\}$.$\qquad$}
For our last step, we suppose that $f(R^1)=\{a,b,c\}$. Now, consider the profile $R^7$ that arises from $R^1$ by swapping $b$ and $c$ in the preference relation of voter $4$. We note that if $n=5$, then the set $\{5,\dots, \frac{n+3}{2}\}$ is empty and there is no voter in $R^7$ reports $b,c,a,e,d$. There is no Condorcet winner in $R^7$ and $e$ still Pareto-dominates $d$. We can hence conclude that $d\not\in f(R^7)$ due to \emph{ex post} efficiency and that $|f(R^7)|\geq 3$ due to \Cref{lem:sCC,lem:no2}. Moreover, if $|f(R^7)|\geq 3$ and $e\in f(R^7)$, voter $4$ can manipulate by deviating from $R^7$ to $R^1$. Hence, it follows that $f(R^7)=\{a,b,c\}$. By repeating this argument for one voter after another in $\{4,\dots, \frac{n+3}{2}\}$, we can infer the same for the profile $R^8$ shown below, i.e., $f(R^8)=\{a,b,c\}$.
 
For the next step, we consider the profile $R^9$ derived from $R^8$ by swapping $b$ and $c$ in the preference relation of voter $2$. There is still no Condorcet winner in $R^9$ since $e\succ_M^9 c$. As a consequence, essentially the same arguments as for $R^7$ show that $f(R^9)=\{a,b,c\}$. 

Finally, we consider the profile $R^{10}$ which is derived from $R^9$ by changing the preference relation of voter $3$ to $c,e,d,a,b$. It can be checked that $c$ is the Condorcet winner in $R^{10}$, so $f(R^{10})=\{c\}$ due to Condorcet-consistency. However, this means that voter $5$ can manipulate by deviating from $R^9$ to $R^{10}$, which contradicts the weak strategyproofness of $f$. Hence, the assumption that $f(R^1)=\{a,b,c\}$ must have been wrong.
 \smallskip

 {\medmuskip=0mu\relax
	\thickmuskip=1mu\relax
    \noindent\begin{tabular}{L{0.08\profilewidth} L{0.3\profilewidth} L{0.3\profilewidth} L{0.3\profilewidth}}
		$R^7$: & $1$: $b,e,d,c,a$ & $2:$ $a,b,c,e,d$ & $3:$ $e,d,c,a,b$\\
  \end{tabular}
   \begin{tabular}{L{0.08\profilewidth} L{0.3\profilewidth} L{0.43\profilewidth}}
     & $4$: $c,b,a,e,d$ & $\{5, \dots, \frac{n+3}{2}\}:$ $b,c,a,e,d$ \\
    \end{tabular}
    \begin{tabular}{L{0.08\profilewidth} L{0.43\profilewidth}L{0.43\profilewidth}}
  & $\{\frac{n+5}{2}, \dots,n \}$: $e,d,a,b,c$
	\end{tabular}}
 \smallskip

 {\medmuskip=0mu\relax
	\thickmuskip=1mu\relax
    \noindent\begin{tabular}{L{0.08\profilewidth} L{0.3\profilewidth} L{0.3\profilewidth} L{0.3\profilewidth}}
		$R^8$: & $1$: $b,e,d,c,a$ & $2:$ $a,b,c,e,d$ & $3:$ $e,d,c,a,b$\\
  \end{tabular}
    \begin{tabular}{L{0.08\profilewidth} L{0.43\profilewidth}L{0.43\profilewidth}}
  & $\{4, \dots, \frac{n+3}{2}\}:$ $c,b,a,e,d$ & $\{\frac{n+5}{2}, \dots,n \}$: $e,d,a,b,c$
	\end{tabular}}
 \smallskip

  {\medmuskip=0mu\relax
	\thickmuskip=1mu\relax
    \noindent\begin{tabular}{L{0.08\profilewidth} L{0.3\profilewidth} L{0.3\profilewidth} L{0.3\profilewidth}}
		$R^9$: & $1$: $b,e,d,c,a$ & $2:$ $a,c,b,e,d$ & $3:$ $e,d,c,a,b$\\
  \end{tabular}
    \begin{tabular}{L{0.08\profilewidth} L{0.43\profilewidth}L{0.43\profilewidth}}
  & $\{4, \dots, \frac{n+3}{2}\}:$ $c,b,a,e,d$ & $\{\frac{n+5}{2}, \dots,n \}$: $e,d,a,b,c$
	\end{tabular}}
 \smallskip

  {\medmuskip=0mu\relax
	\thickmuskip=1mu\relax
    \noindent\begin{tabular}{L{0.08\profilewidth} L{0.3\profilewidth} L{0.3\profilewidth} L{0.3\profilewidth}}
		$R^{10}$: & $1$: $b,e,d,c,a$ & $2:$ $a,c,b,e,d$ & $3:$ $c,e,d,a,b$\\
  \end{tabular}
    \begin{tabular}{L{0.08\profilewidth} L{0.43\profilewidth}L{0.43\profilewidth}}
  & $\{4, \dots, \frac{n+3}{2}\}:$ $c,b,a,e,d$ & $\{\frac{n+5}{2}, \dots,n \}$: $e,d,a,b,c$
	\end{tabular}}
 \smallskip

 \textbf{Claim 2: $f(\hat R^1)=\{a,b,d,e\}$}

We now turn attention to our second claim, i.e., that $f(\hat R^1)=\{a,b,d,e\}$. To this end, we note that \Cref{lem:sCC,lem:no2} imply that $|f(\hat R^1)|\geq 3$ and that $c\not\in f(\hat R^1)$ since it is Pareto-dominated by $b$. Hence, we will again show that no choice set of size $3$ is a valid outcome for $f(\hat R^1)$. For the reader's convenience, we display the profile $\hat R^1$ again.
\smallskip

{\medmuskip=0mu\relax
	\thickmuskip=1mu\relax
    \noindent\begin{tabular}{L{0.07\profilewidth} L{0.3\profilewidth} L{0.3\profilewidth} L{0.3\profilewidth}}
	$\hat R^1$: & $1$: $b,e,d,c,a$ & $2:$ $a,b,c,e,d$ & $3:$ $d,a,e,b,c$
    \end{tabular}
   \begin{tabular}{L{0.07\profilewidth} L{0.43\profilewidth} L{0.43\profilewidth} L{0.45\profilewidth}}
	& $\{4, \dots, \frac{n+3}{2}\}$: $b,c,a,e,d$ & $\{\frac{n+5}{2}, \dots,n\}$: $e,d,a,b,c$
    \end{tabular}}\medskip

\emph{Case 2.1: $f(\hat R^1)\neq \{b,d,e\}$.$\qquad$}
Assume for contradiction that $f(\hat R^1)=\{b,d,e\}$. In this case, consider the profile $\hat R^2$ derived from $\hat R^1$ by swapping $a$ and $b$ in the preference relation of voter $2$. 
In the profile $\hat R^2$, $b$ is the Condorcet winner as it is top-ranked by a majority of the voters. So, Condorcet-consistency requires that $f(\hat R^2)=\{b\}$. However, when $f(\hat R^1)=\{b,d,e\}$, this means that voter $3$ can manipulate by deviating from $\hat R^1$ to $\hat R^2$. This contradicts the weak strategyproofness of $f$, so the assumption that $f(\hat R^1)=\{b,d,e\}$ must be wrong.
\smallskip

{\medmuskip=0mu\relax
	\thickmuskip=1mu\relax
    \noindent\begin{tabular}{L{0.07\profilewidth} L{0.3\profilewidth} L{0.3\profilewidth} L{0.3\profilewidth}}
	$\hat R^2$: & $1$: $b,e,d,c,a$ & $2:$ $b,a,c,e,d$ & $3:$ $d,a,e,b,c$
    \end{tabular}
   \begin{tabular}{L{0.07\profilewidth} L{0.43\profilewidth} L{0.43\profilewidth} L{0.45\profilewidth}}
	& $\{4, \dots, \frac{n+3}{2}\}$: $b,c,a,e,d$ & $\{\frac{n+5}{2}, \dots,n\}$: $e,d,a,b,c$
    \end{tabular}}\medskip

\emph{Case 2.2: $f(\hat R^1)\neq \{a,b,e\}$.$\qquad$} For our second case, we assume for contradiction that $f(\hat R^1)=\{a,b,e\}$. Now, consider the profile $\hat R^3$ derived from $\hat R^1$ by swapping $a$ and $d$ in the preference relation of voter $3$. In this profile, alternative $a$ is the Condorcet winner, so $f(\hat R^3)=\{a\}$ due to Condorcet-consistency. However, since $f(\hat R^1)=\{a,b,e\}$, this means that voter $3$ can manipulate by deviating from $\hat R^1$ to $\hat R^3$. This is a contradiction to the weak strategyproofness of $f$, and it therefore follows that $f(\hat R^1)\neq\{a,b,e\}$.
\smallskip

{\medmuskip=0mu\relax
	\thickmuskip=1mu\relax
    \noindent\begin{tabular}{L{0.07\profilewidth} L{0.3\profilewidth} L{0.3\profilewidth} L{0.3\profilewidth}}
	$\hat R^3$: & $1$: $b,e,d,c,a$ & $2:$ $a,b,c,e,d$ & $3:$ $a,d,e,b,c$
    \end{tabular}
   \begin{tabular}{L{0.07\profilewidth} L{0.43\profilewidth} L{0.43\profilewidth} L{0.45\profilewidth}}
	& $\{4, \dots, \frac{n+3}{2}\}$: $b,c,a,e,d$ & $\{\frac{n+5}{2}, \dots,n\}$: $e,d,a,b,c$
    \end{tabular}}\smallskip
\medskip

\emph{Case 2.3: $f(\hat R^1)\neq \{a,b,d\}$.$\qquad$} Assume for contradiction that $f(\hat R^1)=\{a,b,d\}$. We first consider the profile $\hat R^4$, which is derived from $\hat R^1$ by making $d$ into the favorite alternative of voter $4$. There is no Condorcet winner in $\hat R^4$, so $|f(\hat R^4)|\geq 3$ by \Cref{lem:sCC,lem:no2}. Moreover, $c$ is still Pareto-dominated by $b$, so $c\not\in f(\hat R^4)$. Finally, if $e\in f(\hat R^4)$ and $|f(\hat R^4)|\geq 3$, voter $4$ can manipulate by deviating from $\hat R^4$ to $\hat R^1$ as $e$ is his least preferred alternative in $\hat R^4$. Hence, $f(\hat R^4)=\{a,b,d\}$. 

Furthermore, we can apply the same argument for every voter $i\in \{5,\dots, \frac{n+3}{2}\}$ to infer the same for profile $\hat R^5$ shown below. In particular, we note that there is a majority of voters who prefers $e$ to $d$ in $\hat R^5$, which verifies that $d$ is not the Condorcet winner in this profile. 

Next, let $\hat R^6$ denote the profile derived from $\hat R^5$ by letting voter $n$ make $d$ into his favorite alternative. Alternative $d$ is the Condorcet winner in $\hat R^6$, so $f(\hat R^6)=\{d\}$. However, this is means that voter $n$ can manipulate by deviating from $\hat R^5$ to $\hat R^6$ as he prefers $\{d\}$ to $\{a,b,d\}$. This contradicts the weak strategyproofness of $f$, so the assumption that $f(\hat R^1)=\{a,b,d\}$ must have been wrong.
\smallskip

{\medmuskip=0mu\relax
	\thickmuskip=1mu\relax
    \noindent\begin{tabular}{L{0.07\profilewidth} L{0.3\profilewidth} L{0.3\profilewidth} L{0.3\profilewidth}}
	$\hat R^4$: & $1$: $b,e,d,c,a$ & $2:$ $a,b,c,e,d$ & $3:$ $d,a,e,b,c$
    \end{tabular}
    \noindent\begin{tabular}{L{0.07\profilewidth} L{0.3\profilewidth} L{0.5\profilewidth} L{0.3\profilewidth}}
	& $4$: $d,b,c,a,e$ & $\{5, \dots, \frac{n+3}{2}\}$: $b,c,a,e,d$
    \end{tabular}
   \begin{tabular}{L{0.07\profilewidth} L{0.43\profilewidth} L{0.43\profilewidth} L{0.45\profilewidth}}
	 & $\{\frac{n+5}{2}, \dots,n\}$: $e,d,a,b,c$
    \end{tabular}}\smallskip

{\medmuskip=0mu\relax
	\thickmuskip=1mu\relax
    \noindent\begin{tabular}{L{0.07\profilewidth} L{0.3\profilewidth} L{0.3\profilewidth} L{0.3\profilewidth}}
	$\hat R^5$: & $1$: $b,e,d,c,a$ & $2:$ $a,b,c,e,d$ & $3:$ $d,a,e,b,c$
    \end{tabular}
   \begin{tabular}{L{0.07\profilewidth} L{0.43\profilewidth} L{0.43\profilewidth} L{0.45\profilewidth}}
	& $\{4, \dots, \frac{n+3}{2}\}$: $d,b,c,a,e$ & $\{\frac{n+5}{2}, \dots,n\}$: $e,d,a,b,c$
    \end{tabular}}\smallskip

{\medmuskip=0mu\relax
	\thickmuskip=1mu\relax
    \noindent\begin{tabular}{L{0.07\profilewidth} L{0.3\profilewidth} L{0.3\profilewidth} L{0.3\profilewidth}}
	$\hat R^6$: & $1$: $b,e,d,c,a$ & $2:$ $a,b,c,e,d$ & $3:$ $d,a,e,b,c$
    \end{tabular}
    \noindent\begin{tabular}{L{0.07\profilewidth} L{0.5\profilewidth} L{0.3\profilewidth}}
     & $\{4, \dots, \frac{n+3}{2}\}$: $d,b,c,a,e$ & $n$: $d,e,a,b,c$ \\
     & $\{\frac{n+5}{2}, \dots,n-1\}$: $e,d,a,b,c$
    \end{tabular}}\medskip

\emph{Case 2.4: $f(\hat R^1)\neq \{a,d,e\}$.$\qquad$}
   For the last step, we assume for contradiction that $f(\hat R^1)= \{a,d,e\}$. In this case, we first consider the profile $\hat R^7$ which is derived from $\hat R^1$ by letting voter $3$ swap $a$ and $e$. We first note that there is no Condorcet winner in $\hat R^7$ and that $b$ still Pareto-dominates $c$. Hence, $f(\hat R^7)\subseteq \{a,b,d,e\}$ due to \emph{ex post} efficiency and $|f(\hat R^7)|\geq 3$ because of \Cref{lem:sCC,lem:no2}. Now, since $b$ is the least preferred alternative of voter $3$ among the Pareto-optimal ones, we infer from weak strategyproofness that $b\not\in f(\hat R^7)$ as this voter can otherwise manipulate by deviating from $\hat R^7$ to $\hat R^1$. So, $f(\hat R^7)=\{a,d,e\}$. 
   
   Finally, let $\hat R^8$ denote the profile derived from $\hat R^7$ by letting voter $1$ swap $b$ and $e$. It can be checked that $e$ is the Condorcet winner in $\hat R^8$, so $f(\hat R^8)=\{e\}$ by Condorcet-consistency. However, this means that voter $1$ can manipulate by deviating from $\hat R^7$ to $\hat R^8$ as he prefers $\{e\}$ to $\{a,d,e\}$ in $\hat R^7$.\smallskip

   {\medmuskip=0mu\relax
	\thickmuskip=1mu\relax
    \noindent\begin{tabular}{L{0.07\profilewidth} L{0.3\profilewidth} L{0.3\profilewidth} L{0.3\profilewidth}}
	$\hat R^7$: & $1$: $b,e,d,c,a$ & $2:$ $a,b,c,e,d$ & $3:$ $d,e,a,b,c$
    \end{tabular}
   \begin{tabular}{L{0.07\profilewidth} L{0.43\profilewidth} L{0.43\profilewidth} L{0.45\profilewidth}}
	& $\{4, \dots, \frac{n+3}{2}\}$: $b,c,a,e,d$ & $\{\frac{n+5}{2}, \dots,n\}$: $e,d,a,b,c$
    \end{tabular}}\smallskip

    {\medmuskip=0mu\relax
	\thickmuskip=1mu\relax
    \noindent\begin{tabular}{L{0.07\profilewidth} L{0.3\profilewidth} L{0.3\profilewidth} L{0.3\profilewidth}}
	$\hat R^8$: & $1$: $e,b,d,c,a$ & $2:$ $a,b,c,e,d$ & $3:$ $d,e,a,b,c$
    \end{tabular}
   \begin{tabular}{L{0.07\profilewidth} L{0.43\profilewidth} L{0.43\profilewidth} L{0.45\profilewidth}}
	& $\{4, \dots, \frac{n+3}{2}\}$: $b,c,a,e,d$ & $\{\frac{n+5}{2}, \dots,n\}$: $e,d,a,b,c$
    \end{tabular}}\smallskip

This completes the proof of this theorem.
\end{proof}

As the last point in this section, we will prove a variant of \Cref{thm:EvenchanceCondImp} for the case that $n$ is even. To this end, we will use \emph{strong Condorcet-consistency} instead of Condorcet-consistency, which requires that an SDS assigns probability $1$ to an alternative if and only if it is the Condorcet winner. 

\begin{proposition}
	Assume that $m\geq 5$ and $n\geq 8$ is even. There is no even-chance SDS that satisfies weak strategyproofness, strong Condorcet-consistency, and \emph{ex post} efficiency.
\end{proposition}
\begin{proof}
	Just as for the case that $n$ is odd, we will again assume that $m=5$ as it is straightforward to generalize our analysis to a larger number of alternatives by using \emph{ex post} efficiency. Moreover, we also assume that there are precisely $n=8$ voters; it can be checked that the impossibility can be generalized to every even $n>8$ by adding pairs of voters who report $a,b,c,d,e$ and $d,e,b,c,a$. 
 
 Now, we assume for contradiction that there is an even-chance SDS that satisfies weak strategyproofness, strong Condorcet-consistency, and \emph{ex post} efficiency for $m=5$ alternatives and $n=8$ voters. Our proof will focus on the following preference profiles.\smallskip

\setlength\tabcolsep{3 pt}
\noindent
{\medmuskip=0mu\relax
	\thickmuskip=1mu\relax
    \begin{tabular}{L{0.07\profilewidth} L{0.4\profilewidth} L{0.4\profilewidth}}
	$R^1$: & $\{1,2,3\}$: $b,c,a,e,d$ & $\{4,5,6\}$: $a,e,d,b,c$ \\
 & $7:$ $e,d,b,c,a$ & $8$: $e,d,c,b,a$
\end{tabular}}\smallskip

\noindent
{\medmuskip=0mu\relax
	\thickmuskip=1mu\relax
\begin{tabular}{L{0.07\profilewidth} L{0.4\profilewidth} L{0.4\profilewidth}}
	$\hat R^1$: & $\{1,2,3\}$: $b,c,a,e,d$ & $\{4,5,6\}$: $a,e,d,b,c$ \\
 & $7:$ $e,d,b,c,a$ & $8$: $d,e,b,c,a$
\end{tabular}
}\smallskip

Just as for \Cref{thm:EvenchanceCondImp}, we will show that $f(R^1)=\{a,b,c,e\}$ and $f(\hat R^1)=\{a,b,d,e\}$. This means that voter~$8$ can manipulate by deviating from $R^1$ to $\hat R^1$, thus showing that $f$ fails weak strategyproofness.\medskip

\textbf{Claim 1: $f(\hat R^1)=\{a,b,d,e\}$}

We will first show that $f(\hat R^1)=\{a,b,d,e\}$. To this end, we note that there is no Condorcet winner or majority tie in $R^1$, so strong Condorcet-consistency and \Cref{lem:no2} imply that $|f(\hat R^1)|\geq 3$. Moreover, $c$ is Pareto-dominated by $b$, so $c\not\in f(\hat R^1)$ by \emph{ex post} efficiency. We can thus prove our claim by showing that no choice set of size $3$ is chosen for $\hat R^1$. \medskip

\textit{Case 1.1: $f(\hat R^1)\neq \{a,d,e\}$ and $f(\hat R^1)\neq \{b,d,e\}$.$\qquad$}

As the first point, we will show that $f(\hat R^1)\neq \{a,d,e\}$ and $f(\hat R^1)\neq \{b,d,e\}$. Towards this end, we first assume for contradiction that $f(\hat R^1)=\{b,d,e\}$. In this case, we let the voters $6$, $5$, and $4$ one after another deviate to the preference relation $e,a,d,b,c$, which results in the profiles $\hat R^2$, $\hat R^3$, and $\hat R^4$ shown below. For $\hat R^2$, we first note $|f(\hat R^2)|\neq 1$ as there is no Condorcet winner in this profile. Moreover, \Cref{lem:no2} shows that $|f(\hat R^2)|\neq 2$ as there is no majority tie in $\hat R^2$. Finally, if $|f(\hat R^2)|\geq 3$ and $a\in f(\hat R^2)$, voter $6$ can manipulate by deviating from $\hat R^1$ to $\hat R^2$, so we conclude now that $f(\hat R^2)=\{b,d,e\}$. 

Next, for $\hat R^3$, there is still no Condorcet winner, so $|f(\hat R^3)|\neq 1$. Moreover, $f(\hat R^3)\neq \{a,e\}$ as voter $5$ can otherwise manipulate by deviating from $\hat R^2$ to $\hat R^3$. Since $a\mathrel{\hat{\sim}_3} e$ is the only majority tie in $\hat R^3$, we now infer that $|f(\hat R^3)|\neq 2$ due to \Cref{lem:no2}. Finally, it follows now analogously to $\hat R^2$ that $f(\hat R^3)=\{b,d,e\}$ as voter $5$ prefers every set $X$ with $a\in X$, $c\not\in X$, and $|X|\geq 3$ to $\{b,d,e\}$.

Lastly, we observe that $e$ is the Condorcet winner in $\hat R^4$, so strong Condorcet-consistency requires that $f(\hat R^3)=\{e\}$. This, however, conflicts with weak strategyproofness as voter $4$ prefers $\{e\}$ to $\{b,d,e\}$. This shows that our assumption that $f(\hat R^1)=\{b,d,e\}$ was wrong.
\smallskip

\noindent
{\medmuskip=0mu\relax
	\thickmuskip=1mu\relax
\begin{tabular}{L{0.07\profilewidth} L{0.4\profilewidth} L{0.4\profilewidth}}
	$\hat R^2$: & $\{1,2,3\}$: $b,c,a,e,d$ & $\{4,5\}$: $a,e,d,b,c$
 \end{tabular}
 \begin{tabular}{L{0.07\profilewidth} L{0.28\profilewidth} L{0.28\profilewidth}  L{0.28\profilewidth}}
      & $6$: $e,a,d,b,c$ & $7:$ $e,d,b,c,a$ & $8$: $d,e,b,c,a$
 \end{tabular}}\smallskip

{\noindent\medmuskip=0mu\relax
	\thickmuskip=1mu\relax
\begin{tabular}{L{0.07\profilewidth} L{0.4\profilewidth} L{0.4\profilewidth}}
	$\hat R^3$: & $\{1,2,3\}$: $b,c,a,e,d$ & $4$: $a,e,d,b,c$
 \end{tabular}
 \begin{tabular}{L{0.07\profilewidth} L{0.33\profilewidth} L{0.25\profilewidth}  L{0.25\profilewidth}}
      & $\{5,6\}$: $e,a,d,b,c$ & $7:$ $e,d,b,c,a$ & $8$: $d,e,b,c,a$
 \end{tabular}}\smallskip

{\noindent\medmuskip=0mu\relax
	\thickmuskip=1mu\relax
\begin{tabular}{L{0.07\profilewidth} L{0.4\profilewidth} L{0.4\profilewidth}}
	$\hat R^4$: & $\{1,2,3\}$: $b,c,a,e,d$ & $\{4,5,6\}$: $e,a,d,b,c$\\
  & $7:$ $e,d,b,c,a$ & $8$: $d,e,b,c,a$
 \end{tabular}
}\smallskip

Finally, the case that $f(\hat R^1)=\{a,d,e\}$ can be excluded by a symmetric argument when letting voters $1$ and $2$ deviate to $a,b,c,e,d$.\medskip

\textit{Case 1.2: $f(\hat R^1)\neq \{a,b,d\}$.$\qquad$} 
Assume for contradiction that $f(\hat R^1)= \{a,b,d\}$. First, we let voter $7$ deviate to the preference relation $e,b,d,c,a$ to derive the profile $\hat R^5$. We note for this profile that there is no Condorcet winner, so $|f(\hat R^5)|\neq 1$ by strong Condorcet-consistency. Next, $f(\hat R^5)\neq \{b,d\}$ as voter $7$ can otherwise manipulate by deviating from $\hat R^1$ to $\hat R^5$. Since $b\mathrel{\hat\sim_M^5} d$ is the only majority tie in $\hat R^5$, we infer from \Cref{lem:no2} that $|f(\hat R^5)|\neq 2$. As a consequence, it holds that $|f(\hat R^5)|\geq 3$, which implies that $e\not\in f(\hat R^5)$. Otherwise, voter $7$ can manipulate by deviating from $\hat R^1$ to $\hat R^5$. This proves that $f(\hat R^5)=\{a,b,d\}$

As the second step, we consider the profile $\hat R^6$ derived from $\hat R^4$ by assigning voter $6$ the preference relation $a,b,d,e,c$. Once again, strong Condorcet-consistency implies that $|f(\hat R^6)|\neq 1$ as there is no Condorcet winner in $\hat R^6$. Moreover, $f(\hat R^6)\neq \{b,e\}$ because voter $6$ can otherwise manipulate by deviating from $\hat R^6$ to $\hat R^5$. As a consequence, we derive from \Cref{lem:no2} that $|f(\hat R^6)|\neq 2$ because $b\mathrel{\hat \sim_M^6} e$ is the only majority tie in $\hat R^6$. As the last point on $\hat R^6$, we observe that, if $|f(\hat R^6)|\geq 3$, then $e\not\in f(\hat R^6)$ as voter $6$ can otherwise manipulate by deviating from $\hat R^6$ back to $\hat R^5$. This implies that $f(\hat R^6)=\{a,b,d\}$. 

Finally, voter $7$ can now manipulate by swapping $b$ and $e$. This step results in the profile $\hat R^7$, where $b$ is the Condorcet winner and therefore uniquely chosen. However, voter $7$ prefers the set $\{b\}$ to the set $\{b,d,e\}$, so this violates weak strategyproofness. This shows that our assumption that $f(\hat R^1)= \{a,b,d\}$ must have been wrong. 
\smallskip\\

\noindent
{\medmuskip=0mu\relax
	\thickmuskip=1mu\relax
\begin{tabular}{L{0.07\profilewidth} L{0.4\profilewidth} L{0.4\profilewidth}}
	$\hat R^5$: & $\{1,2,3\}$: $b,c,a,e,d$ & $\{4,5,6\}$: $a,e,d,b,c$ \\
 & $7:$ $e,b,d,c,a$ & $8$: $d,e,b,c,a$
\end{tabular}}\smallskip

\noindent
{\medmuskip=0mu\relax
	\thickmuskip=1mu\relax
\begin{tabular}{L{0.07\profilewidth} L{0.4\profilewidth} L{0.4\profilewidth}}
	$\hat R^6$: & $\{1,2,3\}$: $b,c,a,e,d$ & $\{4,5\}$: $a,e,d,b,c$
 \end{tabular}
 \begin{tabular}{L{0.07\profilewidth} L{0.28\profilewidth} L{0.28\profilewidth}  L{0.28\profilewidth}}
      & $6$: $a,b,d,e,c$ & $7:$ $e,b,d,c,a$ & $8$: $d,e,b,c,a$
 \end{tabular}}\smallskip

\noindent
{\medmuskip=0mu\relax
	\thickmuskip=1mu\relax
\begin{tabular}{L{0.07\profilewidth} L{0.4\profilewidth} L{0.4\profilewidth}}
	$\hat R^7$: & $\{1,2,3\}$: $b,c,a,e,d$ & $\{4,5\}$: $a,e,d,b,c$
 \end{tabular}
 \begin{tabular}{L{0.07\profilewidth} L{0.28\profilewidth} L{0.28\profilewidth}  L{0.28\profilewidth}}
      & $6$: $a,b,d,e,c$ & $7:$ $b,e,d,c,a$ & $8$: $d,e,b,c,a$
 \end{tabular}}\medskip

\textit{Case 1.3: $f(\hat R^1)\neq \{a,b,e\}$.$\qquad$} 
We assume for contradiction that $f(\hat R^1)=\{a,b,e\}$. In this case, we first let voter $8$ swap $b$ and $e$, which results in the profile $\hat R^{8}$. Since there is no Condorcet winner in $\hat R^8$, strong Condorcet-consistency necessitates that $|f(\hat R^8)|\neq 1$. Moreover, it holds that $f(\hat R^8)\neq \{b,e\}$ as voter $8$ can otherwise manipulate by deviating from $\hat R^1$ to $\hat R^8$. Since $b\mathrel{\hat\sim_M^{8}} e$ is the only majority tie in $\hat R^{8}$, it follows from \Cref{lem:no2} that $|f(\hat R^{8})|\neq 2$. Finally, if $|f(\hat R^8)|\geq 3$, then $d\not\in f(\hat R^{8})$ as every \emph{ex post} efficient outcome with $d\in f(\hat R^{8})$ and $|f(\hat R^{7})|\geq 3$ means that voter $8$ can manipulate by deviating from $\hat R^1$ to $\hat R^{8}$. Hence, we conclude that $f(\hat R^{8})=\{a,b,e\}$.

As second step, we consider the profile $\hat R^{9}$ which is derived from $\hat R^{8}$ by assigning voter $6$ the preference relation $a,b,e,d,c$. First, we note once again that there is no Condorcet winner in $\hat R^9$, so no singleton set can be chosen. Moreover, $f(\hat R^9)\neq \{b,d\}$ as voter $6$ can otherwise manipulate by deviating back to $\hat R^8$. Consequently, \Cref{lem:no2} shows now that $|f(\hat R^9)|\neq 2$ as $b\mathrel{\hat \sim_M^9} d$ is the only majority tie in $\hat R^9$. This proves that $|f(\hat R^9)|\geq 3$. In turn, it follows that $d\not\in f(\hat R^9)$ as voter $6$ can otherwise manipulate by deviating from $\hat R^9$ to $\hat R^8$. Hence, the only valid outcome for this profile is $f(\hat R^9)=\{a,b,e\}$.

Finally, that means that voter $8$ can manipulate by swapping $b$ and $d$. In the resulting profile $\hat R^{10}$, $b$ is the Condorcet winner, so strong Condorcet-consistency requires that the set $\{b\}$ is chosen. As voter $8$ prefers $b$ to both $a$ and $e$ in $\hat R^{9}$, this is a manipulation and hence proves that $f(\hat R^1)\neq\{a,b,e\}$.
\smallskip

\noindent
{\medmuskip=0mu\relax
	\thickmuskip=1mu\relax
\begin{tabular}{L{0.08\profilewidth} L{0.4\profilewidth} L{0.4\profilewidth}}
	$\hat R^8$: & $\{1,2,3\}$: $b,c,a,e,d$ & $\{4,5,6\}$: $a,e,d,b,c$ \\
 & $7:$ $e,d,b,c,a$ & $8$: $d,b,e,c,a$
\end{tabular}}\smallskip

\noindent
{\medmuskip=0mu\relax
	\thickmuskip=1mu\relax
\begin{tabular}{L{0.08\profilewidth} L{0.4\profilewidth} L{0.4\profilewidth}}
	$\hat R^9$: & $\{1,2,3\}$: $b,c,a,e,d$ & $\{4,5\}$: $a,e,d,b,c$ \\
 \end{tabular}
 \begin{tabular}{L{0.08\profilewidth} L{0.28\profilewidth} L{0.28\profilewidth}  L{0.28\profilewidth}}
      & $6$: $a,b,e,d,c$ & $7:$ $e,d,b,c,a$ & $8$: $d,b,e,c,a$
 \end{tabular}}\smallskip

\noindent
{\medmuskip=0mu\relax
	\thickmuskip=1mu\relax
\begin{tabular}{L{0.08\profilewidth} L{0.4\profilewidth} L{0.4\profilewidth}}
	$\hat R^{10}$: & $\{1,2,3\}$: $b,c,a,e,d$ & $\{4,5\}$: $a,e,d,b,c$ \\
 \end{tabular}
 \begin{tabular}{L{0.08\profilewidth} L{0.28\profilewidth} L{0.28\profilewidth}  L{0.28\profilewidth}}
      & $6$: $a,b,e,d,c$ & $7:$ $e,d,b,c,a$ & $8$: $b,d,e,c,a$
 \end{tabular}}\smallskip

\textbf{Claim 2: $f(\hat R^1)=\{a,b,c,e\}$}

As our second claim, we will show that $f(R^1)=\{a,b,c,e\}$. To this end, we first note that there is no Condorcet winner and no majority tie in $R^1$, so strong Condorcet-consistency and \Cref{lem:no2} imply that $|f(R^1)|\geq 3$. Moreover, $d$ is Pareto-dominated by $e$ in $R^1$, so $d\not\in f(R^1)$ by \emph{ex post} efficiency. Hence, we will again show that $f(R^1)=\{a,b,c,e\}$ by ruling out the four remaining sets of size $3$.\medskip

\textit{Case 2.1: $f(R^1)\neq \{a,b,c\}$.$\qquad$} Assume for contradiction that $f(R^1)= \{a,b,c\}$. In this case, we first let voter $8$ manipulate by reporting the preference relation $e,d,b,c,a$. In the resulting profile $R^2$, $b$ Pareto-dominates $c$ and $e$ Pareto-dominates $d$, so $f(R^2)\subseteq \{a,b,e\}$ by \emph{ex post} efficiency. In turn, weak strategyproofness rules out that $f(R^2)\in \{\{e\}, \{b,e\}, \{a,b,e\}\}$ as voter $8$ can otherwise manipulate by deviating from $R^1$ to $R^2$, and that $f(R^2)\in \{\{a\}, \{a,b\}\}$ as voter $8$ can otherwise manipulate by deviating from $R^2$ to $R^1$. Thus, the only valid outcome is $\{b\}$, i.e., $f(R^2)=\{b\}$. However, this contradicts strong Condorcet-consistency as $b$ is not the Condorcet winner in $R^2$, so the assumption that $f(R^1)=\{a,b,c\}$ is wrong.
\smallskip

\noindent
{\medmuskip=0mu\relax
	\thickmuskip=1mu\relax
    \begin{tabular}{L{0.07\profilewidth} L{0.4\profilewidth} L{0.4\profilewidth}}
	$R^2$: & $\{1,2,3\}$: $b,c,a,e,d$ & $\{4,5,6\}$: $a,e,d,b,c$ \\
 & $7:$ $e,d,b,c,a$ & $8$: $e,d,b,c,a$
\end{tabular}}
\medskip

\textit{Case 2.2: $f(R^1)\neq \{b,c,e\}$.$\qquad$} 
For our second case, we suppose for contradiction that $f(R^1)=\{b,c,e\}$. We consider the profiles $R^3$, $R^4$, and $R^5$ derived from $R^1$ by letting voters $6$, $5$, and $4$ swap $a$ and $e$. Using analogous reasoning as in Case 1.1, it can be shown that $f(R^4)=\{b,c,e\}$. However, this means that voter $6$ can manipulate by deviating from $R^4$ to $R^5$ as $e$ is the Condorcet winner in $R^5$. Thus, strong Condorcet-consistency postulates that $f(R^5)=\{e\}$, which means that voter $4$ can manipulate by deviating from $R^5$ to $R^6$. Hence, weak strategyproofness conflicts with Condorcet-consistency, thus proving that $f(R^1)\neq\{b,c,e\}$. 
\smallskip

\noindent
{\medmuskip=0mu\relax
	\thickmuskip=1mu\relax
\begin{tabular}{L{0.07\profilewidth} L{0.4\profilewidth} L{0.4\profilewidth}}
	$R^3$: & $\{1,2,3\}$: $b,c,a,e,d$ & $\{4,5\}$: $a,e,d,b,c$
 \end{tabular}
 \begin{tabular}{L{0.07\profilewidth} L{0.28\profilewidth} L{0.28\profilewidth}  L{0.28\profilewidth}}
      & $6$: $e,a,d,b,c$ & $7:$ $e,d,b,c,a$ & $8$: $e,d,c,b,a$
 \end{tabular}}\smallskip

{\noindent\medmuskip=0mu\relax
	\thickmuskip=1mu\relax
\begin{tabular}{L{0.07\profilewidth} L{0.4\profilewidth} L{0.4\profilewidth}}
	$R^4$: & $\{1,2,3\}$: $b,c,a,e,d$ & $4$: $a,e,d,b,c$
 \end{tabular}
 \begin{tabular}{L{0.07\profilewidth} L{0.33\profilewidth} L{0.25\profilewidth}  L{0.25\profilewidth}}
      & $\{5,6\}$: $e,a,d,b,c$ & $7:$ $e,d,b,c,a$ & $8$: $e,d,c,b,a$
 \end{tabular}}\smallskip

{\noindent\medmuskip=0mu\relax
	\thickmuskip=1mu\relax
\begin{tabular}{L{0.07\profilewidth} L{0.4\profilewidth} L{0.4\profilewidth}}
	$R^5$: & $\{1,2,3\}$: $b,c,a,e,d$ & $\{4,5,6\}$: $e,a,d,b,c$\\
  & $7:$ $e,d,b,c,a$ & $8$: $e,d,c,b,a$
 \end{tabular}}\medskip

\textit{Case 2.3: $f(R^1)\neq \{a,c,e\}$.$\qquad$} 
For our third case, we assume for contradiction that $f(R^1)= \{a,c,e\}$. In this case, we derive the profiles $R^6$ and $R^7$ from $R^1$ by letting voters $1$ and $2$ reinforce $a$ against $c$. We first note for $R^6$ that $|f(R^6)|\neq 1$ due to strong Condorcet-consistency and the absence of a Condorcet winner. Moreover, $f(R^5)\neq\{a,c\}$ as voter $1$ can otherwise manipulate $f$ by deviating from $R^1$ to $R^6$. Since $a\sim_M^6 c$ is the only majority tie in $R^6$, it follows from \Cref{lem:no2} that $|f(R^6)|\neq 2$, too. Hence, $|f(R^6)|\geq 3$ and, because voter $1$ prefers each set $X$ with $b\in X$, $d\not\in X$, and $|X|\geq 3$ to $\{a,c,e\}$, we conclude that $f(R^6)=\{a,c,e\}$. Furthermore, analogous arguments between $R^6$ and $R^7$ show that $f(R^7)=\{a,c,e\}$.

Next, we consider the profile $R^8$ and $R^9$ derived from $R^7$ by letting voters $1$ and $2$ also reinforce $a$ against $b$. First, we note again that $|f(R^8)|\neq 1$ as there is still no Condorcet winner in this profile. Moreover, $f(R^8)\neq \{a,b\}$ because voter $1$ can otherwise manipulate by deviating from $R^7$ to $R^8$. In turn, \Cref{lem:no2} implies that $|f(R^8)|\neq 2$ as $a\sim_M^8 b$ is the only majority tie in $R^8$. This means that $|f(R^8)|\geq 3$ and it can be verified that every such outcome with $b\in f(R^8)$ allows voter $1$ to manipulate by going from $R^7$ to $R^8$. Hence, we conclude that $f(R^8)=\{a,c,e\}$. However, this means that voter $2$ can manipulate by deviating from $R^8$ to $R^9$: in the latter profile $a$ is the Condorcet winner, so $f(R^9)=\{a\}$ due to Condorcet-consistency. However, voter $2$ prefers $a$ to both $c$ and $e$, so deviating from $R^8$ to $R^9$ makes him better off. This proves that $f(R^1)\neq\{a,c,e\}$.
\smallskip

{\noindent\medmuskip=0mu\relax
	\thickmuskip=1mu\relax
\begin{tabular}{L{0.07\profilewidth} L{0.36\profilewidth} L{0.4\profilewidth}}
	$R^6$: & $1$: $b,a,c,e,d$ & $\{2,3\}$: $b,c,a,e,d$ 
 \end{tabular}
 \begin{tabular}{L{0.07\profilewidth} L{0.36\profilewidth} L{0.25\profilewidth}  L{0.25\profilewidth}}
      & $\{4,5,6\}$: $e,a,d,b,c$ & $7:$ $e,d,b,c,a$ & $8$: $e,d,c,b,a$
 \end{tabular}}\smallskip

 {\noindent\medmuskip=0mu\relax
	\thickmuskip=1mu\relax
\begin{tabular}{L{0.07\profilewidth} L{0.36\profilewidth} L{0.4\profilewidth}}
	$R^7$: & $\{1,2\}$: $b,a,c,e,d$ & $3$: $b,c,a,e,d$ 
 \end{tabular}
 \begin{tabular}{L{0.07\profilewidth} L{0.36\profilewidth} L{0.25\profilewidth}  L{0.25\profilewidth}}
      & $\{4,5,6\}$: $e,a,d,b,c$ & $7:$ $e,d,b,c,a$ & $8$: $e,d,c,b,a$
 \end{tabular}}\smallskip

  {\noindent\medmuskip=0mu\relax
	\thickmuskip=1mu\relax
 \begin{tabular}{L{0.07\profilewidth} L{0.36\profilewidth} L{0.25\profilewidth}  L{0.25\profilewidth}}
	$R^8$: & $1$: $a,b,c,e,d$ & $2$: $b,a,c,e,d$ & $3$: $b,c,a,e,d$\\
      & $\{4,5,6\}$: $e,a,d,b,c$ & $7:$ $e,d,b,c,a$ & $8$: $e,d,c,b,a$
 \end{tabular}}\smallskip

 {\noindent\medmuskip=0mu\relax
	\thickmuskip=1mu\relax
\begin{tabular}{L{0.07\profilewidth} L{0.36\profilewidth} L{0.4\profilewidth}}
	$R^7$: & $\{1,2\}$: $a,b,c,e,d$ & $3$: $b,c,a,e,d$ 
 \end{tabular}
 \begin{tabular}{L{0.07\profilewidth} L{0.36\profilewidth} L{0.25\profilewidth}  L{0.25\profilewidth}}
      & $\{4,5,6\}$: $e,a,d,b,c$ & $7:$ $e,d,b,c,a$ & $8$: $e,d,c,b,a$
 \end{tabular}}
\medskip

\textit{Case 2.4: $f(R^1)\neq \{a,b,e\}$.\qquad} 
As the last case, suppose for contradiction that $f(R^1)=\{a,b,e\}$. We first consider the profile $R^{10}$ that is derived from $R^{1}$ by swapping $a$ and $b$ in the preference relation of voter $8$. For this profile, no singleton set is chosen as there is no Condorcet winner. Moreover, $f(R^{10})\neq \{a,b\}$ because voter $8$ can otherwise manipulate by going from $R^{10}$ to $R^1$. We then derive from \Cref{lem:no2} that $|f(R^{10})|\neq 2$ as $a\sim_M^{10} b$ is the only majority tie in $R^{10}$. Consequently, $|f(R^{10})|\geq 3$. Now, if $f(R^{10})=\{a,c,e\}$ or $f(R^{10})=\{b,c,e\}$, then voter $8$ can manipulate by deviating from $R^{1}$ to $R^{10}$. By contrast, if $f(R^{10})=\{a,b,c\}$, then voter $8$ can manipulate by going into the other direction. Hence, the only valid choice set of size $3$ is $f(R^{10})=\{a,b,e\}$. Finally, we note that $f(R^{10})\neq \{a,b,c,e\}$ as voter $8$ can otherwise manipulate to the profile $\hat R^1$ for which $f(\hat R^1)=\{a,b,d,e\}$. Hence, the only valid outcome is $f(R^{10})=\{a,b,e\}$. 

Next, we let voter $7$ change his preference relation to $e,d,a,b,c$ to derive the profile $R^{11}$. We note that there is still no Condorcet winner in $R^{11}$, so $|f(R^{11})|\neq 1$ due to strong Condorcet-consistency. Furthermore, $f(R^{11})\neq \{a,c\}$ as voter $5$ can otherwise manipulate by deviating to $f(R^{10})$. Since $a\sim_M^{11} c$ is the only majority tie in $R^{11}$, this means that $|f(R^{11})|\neq 2$ due to \Cref{lem:no2}. Finally, it follows that $f(R^{11})\neq X$ for every Pareto-optimal set $X$ with $|X|\geq 3$ and $c\in X$ because voter $7$ can otherwise manipulate by deviating back to $R^{11}$. Hence, the only valid outcome for $R^{11}$ is $f(R^{11})=\{a,b,e\}$.

Finally, we consider the profile $R^{12}$ derived from $R^{11}$ by assigning voter $1$ the preference relation $b,e,a,c,d$. In this profile $a$ is the Condorcet winner, so strong Condorcet-consistency requires that $f(R^{12})=\{a\}$. However, this means that voter $1$ can manipulate by reverting back to $R^{11}$ as he prefers the set $\{a,b,e\}$ to $\{a\}$ according to his preferences in $R^{12}$. Consequently, we now infer that $f(R^{1})\neq\{a,b,e\}$.
\smallskip

\noindent
{\medmuskip=0mu\relax
	\thickmuskip=1mu\relax
    \begin{tabular}{L{0.08\profilewidth} L{0.4\profilewidth} L{0.4\profilewidth}}
	$R^{10}$: & $\{1,2,3\}$: $b,c,a,e,d$ & $\{4,5,6\}$: $a,e,d,b,c$ \\
 & $7:$ $e,d,b,c,a$ & $8$: $e,d,c,a,b$
\end{tabular}}\smallskip

\noindent
{\medmuskip=0mu\relax
	\thickmuskip=1mu\relax
    \begin{tabular}{L{0.08\profilewidth} L{0.4\profilewidth} L{0.4\profilewidth}}
	$R^{11}$: & $\{1,2,3\}$: $b,c,a,e,d$ & $\{4,5,6\}$: $a,e,d,b,c$ \\
 & $7:$ $e,d,a,b,c$ & $8$: $e,d,c,a,b$
\end{tabular}}\smallskip

{\noindent\medmuskip=0mu\relax
	\thickmuskip=1mu\relax
\begin{tabular}{L{0.07\profilewidth} L{0.36\profilewidth} L{0.4\profilewidth}}
	$R^6$: & $1$: $b,e,a,c,d$ & $\{2,3\}$: $b,c,a,e,d$ 
 \end{tabular}
 \begin{tabular}{L{0.07\profilewidth} L{0.36\profilewidth} L{0.25\profilewidth}  L{0.25\profilewidth}}
      & $\{4,5,6\}$: $e,a,d,b,c$ & $7:$ $e,d,a,b,c$ & $8$: $e,d,c,a,b$
 \end{tabular}}\smallskip

This completes the proof of this proposition.
\end{proof}

\subsection{Proof of \Cref{thm:impSDS}}

   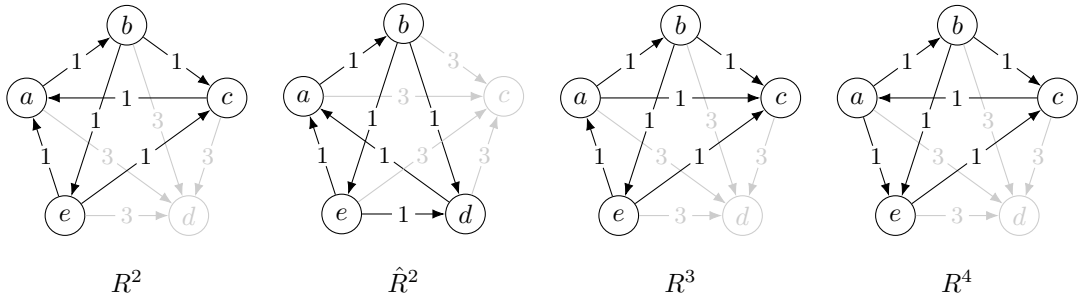
\begin{figure*}[tb]
    \newcommand\dist{1.4}
        \centering
        \begin{tikzpicture}
        \tikzstyle{mynode}=[fill=white,circle,draw,minimum size=1.5em,inner sep=0pt]
	   \tikzstyle{mylabel}=[fill=white,circle,inner sep=0.5pt]
	  \node[mynode] (a) at (162:\dist) {$a$};
	  \node[mynode] (b) at (90: \dist) {$b$};
	  \node[mynode] (c) at (18: \dist) {$c$};
	  \node[mynode, draw=black!20] (d) at (306:\dist) {\color{black!20}$d$};
	  \node[mynode] (e) at (234:\dist) {$e$};
   \node at (0, -2) {$R^2$};
	  	  
	  \draw[-Latex] (a) edge node[mylabel] {\small $1$} (b);
	  \draw[-Latex,black!20] (a) edge node[mylabel] {\small\color{black!20} $3$} (d);
	  \draw[-Latex] (b) edge node[mylabel] {\small $1$} (c);
	  \draw[-Latex,black!20] (b) edge node[mylabel] {\small\color{black!20} $3$} (d);
	  \draw[-Latex] (b) edge node[mylabel] {\small $1$} (e);
	  \draw[-Latex,black!20] (c) edge node[mylabel] {\small\color{black!20} $3$} (d);
      \draw[-Latex] (c) edge node[mylabel] {\small $1$} (a);
	  \draw[-Latex] (e) edge node[mylabel] {\small $1$} (a);
	  \draw[-Latex] (e) edge node[mylabel] {\small $1$} (c);
	  \draw[-Latex,black!20] (e) edge node[mylabel] {\small\color{black!20} $3$} (d);
        \end{tikzpicture}
        \hspace{0.3cm}
        \begin{tikzpicture}
        \tikzstyle{mynode}=[fill=white,circle,draw,minimum size=1.5em,inner sep=0pt]
	   \tikzstyle{mylabel}=[fill=white,circle,inner sep=0.5pt]
	  \node[mynode] (a) at (162:\dist) {$a$};
	  \node[mynode] (b) at (90: \dist) {$b$};
	  \node[mynode, draw=black!20] (c) at (18: \dist) {\color{black!20}$c$};
	  \node[mynode] (d) at (306:\dist) {$d$};
	  \node[mynode] (e) at (234:\dist) {$e$};
      \node at (0, -2) {$\hat R^2$};
	  	  
	  \draw[-Latex] (a) edge node[mylabel] {\small $1$} (b);
	  \draw[-Latex,black!20] (a) edge node[mylabel] {\small\color{black!20} $3$} (c);
	  \draw[-Latex] (b) edge node[mylabel] {\small $1$} (d);
	  \draw[-Latex,black!20] (b) edge node[mylabel] {\small\color{black!20} $3$} (c);
	  \draw[-Latex] (b) edge node[mylabel] {\small $1$} (e);
	  \draw[-Latex,black!20] (d) edge node[mylabel] {\small\color{black!20} $3$} (c);
      \draw[-Latex] (d) edge node[mylabel] {\small $1$} (a);
	  \draw[-Latex] (e) edge node[mylabel] {\small $1$} (a);
	  \draw[-Latex] (e) edge node[mylabel] {\small $1$} (d);
	  \draw[-Latex,black!20] (e) edge node[mylabel] {\small\color{black!20} $3$} (c);
        \end{tikzpicture}
        \hspace{0.3cm}
        \begin{tikzpicture}
        \tikzstyle{mynode}=[fill=white,circle,draw,minimum size=1.5em,inner sep=0pt]
	   \tikzstyle{mylabel}=[fill=white,circle,inner sep=0.5pt]
	  \node[mynode] (a) at (162:\dist) {$a$};
	  \node[mynode] (b) at (90: \dist) {$b$};
	  \node[mynode] (c) at (18: \dist) {$c$};
	  \node[mynode, draw=black!20] (d) at (306:\dist) {\color{black!20}$d$};
	  \node[mynode] (e) at (234:\dist) {$e$};
   \node at (0, -2) {$R^3$};
	  	  
	  \draw[-Latex] (a) edge node[mylabel] {\small $1$} (b);
	  \draw[-Latex,black!20] (a) edge node[mylabel] {\small\color{black!20} $3$} (d);
	  \draw[-Latex] (b) edge node[mylabel] {\small $1$} (c);
	  \draw[-Latex,black!20] (b) edge node[mylabel] {\small\color{black!20} $3$} (d);
	  \draw[-Latex] (b) edge node[mylabel] {\small $1$} (e);
	  \draw[-Latex,black!20] (c) edge node[mylabel] {\small\color{black!20} $3$} (d);
      \draw[-Latex] (a) edge node[mylabel] {\small $1$} (c);
	  \draw[-Latex] (e) edge node[mylabel] {\small $1$} (a);
	  \draw[-Latex] (e) edge node[mylabel] {\small $1$} (c);
	  \draw[-Latex,black!20] (e) edge node[mylabel] {\small\color{black!20} $3$} (d);
        \end{tikzpicture}
        \hspace{0.3cm}
        \begin{tikzpicture}
        \tikzstyle{mynode}=[fill=white,circle,draw,minimum size=1.5em,inner sep=0pt]
	   \tikzstyle{mylabel}=[fill=white,circle,inner sep=0.5pt]
	  \node[mynode] (a) at (162:\dist) {$a$};
	  \node[mynode] (b) at (90: \dist) {$b$};
	  \node[mynode] (c) at (18: \dist) {$c$};
	  \node[mynode, draw=black!20] (d) at (306:\dist) {\color{black!20}$d$};
	  \node[mynode] (e) at (234:\dist) {$e$};
   \node at (0, -2) {$R^4$};
	  	  
	  \draw[-Latex] (a) edge node[mylabel] {\small $1$} (b);
	  \draw[-Latex,black!20] (a) edge node[mylabel] {\small\color{black!20} $3$} (d);
	  \draw[-Latex] (b) edge node[mylabel] {\small $1$} (c);
	  \draw[-Latex,black!20] (b) edge node[mylabel] {\small\color{black!20} $3$} (d);
	  \draw[-Latex] (b) edge node[mylabel] {\small $1$} (e);
	  \draw[-Latex,black!20] (c) edge node[mylabel] {\small\color{black!20} $3$} (d);
      \draw[-Latex] (c) edge node[mylabel] {\small $1$} (a);
	  \draw[-Latex] (a) edge node[mylabel] {\small $1$} (e);
	  \draw[-Latex] (e) edge node[mylabel] {\small $1$} (c);
	  \draw[-Latex,black!20] (e) edge node[mylabel] {\small\color{black!20} $3$} (d);
        \end{tikzpicture}
        \caption{Weighted majority relations for the profiles $R^2$, $\hat R^2$, $R^3$, and $R^4$ in the proof of \Cref{thm:impSDS}. An arrow from $x$ to $y$ with weight $w$ means that $w$ more voters prefer $x$ to $y$ than $y$ to $x$. Pareto-dominated alternatives and their corresponding edges are colored in gray to improve readability.}
        \label{fig:maj1}
    \end{figure*} 

We next turn to the proof of \Cref{thm:impSDS}. To this end, we recall that we consider for this result the domain $\mathcal{L}^*=\bigcup \{\mathcal{L}^N\colon N\subseteq\mathbb{N} \text{ is finite and non-empty}\}$ containing all strict preference profiles, regardless of the number of voters. As first step to prove \Cref{thm:impSDS}, we show that every pairwise and weakly strategyproof SDS on $\mathcal{L}^*$ satisfies a property known as set-monotonicity \citep{Bran11c}: if a voter weakens an alternative that is assigned probability $0$, then the outcome does not change. To formally state this result, we recall that $R^{i:yx}$ is the profile derived from $R$ by only reinforcing $y$ against $x$ in the preference relation of voter $i$. 

\begin{lemma}\label{lem:SMON}
    Let $f$ denote a weakly strategyproof and pairwise SDS on $\mathcal{L}^*$. It holds that $f(R)=f(R^{i:yx})$ for all profiles $R$ with $f(R,x)=0$. 
\end{lemma}
\begin{proof}
    Let $f$ denote a weakly strategyproof and pairwise SDS and suppose for contradiction that there is a profile $R$ such that $f(R)\neq f(R^{i:yx})$ even though $f(R,x)=0$. In this case, we define $U^+=\{z\in A\setminus \{x,y\}\colon f(R^{i:yx},z)\geq f(R,z)\}$ and $U^-=\{z\in A\setminus \{x,y\}\colon f(R^{i:yx},z)<f(R,z)\}$. Now, consider the profile $R'$ derived from $R$ by adding two voters $i^*$, $j^*$ whose preference relations satisfy that \emph{(i)} $u\succ_i^* x\succ_i^* y\succ_i^* v$ for all $u\in U^+$, $v\in U^-$ and \emph{(ii)} $u\succ_i^* v$ if and only if $v\succ_j^* u$ for all $u,v\in A$. Since the preference relations of voter $i^*$ and $j^*$ are completely inverse, pairwiseness requires that $f(R')=f(R)$. Next, let $\hat R$ denote the profile derived from $R'$ by reinforcing $y$ against $x$ in the preference relation of voter $i^*$. It is easy to see that $|\{i\in N_{\hat R}\colon u\succ_i v\}|-|\{i\in N_{\hat R}\colon v\succ_i u\}|=|\{i\in N_{R^{i:yx}}\colon u\succ_i v\}|-|\{i\in N_{R^{i:yx}}\colon v\succ_i u\}|$ for all $u,v\in A$, so pairwiseness implies that $f(\hat R)=f(R^{i:yx})$. Finally, we claim that voter $i^*$ prefers $f(\hat R)$ to $f(R')$, which means that he can manipulate by deviating from $\hat R$ to $R'$. 
    To see this, we note for every $z\in U^+\cup \{x\}$ that $f(\hat R, U(\succsim_i^*, z))\geq f(R', U(\succsim_i^*, z))$ as the alternatives in $U^+$ only gain probability and $x$ cannot lose probability as $f(R',x)=f(R,x)=0$. 
    On the other hand, if $z\in U^-\cup \{y\}$, we define $L=\{w\in A\colon z\succ_{i^*}, w\}$ and observe that $f(\hat R, U(\succsim_i^*,z))=1-f(\hat R, L)\geq 1-f(R', L)=f(R', U(\succsim_i^*,z))$ as $L\subseteq U^-$. Finally, since $f(\hat R)\neq f(R')$, one of these inequalities must be strict, which proves that voter $i^*$ can indeed manipulate by deviating from $R'$ to $\hat R$. Hence, our initial assumption is wrong and $f(R)= f(R^{i:yx})$.
\end{proof}

We will now prove \Cref{thm:impSDS}.

\impSDS*
\begin{proof}
    In this proof, we will focus on the case that $m=5$ and $n=3$. To generalize the result to more alternatives, we can simply append the extra alternatives at the bottom of the preference relations of all voters as \emph{ex post} efficiency then implies that these alternatives are chosen with probability $0$ and they hence do not affect our results. To add more voters, we can just add voters with completely inverse preferences. These voters do not change the majority margins, so pairwiseness requires that the outcome is not allowed to change.

    Now, assume for contradiction that there is a neutral and pairwise SDS $f$ on $\mathcal{L}^*$ for $m=5$ alternatives that satisfies strategyproofness and \emph{ex post} efficiency. We will 
     consider the following six profiles. \smallskip

    \setlength\tabcolsep{3 pt}
    \noindent\begin{profile}{L{0.07\profilewidth} L{0.27\profilewidth} L{0.27\profilewidth} L{0.27\profilewidth}}
        $R^1$: & $1$: $a,b,e,d,c$ & $2$: $b,c,e,d,a$ & $3$: $e,d,c,a,b$\\
        $R^2$: & $1$: $a,b,e,c,d$ & $2$: $b,c,e,a,d$ & $3$: $e,c,a,b,d$\\
        $\hat R^1$: & $1$: $a,b,e,d,c$ & $2$: $b,c,e,d,a$ & $3$: $d,e,a,b,c$\\
        $\hat R^2$: & $1$: $a,b,e,d,c$ & $2$: $b,e,d,a,c$ & $3$: $d,e,a,b,c$\\
        $R^3$: & $1$: $a,b,e,c,d$ & $2$: $b,c,e,a,d$ & $3$: $e,a,c,b,d$\\
        $R^4$: & $1$: $a,b,e,c,d$ & $2$: $b,c,a,e,d$ & $3$: $e,c,a,b,d$
    \end{profile}

    We first note that $d$ is Pareto-dominated by $e$ in $R^1$ and $R^2$, and $c$ is Pareto-dominated by $b$ in $\hat R^1$ and $\hat R^2$. Hence, \emph{ex post} efficiency shows that $f(R^1,d)=f(R^2,d)=0$ and $f(\hat R^1,c)=f(\hat R^2,c)=0$. Moreover, \Cref{lem:SMON} implies that $f(R^1)=f(R^2)$ and $f(\hat R^1)=f(\hat R^2)$. Next, we consider the majority margins of $R^2$ and $\hat R^2$, which are depicted in \Cref{fig:maj1}. In particular, we observe that $|\{i\in N_{R^2}\colon x\succ_i y\}|-|\{i\in N_{R^2}\colon y\succ_i x\}|=|\{i\in N_{\hat R^2}\colon \tau(x)\succ_i \tau(y)\}|-|\{i\in N_{\hat R^2}\colon \tau(y)\succ_i \tau(x)\}|$ for all $x,y\in A$, where $\tau$ is the permutation defined by $\tau(a)=a$, $\tau(b)=b$, $\tau(c)=d$, $\tau(d)=c$, and $\tau(e)=e$. Hence, neutrality and pairwiseness imply that $f(R^2,x)=f(\hat R^2,x)$ for $x\in \{a,b,e\}$, $f(R^2,c)=f(\hat R^2,d)$, and $f(R^2,d)=f(\hat R^2,c)=0$. Combined with our previous observation, we now infer that $f(R^1,x)=f(\hat R^1,x)$ for $x\in \{a,b,c\}$ and $f(R^1,c)=f(\hat R^1,d)$. Hence, if $f(R^1,c)>0$, voter $3$ can manipulate by deviating from $R^1$ to $R^2$ as he prefers $d$ to $c$. 

    We thus suppose from now on that $f(R^1,c)=f(R^2,c)=0$ and consider the profile $R^3$ which is derived from $R^2$ by letting voter $3$ swap $a$ and $c$. Since $f(R^2,c)=0$, \Cref{lem:SMON} implies that $f(R^2)=f(R^3)$, so $f(R^3,c)=f(R^3,d)=0$. Next, $a$, $b$, and $e$ are symmetric in the weighted majority relation of $R^3$ (see \Cref{fig:maj1}), so neutrality and pairwiseness require that $f(R^3,a)=f(R^3,b)=f(R^3,e)$. Thus, we conclude that $f(R^2,x)=f(R^3,x)=\frac{1}{3}$ for all $x\in \{a,b,e\}$. 
    
    As last point, consider the profile $R^4$ that arises from $R^2$ by letting voter $2$ swap $a$ and $e$. The corresponding weighted majority relation is depicted in $R^4$ and it can be checked that this relation can be derived from the one of $R^2$ by permuting the alternatives according to the permutation $\tau$ defined by $\tau(a)=b$, $\tau(b)=e$, $\tau(c)=a$, $\tau(d)=d$, and $\tau(e)=c$. By neutrality from $R^2$ to $R^4$, it hence follows that $f(R^4,a)=f(R^4,b)=f(R^4,c)=\frac{1}{3}$. However, this means that voter $2$ can manipulate by deviating from $R^2$ to $R^4$ as he prefers $c$ to $e$. This contradicts the properties of $f$, so the assumption that an SDS satisfies all axioms of this theorem is wrong. 
\end{proof}

\subsection{Proof of \Cref{thm:SDimpossibity}}

As the fourth point, we present here our simplified proof of \Cref{thm:SDimpossibity}. In particular, our proof only reasons about 13 profiles (note here also that $R^1$ can be turned into $R^7$ and $R^8$ by permuting alternatives, and $R^2$ can be turned into $R^{13}$ by permuting alternatives and voters, so our proofs only needs 10 ``canonical'' profiles). By contrast, the computer-generated proof by \citet{BBEG16a} needs 47 canonical profiles and relates them in an convoluted way. 

    \begin{figure*}[t]
    \centering
    \begin{profile}{ll@{\hskip 2em}l@{\hskip 2em}l@{\hskip 2em}l}
        $R^1$: & $1\colon \{b,c\}, \{a,d\}$ & $2\colon \{a,d\}, \{b,c\}$ & $3\colon \{a,b\}, c,d$ & $4\colon \{c,d\},a,b$\\
        $R^2$: & $1\colon \{b,c\}, \{a,d\}$ & $2\colon \{a,d\}, \{b,c\}$ & $3\colon \{a,b\}, c,d$ & \textcolor{red}{$4\colon c,\{a,d\},b$}\\
        $R^3$: & $1\colon \{b,c\}, \{a,d\}$ & $2\colon \{a,d\}, \{b,c\}$ & \textcolor{red}{$3\colon b,a, \{c,d\}$} & $4\colon c,\{a,d\},b$\\
        $R^4$: & $1\colon \{b,c\}, \{a,d\}$ & $2\colon \{a,d\}, \{b,c\}$ & $3\colon b,a, \{c,d\}$ & \textcolor{red}{$4\colon c,d,\{a,b\}$}\\
        $R^5$: & \textcolor{red}{$1\colon \{b,c\}, d,a$} & $2\colon \{a,d\}, \{b,c\}$ & $3\colon b,a, \{c,d\}$ & $4\colon c,d,\{a,b\}$\\
        $R^6$: & $1\colon \{b,c\}, d,a$ & \textcolor{red}{$2\colon d, \{a,b\}, c$} & $3\colon b,a, \{c,d\}$ & $4\colon c,d,\{a,b\}$\\\hline
        $R^7$: & $1\colon \{a,b\}, \{c,d\}$ & $2\colon \{c,d\}, \{a,b\}$ & $3\colon \{a,c\}, b,d$ & $4\colon \{b,d\},c,a$\\
        $R^8$: & $1\colon \{a,b\}, \{c,d\}$ & $2\colon \{c,d\}, \{a,b\}$ & $3\colon \{a,d\}, b,c$ & $4\colon \{b,c\},d,a$\\
        $R^9$: & $1\colon \{a,b\}, \{c,d\}$ & $2\colon \{c,d\}, \{a,b\}$ &\textcolor{red}{$3\colon \{b,c\}, d,a$} & \textcolor{red}{$4\colon \{b,d\}, c,a$}\\
        $R^{10}$: & $1\colon \{a,b\}, \{c,d\}$ & \textcolor{red}{$2\colon c,d,\{a,b\}$} & $3\colon \{b,c\}, d,a$ & $4\colon \{b,d\}, c,a$\\
        $R^{11}$: & \textcolor{red}{$1\colon b,a, \{c,d\}$} & $2\colon c,d,\{a,b\}$ & $3\colon \{b,c\}, d,a$ & $4\colon \{b,d\}, c,a$\\\hline
        $R^{12}$: & $1\colon \{b,c\}, d,a$ & $2\colon d, \{a,b\}, c$ & \textcolor{red}{$3\colon \{a,b\}, \{c,d\}$} & $4\colon c,d,\{a,b\}$\\
        $R^{13}$: & $1\colon \{b,c\}, d,a$ & $2\colon d, \{a,b\}, c$ & $3\colon \{a,b\}, \{c,d\}$ & \textcolor{red}{$4\colon \{c,d\},\{a,b\}$}
    \end{profile}
    \caption{Profiles used in the proof of \Cref{thm:SDimpossibity}. Preference relations highlighted in red indicate manipulations and horizontal lines indicate the three steps of the proof.}
    \label{fig:profilesSDimpossibility}
    \end{figure*}

\setcounter{theorem}{4}
\begin{theorem}[\citet{BBEG16a}]
    Assume $n\geq 4$ and $m\geq 4$. There is no anonymous and neutral SDS on $\mathcal{R}^N$ that satisfies \emph{ex ante} efficiency and weak strategyproofness. 
\end{theorem}\begin{proof}
    We focus on the case that $m=n=4$; to extend the result to larger numbers of voters or alternatives, we can use standard inductive arguments that add completely indifferent voters and universally bottom-ranked alternatives. Hence, suppose for contradiction that there is an anonymous and neutral SDS that satisfies weak strategyproofness and \emph{ex ante} efficiency when $m=n=4$. To derive a contradiction, we consider the preference profiles in \Cref{fig:profilesSDimpossibility}. Moreover, we will proceed in three steps and first show that $f(R^6,b)+f(R^6,c)=1$, then that $f(R^6,b)=1$, and finally infer the contradiction.\medskip

    \textbf{Step 1:} We will first show that $f(R^6,b)+f(R^6,c)=1$. To this end, we first consider the profile $R^1$, where $a$ is symmetric to $c$ and $b$ is symmetric to $d$. Hence, anonymity and neutrality require that $f(R^1,a)=f(R^1,c)$ and $f(R^1, b)=f(R^1,d)$. Moreover, every lottery $p$ with $p(a)=p(c)$ and $p(b)=p(d)>0$ is \emph{ex ante} dominated by the lottery $q$ with $q(a)=q(c)=\frac{1}{2}$. Consequently, $f(R^1,a)=f(R^1,c)=\frac{1}{2}$. 
    
    Next, let $R^2$ denote the profile derived from $R^1$ by replacing the preference relation of voter $4$ with $c,\{a,d\},b$. It holds that $f(R^2,c)\geq \frac{1}{2}$ as otherwise, voter $4$ can manipulate by deviating from $R^2$ to $R^1$, so $f(R^2,a)+f(R^2,b)\leq \frac{1}{2}$. 
    Moreover, $d$ is Pareto-dominated by $a$ in $R^2$, so $f(R^2,d)=0$. 
    
    The profile $R^3$ is derived from $R^2$ by changing the preference relation of voter $3$ to $b,a,\{c,d\}$. Since $a$ still Pareto-dominates $d$, it follows that $f(R^3,d)=0$. In turn, weak strategyproofness shows that $f(R^3,a)+f(R^3,b)\leq \frac{1}{2}$ as voter $3$ can otherwise manipulate by deviating from $R^2$ to $R^3$. This implies that $f(R^3,c)\geq \frac{1}{2}$.
    
    As the fourth step, we analyze the profile $R^4$ which arises from $R^3$ by assigning voter $4$ the preference relation $c,d,\{a,b\}$. First, we note that $b$ is symmetric to $c$ and $a$ is symmetric to $d$ in $R^4$, so $f(R^4,b)=f(R^4,c)$ and $f(R^4,a)=f(R^4,d)$. Next, if $f(R^4,c)<\frac{1}{2}$, then voter $4$ can manipulate by deviating from $R^4$ to $R^3$ since $f(R^3,c)> f(R^4,c)$ and $f(R^3,c)+f(R^3,d)\geq\frac{1}{2}=f(R^4,c)+f(R^4,d)$ (as $f(R^4,b)=f(R^4,c)$ and $f(R^4,a)=f(R^4,d)$). By our symmetry conditions, it follows that $f(R^4,b)=f(R^4,c)=\frac{1}{2}$.
    
    The profile $R^5$ arises from $R^4$ by letting voter $1$ change his preference relation to $\{b,c\},d,a$. Based on weak strategyproofness to $R^4$, it is easy to infer that $f(R^5,b)+f(R^5,c)=1$ as voter $1$ could otherwise manipulate to $R^4$. 
    
    Finally, the profile $R^6$ is derived from $R^5$ by letting voter $2$ manipulate to $d, \{a,b\}, c$. Since $f(R^5,b)+f(R^5,c)=1$ and $b$ and $c$ are voter $2$'s least preferred alternatives in $R^5$, weak strategyproofness requires that $f(R^6,b)+f(R^6,c)=1$.\medskip

    \textbf{Step 2:} We will next show that $f(R^6,b)=1$. To this end, consider the profiles $R^7$ and $R^8$. It can be checked that these profiles are symmetric to $R^1$: $R^7$ arises from $R^1$ by mapping $a$ to $c$, $b$ to $a$, $c$ to $b$, and $d$ to $d$; $R^8$ arises from $R^1$ by mapping $a$ to $d$, $b$ to $a$, $c$ to $b$, and $d$ to $c$. Hence, symmetric arguments as for $R^1$ show that $f(R^7,b)=f(R^7,c)=f(R^8,b)=f(R^8,d)=\frac{1}{2}$.
    
    The profile $R^9$ arises from $R^7$ when replacing the preference relation of voter $3$ with $\{b,c\}, d,a$. Since $f(R^7,b)+f(R^7,c)=1$, strategyproofness hence requires that $f(R^9, b)+f(R^9,c)=1$, too, because otherwise voter $3$ can manipulate by deviating from $R^9$ to $R^7$. Moreover, we can derive the profile $R^9$ from $R^8$ by assigning the third voter in this profile the preference relation $\{b,d\}, c,a$ (and swapping the third and fourth voter). Since $f(R^8,b)+f(R^8,d)=1$, we infer that $f(R^9, b)+f(R^9,d)=1$ because of strategyproofness. Combining these two equations implies that $f(R^9,b)=1$. 
    
    The profile $R^{10}$ is derived from $R^9$ by letting voter $2$ change his preference relation to $c,d,\{a,b\}$. First, we note that $f(R^9,a)=0$ as $b$ Pareto-dominates $a$. As a consequence, strategyproofness requires that $f(R^{10},b)=1$ as every other outcome constitutes a manipulation for voter $2$.
    
    Next, the profile $R^{11}$ follows from $R^{10}$ by assigning voter $1$ the preference relation $b,a,\{c,d\}$. It is easy to verify that $f(R^{11},b)=1$; otherwise, voter $1$ can manipulate by deviating from $R^{11}$ to $R^{10}$. 
    
    Finally, the profile $R^6$ arises now from $R^{11}$ by assigning voter $4$ the preference relation $d, \{a,b\},c$ (and reordering the voters). By Step 1, we know that $f(R^6, b)+f(R^6,c)=1$. Hence, if $f(R^6,c)>0$, voter $4$ could manipulate by deviating from $R^6$ to $R^{11}$ as he prefers $b$ to $c$, so we conclude now that $f(R^6,b)=1$.\medskip

    \textbf{Step 3:} Finally, we will derive a contradiction. To this end, we consider the profile $R^{12}$ which is derived from $R^6$ by assigning voter $3$ the preference relation $\{b,a\}, \{c,d\}$. Alternative $b$ Pareto-dominates $a$ in $R^{12}$, so $f(R^{12}, a)=0$. In turn, weak strategyproofness shows that $f(R^{12},b)=1$ as voter $3$ can manipulate by deviating back to $R^6$ otherwise. 
    
    Finally, the profile $R^{13}$ arises from $R^{12}$ by letting voter $4$ change his preference relation to $\{c,d\}, \{a,b\}$. Since $b$ still Pareto-dominates $a$, we can infer analogously to the last step that $f(R^{13},b)=1$. However, voter $2$ can now manipulate to the profile $R^8$ by reporting $\{a,d\},b,c$ (and reordering the voters). Since $f(R^8,b)=f(R^8,d)=\frac{1}{2}$ and voter $2$ prefers $d$ to $b$, this constitutes a manipulation, which contradicts our assumptions on $f$.
\end{proof}

\subsection{Proof of \Cref{thm:bidicatorship}}\label{app:bidictatorship}

Finally, we will discuss the proof of \Cref{thm:bidicatorship} in detail. To this end, we will interpret even-chance SDSs again as set-valued voting rules; see \Cref{subsec:topsonly} for details. 

For the proof of \Cref{thm:bidicatorship}, we need additional terminology. Following \citet{BBL21b}, we say that a group of voters $G\subseteq N$ is \emph{decisive} for an SDS $f$ if $f(R)\subseteq T_i(R)$ for all voters $i\in G$ and all preference profiles $R$ such that ${\succsim_i}={\succsim_j}$ for all $i,j\in G$. Similarly, a group of voters $G$ is \emph{nominating} for an SDS $f$ if $f(R)\cap T_i(R)\neq\emptyset$ for all voters $i\in G$ and profiles $R$ such that ${\succsim_i}={\succsim_j}$ for all $i,j\in G$. Finally, a voter $i$ is a \emph{weak dictator} for $f$ if $\{i\}$ is a nominating group for $f$. Less formally, a group of voters is decisive if it can enforce that a subset of their top-ranked alternatives is chosen when all voters in the group report the same preference relation, and it is nominating if it can enforce that at least one of their top-ranked alternatives is chosen. Based on this notation, \citet{BBL21b} have shown the following result.\footnote{In fact, \citet{BBL21b} show an even stronger result as they use a strategyproofness notion that is weaker than our weak strategyproofness.}

\begin{lemma}[\citet{BBL21b}]\label{lem:decisiveVSnominating}
    Assume that $m\geq 3$ and $n\geq 2$, and let $f$ denote a weakly strategyproof and \emph{ex post} efficient even-chance SDS. A group of voters $G$ with $\emptyset\subsetneq G\subsetneq N$ is decisive for $f$ if and only if $N\setminus G$ is not nominating for $f$. 
\end{lemma}

\Cref{lem:decisiveVSnominating} is of interest to us as the notion of bidictatorial SDSs is closely related to dictating groups of voters. We thus analyze the structure of the decisive groups in the next lemma.

\begin{lemma}\label{lem:contraction}
    Assume that $m\geq 3$ and $n\geq 2$, and let $f$ denote a weakly strategyproof and \emph{ex post} efficient even-chance SDS. Moreover, suppose there are two decisive groups $I$, $J$ for $f$ such that $1<|I|=|J|<n$ and $|I\setminus J|=|J\setminus I|=1$. The group $I\cap J$ is also decisive for $f$. 
\end{lemma}
\begin{proof}
Let $f$ denote an weakly strategyproof and \emph{ex post} efficient even-chance SDS and let $I$ and $J$ denote two decisive groups as given by the lemma. We moreover define $i$ as the voter in $I\setminus J$, $j$ as the voter in $J\setminus I$ and $H$ as $N\setminus (I\cup J)$. In the subsequent proof, we will focus on the case that $m=3$ as it is straightforward to generalize the lemma to larger values of $m$ by adding dummy alternatives that are universally bottom-ranked. These dummy alternatives will be Pareto-dominated and therefore do not affect our analysis. 

We will prove the lemma in three steps and start with the central part. To this end, consider the following six profiles.\smallskip

\setlength{\tabcolsep}{3pt}
\noindent\begin{profile}{L{0.07\profilewidth} L{0.2\profilewidth} L{0.25\profilewidth} L{0.17\profilewidth} L{0.25\profilewidth}}
$R^1$: & $i$: $b,a,c$ & $I\cap J$: $c, \{a,b\}$ & $j$: $a,c,b$ & $H$: $a,b,c$\\
$R^2$: & $i$: $\{a,b\},c$ & $I\cap J$: $c, \{a,b\}$ & $j$: $a,c,b$ & $H$: $a,b,c$\\
$R^3$: & $i$: $\{a,b\},c$ & $I\cap J$: $c, \{a,b\}$ & $j$: $a,b,c$ & $H$: $a,b,c$\\
$R^4$: & $i$: $b,c,a$ & $I\cap J$: $c, \{a,b\}$ & $j$: $a,b,c$ & $H$: $a,b,c$\\
$R^5$: & $i$: $b,a,c$ & $I\cap J$: $c, \{a,b\}$ & $j$: $a,b,c$ & $H$: $a,b,c$\\
$R^6$: & $i$: $a,b,c$ & $I\cap J$: $c, \{a,b\}$ & $j$: $a,b,c$ & $H$: $a,b,c$
\end{profile}

For the start of our analysis, we will assume that $f(R^1)=\{a,c\}$ and $f(R^4)=\{b,c\}$; we will show later why this holds. Now, consider the profile $R^2$ which is derived from $R^1$ by letting voter $i$ manipulate. All voters prefer $a$ to $b$ in $R^2$, so $b$ is Pareto-dominated. Consequently, $f(R^2)\subseteq \{a,c\}$ due to \emph{ex post} efficiency. If $f(R^2)=\{a\}$, voter $i$ can manipulate by deviating from $R^1$ to $R^2$. Conversely, if $f(R^2)=\{c\}$, voter $i$ can manipulate by deviating from $R^2$ to $R^1$. Hence, the only valid outcome for $f(R^2)=\{a,c\}$. Next, the profile $R^3$ is derived from $R^2$ by letting voter $j$ swap $b$ and $c$. Since $a$ still Pareto-dominates $b$, we have $f(R^3)\subseteq \{a,c\}$. Moreover, a similar analysis as for $R^2$ shows that $f(R^3)=\{a,c\}$.

Next, we will determine the outcome for $R^5$. To this end, we first note that voter $i$ can deviate from $R^3$ to $R^5$ by breaking the tie between $a$ and $b$. Since $f(R^3)=\{a,c\}$, strategyproofness from $R^3$ to $R^5$ requires that $c\in f(R^5)$ and $f(R^5)\neq \{a,b,c\}$. This leaves us with three possible outcomes: $f(R^5)=\{a,c\}$, $f(R^5)=\{b,c\}$, or $f(R^5)=\{c\}$. Now, if $f(R^5)=\{b,c\}$, then voter $j$ can manipulate to $R^1$ as $f(R^1)=\{a,c\}$ and voter $j$ prefers $a$ to $b$. On the other hand, if $f(R^5)=\{a,c\}$, then voter $i$ can manipulate by deviating to $R^4$ as $f(R^4)=\{b,c\}$ by assumption. Hence, $f(R^5)=\{c\}$.

Finally, consider the profile $R^6$ that arises from $R^5$ by letting voter $i$ swap $a$ and $b$. Since $f(R^5)=\{c\}$, $f(R^6)=\{c\}$ as any other outcome constitutes a manipulation for voter $i$. Finally, in $R^6$, all voters in $N\setminus (I\cap J)$ have the same preference relation but $a\not\in f(R^5)$. This shows that the group $N\setminus (I\cap J)$ is not nominating, so $I\cap J$ is decisive by \Cref{lem:decisiveVSnominating}. 

We next prove our claims for $R^1$ and $R^4$. \medskip

\textbf{Claim 1: $f(R^1)=\{a,c\}$}

Tor prove this claim, we first consider the profile $\hat R^1$ defined below.\smallskip

\noindent\begin{profile}{L{0.07\profilewidth} L{0.17\profilewidth} L{0.25\profilewidth} L{0.17\profilewidth} L{0.25\profilewidth}}
$\hat R^1$: & $i$: $a,b,c$ & $I\cap J$: $\{b,c\}, a$ & $j$: $c,a,b$ & $H$: $a,b,c$
\end{profile}

We will show that $f(\hat R^1)=\{c\}$ or $f(\hat R^1)=\{b,c\}$. To this end, we first note that if $f(\hat R^1)=\{b\}$, then voter $j$ can deviate to the preference relation $a,b,c$ and strategyproofness requires for the resulting profile $\hat R^{1,1}$ that $f(\hat R^{1,1})=\{b\}$. However, all voters in $N\setminus (I\cap J)$ report $a,b,c$ in $\hat R^{1,1}$ and $a\not\in f(R^2)$, so the set $N\setminus (I\cap J)$ is not nominating. By \Cref{lem:decisiveVSnominating}, it then follows that the group $I\cap J$ is decisive and the lemma would follow at this point. Hence, we suppose that $f(\hat R^1)\neq \{b\}$. Moreover, suppose for contradiction that $a\in f(\hat R^1)$. Then, the voters $k\in I\cap J$ can one after another deviate to the preference relation $c,a,b$. By strategyproofness, it follows for every step that, if $a$ is chosen before the manipulation, it must be chosen after the manipulation; otherwise, we would choose a subset of the voter's most preferred alternative after the manipulation but not before, which contradicts weak strategyproofness. Since $a\in f(\hat R^1)$, this means for the profile $\hat R^{1,2}$, where all voters $k\in I\cap J$ report $c,a,b$, that $a\in f(\hat R^{1,2})$. However, in this profile, all voters $k\in J$ report $c,a,b$, so the decisiveness of $J$ implies that $f(\hat R^{1,2})=\{c\}$. This is a contradiction, so $a\not\in f(\hat R^1)$ and $f(\hat R^1)=\{c\}$ or $f(\hat R^1)=\{b,c\}$.

Next, let $\hat R^2$ and $\hat R^3$ denote the two subsequent profiles.\smallskip

\noindent\begin{profile}{L{0.07\profilewidth} L{0.17\profilewidth} L{0.25\profilewidth} L{0.17\profilewidth} L{0.25\profilewidth}}
$\hat R^2$: & $i$: $b,a,c$ & $I\cap J$: $\{b,c\}, a$ & $j$: $c,a,b$ & $H$: $a,b,c$\\
$\hat R^3$: & $i$: $b,a,c$ & $I\cap J$: $\{b,c\}, a$ & $j$: $a,c,b$ & $H$: $a,b,c$
\end{profile}

First, it follows that $a\not\in f(\hat R^2)$ and $a\not\in f(\hat R^3)$. Otherwise, we can let the voters in $I\cap J$ one after another deviate to $b,a,c$. For each step, strategyproofness implies that $a$ stays chosen, but all voters in $I$ report $b,a,c$ in the final profile. Hence the decisiveness of the voters in $I$ requires that $b$ is uniquely chosen, which contradicts that $a\in f(\hat R^2)$ (resp. $a\in f(\hat R^3)$), so $f(\hat R^2)\subseteq \{b,c\}$ and $f(\hat R^3)\subseteq \{b,c\}$. Moreover, $f(\hat R^2)\neq \{b\}$ as voter $i$ can otherwise manipulate by deviating from $\hat R^1$ to $\hat R^2$. This means that $f(\hat R^2)=\{b,c\}$ or $f(\hat R^2)=\{c\}$. From this, we infer that $f(\hat R^3)\neq \{b\}$ as otherwise voter $j$ can manipulate by deviating from $\hat R^3$ to $\hat R^2$. Hence, $f(\hat R^3)=\{c\}$ or $f(\hat R^3)=\{b,c\}$.

Next, we suppose that $I\cap J=\{i_1, \dots, i_\ell\}$ and consider the following profiles $\hat R^{3,k}$ for $k\in \{0,\dots, \ell\}$. In particular, $\hat R^{3,k}$ arises from $\hat R^{3,{k-1}}$ by assigning voter $i_k$ the preference relation $c, \{a,b\}$.\smallskip

\noindent
{\medmuskip=0mu\relax
	\thickmuskip=1mu\relax
\begin{tabular}{L{0.10\profilewidth} L{0.3\profilewidth} L{0.3\profilewidth} L{0.3\textwidth}}
	$\hat R^{3,k}$: & $i$: $b,a,c$ & $j$: $a,c,b$ & $H$: $a,b,c$
\end{tabular}
\begin{tabular}{L{0.10\profilewidth} L{0.4\profilewidth} L{0.45\profilewidth}}
	& $\{i_1,\dots, i_k\}$: $c, \{a,b\}$ &  $\{i_{k+1}, \dots, i_\ell\}$: $\{b,c\},a$
\end{tabular}}\smallskip

Now, if $f(\hat R^{3, k-1})=\{c\}$, then weak strategyproofness implies that $f(\hat R^{3, k})=\{c\}$. On the other hand, if $f(\hat R^{3, k-1})=\{b,c\}$, weak strategyproofness requires that $c\in f(\hat R^{3, k})$ and that $f(\hat R^{3, k})\neq \{a,b,c\}$ as voter $i_k$ can otherwise manipulate by deviating to $\hat R^{3, k-1}$. Hence, $f(\hat R^{3, k})=\{c\}$, $f(\hat R^{3, k})=\{b,c\}$, or $f(\hat R^{3, k})=\{a,c\}$ in this case. Finally, if $f(\hat R^{3, k-1})=\{a,c\}$, then $a\in f(\hat R^{3,k})$ as voter $i_k$ can otherwise manipulate by deviating from $\hat R^{3, k-1}$ to $\hat R^k$. Moreover, $c\in f(\hat R^{3,k})$ and $f(\hat R^{3,k})\neq \{a,b,c\}$ because voter $i_k$ can otherwise manipulate by deviating from $\hat R^{3,k}$ to $\hat R^{3,k-1}$. Consequently, $f(\hat R^{3,k})=\{a,c\}$ in this case. Since $f(\hat R^{3,0})=\{b,c\}$ or $f(\hat R^{3,0})=\{c\}$ (as $\hat R^{3,0}=\hat R^3$), this means for the profile $\hat R^4=\hat R^{3,\ell}$ that $f(\hat R^4)\in \{\{c\}, \{b,c\}, \{a,c\}\}$.\smallskip

\begin{profile}{L{0.07\profilewidth} L{0.17\profilewidth} L{0.25\profilewidth} L{0.17\profilewidth} L{0.25\profilewidth}}
    $\hat R^{4}$: & $i$: $b,a,c$ & $I\cap J$: $c, \{a,b\}$ &  $j$: $a,c,b$ & $H$: $a,b,c$
\end{profile}

Now, if $f(\hat R^4)=\{b,c\}$, then voter $j$ can manipulate by deviating to $c, \{a,b\}$. In the resulting profile $\hat R^{4,1}$, all voters in $J$ report $c,\{a,b\}$, so $c$ must be uniquely chosen due to the decisiveness of $J$. So, $f(\hat R^4)\neq \{b,c\}$. Next, if $f(\hat R^4)=\{c\}$, then voter $i$ can manipulate to $a,b,c$ and strategyproofness requires for the resulting profile $\hat R^{4,2}$ that $f(\hat R^{4,2})=\{c\}$ as $c$ is his least-preferred alternative in $\hat R^4$. Moreover, let $\hat R^{4,3}$ denote the preference profile derived from $\hat R^{4,2}$ by also assigning voter $j$ the preference relation $a,b,c$. All voters in $N\setminus (I\cap J)$ report $a,b,c$ in this profile, so $b$ is now Pareto-dominated. Moreover, since voter $j$ prefers $a$ to $c$, the outcomes $f(\hat R^{4,3})=\{a,c\}$ and $f(\hat R^{4,3})=\{a\}$ constitute manipulations for him. This implies that $f(\hat R^{4,3})=\{c\}$. However, this means that the group $N\setminus (I\cap J)$ is not nominating, so the group $I\cap J$ is decisive and the lemma follows again. Hence, we suppose that $f(\hat R^4)=\{a,c\}$. Finally, we note that $R^1=\hat R^4$, so this proves our first auxiliary claim.\medskip

\textbf{Claim 2: $f(R^4)=\{b,c\}$}

We note that the proof of this claim is very similar to the last one, so we will keep the explanations short. First, consider the profile $\bar R^1$ shown below.\smallskip 

\begin{profile}{L{0.07\profilewidth} L{0.17\profilewidth} L{0.25\profilewidth} L{0.17\profilewidth} L{0.25\profilewidth}}
$\bar R^1$: & $i$: $c,b,a$ & $I\cap J$: $\{a,c\}, b$ & $j$: $b,a,c$ & $H$: $b,a,c$
\end{profile}

We note that the profile $\bar R^1$ is symmetric to the profile $\hat R^1$ from the last claim: we only need to exchange voters $i$ and $j$ and alternatives $a$ and $b$ to transform $\bar R^1$ to $\hat R^1$. Consequently, we can use symmetric arguments to infer that $f(\bar R^1)\subseteq \{a,c\}$ and that $I\cap J$ is decisive if $f(\bar R^1)=\{a\}$. Hence, we focus on the case that $f(\bar R^1)\in \{\{c\}, \{a, c\}\}$. 

Next, let $H=\{j_1,\dots, j_{\ell'}\}$ and consider the profiles $\bar R^{1,k}$ which are defined by $\bar R^{1,0}=\bar R^1$ and $\bar R^{1,k}$ is derived from $\bar R^{1,k-1}$ by assigning voter $j_k$ the preference relation $a,b,c$.\smallskip  

\noindent
{\medmuskip=0mu\relax
	\thickmuskip=1mu\relax
\begin{tabular}{L{0.10\profilewidth} L{0.3\profilewidth} L{0.3\profilewidth} L{0.3\textwidth}}
	$\bar R^{1,k}$: & $i$: $c,b,a$ & $I\cap J$: $\{a,c\}, b$ & $j$: $b,a,c$
\end{tabular}
\begin{tabular}{L{0.10\profilewidth} L{0.4\profilewidth} L{0.45\profilewidth}}
	& $\{j_1,\dots, j_k\}$: $a,b,c$ & $\{j_{k+1},\dots, j_{\ell'}\}$: $b,a,c$ 
\end{tabular}}\smallskip

First, it holds for every profile $\bar R^{1,k}$ that $f(\bar R^{1,k})\subseteq \{a,c\}$. Otherwise, $b\in f(\bar R^{1,k})$ and the voters $I\cap J$ can one after another deviate to the preference relation $c,b,a$. For each step, strategyproofness requires that $b$ is chosen after the deviation if it was chosen before. Hence, even if all voters in $I$ report $c,b,a$, $b$ still must be chosen. However, this conflicts with the decisiveness of $I$, which postulates that $c$ is uniquely chosen if all voters in $I$ report $c,b,a$. This proves that $f(\bar R^{1,k})\subseteq \{a,c\}$. Next, we note that if $f(\bar R^{1, k-1})\neq \{a\}$, then $f(\bar R^{1,k})\neq \{a\}$. In more detail, by our previous insight, if $f(\bar R^{1,k-1})\neq \{a\}$, then $f(\bar R^{1,k-1})=\{c\}$ or $f(\bar R^{1,k-1})=\{a,c\}$. Since voter $j_k$ prefers $a$ to $c$ in $\hat R^{1,k-1}$, strategyproofness excludes that $f(\bar R^{1,k})= \{a\}$. By combining these insights and the fact that $f(\bar R^{1,0})=\{\{c\}, \{a,c\}\}$, it follows that $f(\bar R^{1,\ell'})\in \{\{c\}, \{a,c\}\}$.

Next, let $\bar R^2$ and $\bar R^3$ denote the following profiles.

\begin{profile}{L{0.07\profilewidth} L{0.17\profilewidth} L{0.25\profilewidth} L{0.17\profilewidth} L{0.25\profilewidth}}
$\bar R^2$: & $i$: $c,b,a$ & $I\cap J$: $\{a,c\}, b$ & $j$: $a,b,c$ & $H$: $a,b,c$\\
$\bar R^3$: & $i$: $b,c,a$ & $I\cap J$: $\{a,c\}, b$ & $j$: $a,b,c$ & $H$: $a,b,c$
\end{profile}

Using the decisiveness of $J$, it can be shown that $f(\bar R^2)\subseteq \{a,c\}$ and $f(\bar R^3)\subseteq \{a,c\}$. Moreover, strategyproofness from $\bar R^{1,\ell'}$ to $\bar R^2$ implies that $f(\bar R^2)\neq \{a\}$ as voter $j$ can otherwise manipulate by deviating from $\bar R^{1,\ell'}$ to $\bar R^2$. Hence, $f(\bar R^2)=\{c\}$ or $f(\bar R^2)=\{a,c\}$. In turn, this implies that $f(\bar R^3)\neq \{a\}$, too; otherwise, voter $i$ can manipulate by deviating from $\bar R^3$ to $\bar R^2$. This shows that $f(\bar R^3)\in \{\{c\}, \{a,c\}\}$. 

For the next step, we recall that $\ell=|I\cap J|$ and $I\cap J=\{i_1,\dots, i_\ell\}$. Moreover, let $\bar R^{3,k}$ denote profiles such that $\bar R^{3,0}=\bar R^3$ and $\bar R^{3,k}$ arises from $\bar R^{3,k-1}$ by assigning voter $i_k$ the preference relation $c, \{a,b\}$.\smallskip

\noindent
{\medmuskip=0mu\relax
	\thickmuskip=1mu\relax
\begin{tabular}{L{0.10\profilewidth} L{0.3\profilewidth} L{0.3\profilewidth} L{0.3\textwidth}}
	$\bar R^{3,k}$: & $i$: $b,c,a$ & $j$: $a,b,c$ & $H$: $a,b,c$
\end{tabular}
\begin{tabular}{L{0.10\profilewidth} L{0.4\profilewidth} L{0.45\profilewidth}}
	& $\{i_1,\dots, i_k\}$: $c, \{a,b\}$ &  $\{i_{k+1}, \dots, i_\ell\}$: $\{a,c\},b$
\end{tabular}}\smallskip

We investigate the relationship between $f(\bar R^{3,k-1})$ and $f(\bar R^{3,k})$. First, if $f(\bar R^{3,k-1})=\{c\}$, then $f(\bar R^{3,k})=\{c\}$, too, as voter $i_k$ can otherwise manipulate by deviating back to $f(\bar R^{3,k-1})$. Second, if $f(\bar R^{3,k-1})=\{a,c\}$, then $c\in f(\bar R^{3,k})$ and $f(\bar R^{3,k})\neq \{a,b,c\}$ as voter $i_k$ can otherwise manipulate back to $\bar R^{3,k-1}$. This proves that $f(\bar R^{3,k})\in \{\{a,c\}, \{b,c\}, \{c\}\}$ if $f(\bar R^{3,k-1})=\{a,c\}$. Finally, if $f(\bar R^{3,k-1})=\{b,c\}$, we can again conclude that $c\in f(\bar R^{3,k})$ and $f(\bar R^{3,k})\neq\{a,b,c\}$ by weak strategyproofness from $\bar R^{3,k}$ to $\bar R^{3,k-1}$. Moreover, weak strategyproofness in the other direction implies that $f(\bar R^{3,k})\not\subseteq \{b,c\}$. Hence, in this case, only $f(R^{3,k})=\{a,c\}$ is possible. In summary, these insights combined with the fact that $f(\bar R^{3,0})\in \{\{c\}, \{b,c\}\}$ imply for the profile $\bar R^4=\bar R^{3,\ell}$ that $f(\bar R^{4})\in \{\{c\}, \{b,c\}, \{a,c\}\}$.\smallskip

\begin{profile}{L{0.07\profilewidth} L{0.17\profilewidth} L{0.25\profilewidth} L{0.17\profilewidth} L{0.25\profilewidth}}
    $\bar R^{4}$: & $i$: $b,c,a$ & $I\cap J$: $c, \{a,b\}$ &  $j$: $a,b,c$ & $H$: $a,b,c$
\end{profile}

Now, if $f(\bar R^4)=\{a,c\}$, voter $i$ can manipulate by deviating to $c, \{a,b\}$. In the resulting profile, all voters of $I$ report $c, \{a,b\}$, so the decisiveness of $I$ implies that $\{c\}$ is chosen. Since voter $i$ prefers $\{c\}$ to $\{a,c\}$, this shows that $f(\bar R^4)\neq\{a,c\}$. Next, if $f(\bar R^4)=\{c\}$, then $I\cap J$ is decisive. This follows by letting the voter $H\cup \{j\}$ one after another deviate to $b,c,a$. Strategyproofness implies for every step that $c$ stays the unique winner. However, in the resulting profile all voters in $N\setminus (I\cap J)$ report $b,c,a$, but $b$ is not chosen. Consequently, the set $N\setminus (I\cap J)$ is not nominating, so $I\cap J$ must be decisive by \Cref{lem:decisiveVSnominating}. Hence, if $f(\bar R^4)=\{c\}$, the lemma follows. Finally, this means that $f(\bar R^4)=\{b,c\}$, which proves our claim as $\bar R^4=R^4$.
\end{proof}

Based on \Cref{lem:decisiveVSnominating,lem:contraction}, we next show that there are at least one and at most two weak dictators and that the set of weak dictators is decisive for $f$. 

\begin{lemma}\label{lem:oligarchy}
    Assume that $m\geq 3$ and $n\geq 2$, and let $f$ denote a weakly strategyproof and \emph{ex post} efficient even-chance SDS. Moreover, let $G$ denote the set of weak dictators of $f$. It holds that $1\leq |G|\leq 2$ and that $G$ is decisive for $f$. 
\end{lemma}
\begin{proof}
    Let $f$ denote an even-chance SDS that satisfies all axioms of the lemma and let $G$ denote the set of weak dictators of $f$. We will split up the lemma in three separate claims.\medskip

    \textbf{Claim 1: $G\neq\emptyset$}
    
    Assume for contradiction that $G=\emptyset$, i.e., that there is no weak dictator for $f$. This means that no singleton set is nominating for $f$, so \Cref{lem:decisiveVSnominating} implies that every set $G\subseteq N$ with $|G|=n-1$ is decisive. Using \Cref{lem:contraction}, it therefore follows that every set of size $n-2$ is also decisive. Moreover, we can repeat this argument to infer that every set of size $1$ is decisive. However, this is impossible as two disjoints sets of voters cannot be simultaneously decisive. Hence, $G\neq\emptyset$.\medskip

    \textbf{Claim 2: $|G|\leq 2$}

    Suppose for contradiction that $|G|\geq 3$, let $i,j,k\in G$ denote three different weak dictators of $f$, and let $N^-=N\setminus\{i,j,k\}$ be the set of remaining voters. Furthermore, consider the following two preference profiles; as usual, all additional alternatives are bottom-ranked by all voters.

    \noindent
    \setlength{\tabcolsep}{3pt}
    \begin{profile}{L{0.07\profilewidth} L{0.15\profilewidth} L{0.2\profilewidth} L{0.2\profilewidth} C{0.3\profilewidth}}
    $R^{1}$: & $i$: $a,b,c$ & $j$: $b,c,a$ & $k$: $c, \{a,b\}$ & $N^-$: $\{a,b,c\}$\\
    $R^{2}$: & $i$: $a,b,c$ & $j$: $\{a,b\},c$ & $k$: $c, \{a,b\}$ & $N^-$: $\{a,b,c\}$
\end{profile}

    First, only $a$, $b$, and $c$ can be chosen in $R^1$ and $R^2$ due to \emph{ex post} efficiency. Next, since voters $i$, $j$, and $k$ are weak dictators for $f$, it follows that $f(R^1)=\{a,b,c\}$. Similarly, $f(R^2)=\{a,c\}$ as voter $i$ top-ranks $a$, voter $k$ top-ranks $c$, and $a$ Pareto-dominates $b$ in $R^2$. However, this means that voter $j$ can manipulate by deviating from $R^2$ to $R^1$ as he prefers $\{a,b,c\}$ to $\{a,c\}$. This is the desired contradiction, so $|G|\leq 2$.\medskip

    \textbf{Claim 3: $G$ is decisive for $f$}

    Finally, we show that $G$ is decisive for $f$. To this end, we note that the voters in $N\setminus G$ are no weak dictators for $f$, which means that the sets $\{i\}$ for $i\in N\setminus G$ are not nominating. By \Cref{lem:decisiveVSnominating}, it thus follows that the sets $N\setminus \{i\}$ are decisive for all $i\in N\setminus G$. Now, if $G=N\setminus \{i\}$ for some $i\in N$, this shows that $G$ is decisive. On the other hand, if $|G|\leq n-2$, \Cref{lem:decisiveVSnominating} shows that every set $H^1$ with $G\subseteq H^1$ and $|H^1|=n-1$ that $H^1$ is decisive. Now, using \Cref{lem:contraction}, it also follows for all sets $H^2$ with $G\subseteq H^2$ and $|H^2|=n-2$ that $H^2$ is decisive. In more detail, we can choose two arbitrary voters $i,j\in N\setminus H^2$ to infer that the sets $H^2\cup \{i\}$ and $H^2\cup \{j\}$ are decisive. \Cref{lem:contraction} then implies that $H^2$ is decisive. Moreover, it is easy to see that one can continue this type of reasoning to show that every set $H^k$ with $G\subseteq H^k$ is decisive. In particular, this shows that $G$ is decisive, too. 
\end{proof}

Finally, we are ready to prove \Cref{thm:bidicatorship}. In particular, we will show that only alternatives from the set of weak dictators can be chosen by a weakly strategyproof and \emph{ex post} efficient even-chance SDS.

\bidictatorship*
\begin{proof}
Let $f$ denote an even chance SDS that satisfies weak strategyproofness and \emph{ex post} efficiency, and let $G$ denote the set of weak dictators of $f$. By \Cref{lem:oligarchy}, the set $G$ has size $1$ or $2$ and is decisive. Now, if $G=\{i\}$ for some voter $i$, this means that $i$ is a dictator for $f$, so $f$ is dictatorial in this case. Hence, suppose that $G=\{i,j\}$ for two distinct voters $i,j$. 
We will show that $f(R)\subseteq T_i(R)\cup T_j(R)$ for all profiles $R$. For this, we consider multiple cases that depend on the relation between $T_i(R)$ and $T_j(R)$.\medskip

\textbf{Case 1: $T_j(R)\cap T_i(R)=\emptyset$}

For the first case, we consider an arbitrary profile $R$ such that $T_i(R)\cap T_j(R)=\emptyset$. In this case, we aim to show that $|f(R)\cap T_i(R)|\geq |f(R)\setminus T_i(R)|$. Because a symmetric argument also shows that $|f(R)\cap T_j(R)|\geq |f(R)\setminus T_j(R)|$, we can the deduce that $f(R)\subseteq T_i(R)\cup T_j(R)$. In more detail, if there is an alternative $x\in f(R)\setminus (T_i(R)\cup T_j(R))$, our inequalities implies that 
\begin{align*}
    |f(R)\cap T_i(R)|&\geq |f(R)\setminus T_i(R)|\\
    &>|f(R)\cap T_j(R)|\\
    &\geq |f(R)\setminus T_j(R)|\\
    &>|f(R)\cap T_i(R)|. 
\end{align*} 

This is a contradiction, so proving that $|f(R)\cap T_i(R)|\geq |f(R)\setminus T_i(R)|$ implies that $f(R)\subseteq T_i(R)\cup T_j(R)$. 

To prove that $|f(R)\cap T_i(R)|\geq |f(R)\setminus T_i(R)|$, we assume for contradiction that $|f(R)\cap T_i(R)|< |f(R)\setminus T_i(R)|$. In particular, we note that this inequality is true if and only if $\frac{|f(R)\setminus T_i(R)|}{|f(R)|}>\frac{1}{2}$. Next, we enumerate the voters $N\setminus \{i\}$ by $j_1,\dots, j_{n-1}$ (note that the weak dictator $j$ is among these voters). Moreover, we define the sequence of profiles $R^0,\dots, R^{n-1}$ by $R^0=R$ and $R^k$ differs from $R^{k-1}$ by assigning voter $j_k$ the preference relation $A\setminus T_i(R), T_i(R)$. By weak strategyproofness, it is easy to infer that $\frac{|f(R^k)\setminus T_i(R)|}{|f(R^k)|}\geq \frac{|f(R^{k-1})\setminus T_i(R)|}{|f(R^{k-1})|}$ because voter $j_k$ otherwise can manipulate by deviating from $R^k$ to $R^{k-1}$. Since $\frac{|f(R^0)\setminus T_i(R)|}{|f(R^0)|}>\frac{1}{2}$, it follows that $\frac{|f(R^{n-1})\setminus T_i(R)|}{|f(R^{n-1})|}>\frac{1}{2}$, too. Equivalently, this means that $|f(R^{n-1})\setminus T_i(R)|>|f(R^{n-1})\cap T_i(R)|$. 

For the last step, let $X$ denote the set of alternatives that voter $i$ ranks second. By \emph{ex post} efficiency, it follows that $f(R^{n-1})\subseteq X\cup T_i(R)$ because $x\succ_i y$ and $x\sim_k y$ for all $k\in N\setminus \{i\}$ and alternatives $x\in X$, $y\in A\setminus (X\cup T_i(R))$. Finally, let $a$ denote an arbitrary alternative in $T_i(R)$ and $b$ denote an arbitrary alternative in $X$, and consider the profile $R'$ derived from $R^{n-1}$ by assigning voter $i$ a preference relation where $a$ is his uniquely most preferred alternative and $b$ his uniquely second-most preferred alternative. By \emph{ex post} efficiency, it follows that $f(R')\subseteq \{a,b\}$ because $a$ Pareto-dominates every other alternative in $T_i(R)$ and $b$ Pareto-dominates every other alternative in $A\setminus T_i(R)$. Moreover, since voters $i$ and $j$ are weak dictators for $f$, we can conclude that $f(R')=\{a,b\}$. However, this means that voter $i$ can manipulate by deviating from $R^{n-1}$ to $R$ because $\frac{|f(R')\cap T_i(R)|}{|f(R')|}=\frac{1}{2}>\frac{|f(R^{n-1})\cap T_i(R)|}{|f(R^{n-1})|}$ and $\frac{|f(R')\cap (T_i(R)\cup X)|}{|f(R')|}=1=\frac{|f(R^{n-1})\cap (T_i(R)\cup X)|}{|f(R^{n-1})|}$. This contradicts the strategyproofness of $f$, so the initial assumption that $|f(R)\cap T_i(R)|< |f(R)\setminus T_i(R)|$ is wrong.\medskip

\textbf{Case 2: $T_i(R)\subseteq T_j(R)$ or $T_j(R)\subseteq T_i(R)$}

For the second case, consider an arbitrary profile $R$ such that $T_i(R)\subseteq T_j(R)$ (we note that the case that $T_j(R)\subseteq T_i(R)$ is symmetric). Moreover, assume for contradiction that there is an alternative $x\in f(R)\setminus T_j(R)$. In this case, let $R'$ denote the profile derived from $R$ by assigning voter $j$ the same preference relation of voter $i$. By the decisiveness of $G=\{i,j\}$ for $f$, it follows that $f(R')\subseteq T_i(R)\subseteq T_j(R)$. However, this means that voter $j$ can manipulate by deviating from $R$ to $R'$ because a subset of voter $j$' favorite alternatives is chosen for $R'$ but not for $R$. Hence, the initial assumption that $f(R)\not\subseteq T_j(R)$ is wrong, which proves for this case that $f$ is bidictatorial.\medskip

\textbf{Case 3: $T_i(R)\setminus T_j(R)\neq\emptyset$, $T_j(R)\setminus T_i(R)\neq\emptyset$, and $T_i(R)\cap T_j(R)\neq\emptyset$}

For the last case, we consider a profile $R$ such that $T_i(R)\setminus T_j(R)\neq\emptyset$, $T_j(R)\setminus T_i(R)\neq\emptyset$, and $T_i(R)\cap T_j(R)\neq\emptyset$. In this case, let $R'$ denote a profile such that (i) voter $i$ top-ranks the alternatives in $T_i(R)\cap T_j(R)$, (ii) voter $j$ reports the same preference relation as in $R$, and (iii) all voters $j\in N\setminus \{i,j\}$ bottom-rank $T_i(R)\cup T_j(R)$. We note for the profile $R'$ that $T_i(R')\subseteq T_j(R')$, so Case 2 requires that $f(R')\subseteq T_j(R)$. Moreover, the alternatives $x\in T_i(R)\cap T_j(R)$ Pareto-dominate the alternatives $y\in (T_i(R)\cup T_j(R))\setminus (T_i(R)\cap T_j(R))$, so $f(R')\subseteq T_i(R)\cap T_j(R)$. Next, let $\hat R$ denote the profile derived from $R'$ by assigning voter $i$ the same preference relation as in $R$. Strategyproofness requires that $f(\hat R)\subseteq T_i(R)$ as otherwise, voter $i$ can manipulate from $\hat R$ to $R'$. Finally, we can now let the voters $k\in N\setminus \{i,j\}$ one after another deviate to the preference relation that they report in $R$. Since all these voters bottom-rank $T_i(R)\cup T_j(R)$, weak strategyproofness requires for each step that a subset of $T_i(R)\cup T_j(R)$ is chosen after the manipulation if it was before the manipulation. Since this sequence results in the profile $R$, we derive that $f(R)\subseteq T_i(R)\cup T_j(R)$.
\end{proof}

\end{document}